\documentclass[11pt]{article}
\usepackage{graphicx} %
\usepackage{fullpage}

\usepackage[T1]{fontenc}

\usepackage{amsmath,amsfonts,amssymb}
\usepackage{mathtools}
\usepackage{amsthm} %

\usepackage{tikz}
\usetikzlibrary{calc,arrows.meta,positioning}

\definecolor{NLSDBlue}{HTML}{123B5D}
\definecolor{NLSDGold}{HTML}{D8A528}
\definecolor{NLSDGray}{HTML}{65727C}

\tikzset{
  nlsdarrow/.style={-{Latex[length=2mm]},thick,draw=NLSDGray},
  nlsdstate/.style={draw=NLSDBlue,thick,rounded corners=2pt,fill=white,
    inner xsep=7pt,inner ysep=5pt,align=center},
  nlsdparty/.style={circle,draw=NLSDBlue,thick,fill=white,minimum size=8mm,
    font=\bfseries},
  nlsdchallenge/.style={draw=NLSDGold,thick,rounded corners=2pt,
    fill=NLSDGold!8,inner xsep=7pt,inner ysep=5pt,align=center},
  nlsdresult/.style={draw=NLSDBlue,thick,rounded corners=2pt,
    fill=NLSDBlue!5,inner xsep=7pt,inner ysep=5pt,align=center}
}

\usepackage{physics}

\usepackage{enumitem}
\usepackage{float} %

\usepackage{mdframed}
\usepackage[
  colorlinks=true,  citecolor=blue,  linkcolor=black,  urlcolor=black,
  pdftitle={Non-Local Search-to-Decision over Finite Fields},
  pdfauthor={}
]{hyperref}

\usepackage{longtable}
\usepackage{array}

\newtheorem{theorem}{Theorem}[section]

\newtheorem{definition}[theorem]{Definition}
\newtheorem{lemma}[theorem]{Lemma}
\newtheorem{claim}[theorem]{Claim}

\newtheorem{remark}[theorem]{Remark}
\newtheorem{corollary}[theorem]{Corollary}

\usepackage{comment}

\renewcommand{\bra}[1]{\langle#1\rvert}

\renewcommand{\ket}[1]{\lvert#1\rangle}

\newcommand{\cX}{{\mathcal X}}

\newcommand{\bfG}{\mathbf{G}}

\newcommand{\bfM}{\mathbf{M}}

\newcommand{\bfX}{\mathbf{X}}
\newcommand{\bfY}{\mathbf{Y}}
\newcommand{\bfZ}{\mathbf{Z}}

\newcommand{\regA}{{\color{gray} \mathbf{A}}}
\newcommand{\regB}{{\color{gray} \mathbf{B}}}
\newcommand{\regC}{{\color{gray} \mathbf{C}}}
\newcommand{\regD}{{\color{gray} \mathbf{D}}}
\newcommand{\regE}{{\color{gray} \mathbf{E}}}

\newcommand{\regR}{{\color{gray} \mathbf{R}}}
\newcommand{\regS}{{\color{gray} \mathbf{S}}}

\newcommand{\regW}{{\color{gray} \mathbf{W}}}
\newcommand{\regX}{{\color{gray} \mathbf{X}}}
\newcommand{\regY}{{\color{gray} \mathbf{Y}}}
\newcommand{\regZ}{{\color{gray} \mathbf{Z}}}

\newcommand{\linear}{{\cal L}}
\newcommand{\density}{{\cal D}}

\newcommand{\bob}{{\cal B}}
\newcommand{\charlie}{{\cal C}}
\newcommand{\expect}{\mathbb{E}}

\newcommand{\ftwo}{\mathbb{F}_2}

\newcommand{\srch}{{\sf srch}}
\newcommand{\pred}{{\sf pred}}
\newcommand{\predpovm}{{\cal P}}

\newcommand{\extractor}{{\cal E}}

\title{Non-Local Search-to-Decision over $\mathbb{F}_2$, and More\footnote{AI disclosure statement can be found on the last page of the paper.}}
\author{Prabhanjan Ananth\footnote{\texttt{prabhanjan@cs.ucsb.edu}} \\ \footnotesize{UCSB}}
\date{}

\newcommand{\HH}{\mathrm{HH}}
\newcommand{\HL}{\mathrm{HL}}
\newcommand{\LH}{\mathrm{LH}}
\newcommand{\LL}{\mathrm{LL}}

\begin{document}
\maketitle

\begin{abstract}
\noindent Non-local search-to-decision asks whether the difficulty of two non-communicating (non-local) parties both predicting $x$ given a bipartite state (that possibly depends on $x$) implies the difficulty of non-local parties both predicting the inner product $\langle r,x\rangle$, for a uniformly random $r$. This problem and its variants have been extensively studied and are motivated by applications to unclonable cryptographic primitives such as unclonable encryption and copy-protection. We study the identical-challenge setting, in which both parties receive the same uniformly random vector $r$, in contrast to the independently sampled challenges considered in prior works. We demonstrate positive results for the cases when $x$ and $r$ are vectors over $\mathbb{F}_2$ and over large finite fields. As a consequence, we obtain a conceptually different proof of indistinguishability-secure unclonable encryption. 
\end{abstract}

\clearpage
\tableofcontents
\clearpage

\section{Introduction}
\label{sec:intro}
\paragraph{Search-to-Decision Reductions in TCS.} Search-to-decision reductions use a procedure for a decision problem as an oracle to solve its associated search problem, and they have been studied in many areas of theoretical computer science. A classic example is the self-reduction for satisfiability: a satisfying assignment can be recovered using polynomially many queries to an oracle that decides whether a formula is satisfiable. In cryptography, the seminal Goldreich--Levin theorem~\cite{GL89} gives, in this broad sense, a prediction-to-reconstruction reduction and shows that every one-way function admits a hard-core predicate. Another example is Regev's reduction, which turns an efficient algorithm for the decisional learning with errors problem, in some parameter regimes, into an efficient algorithm for the search learning with errors problem~\cite{Reg05}.

\paragraph{Search-to-Decision Reductions in Unclonable Cryptography.} Unclonable cryptography is an area of quantum cryptography that leverages the no-cloning principle of quantum mechanics to demonstrate the existence of cryptographic tasks that are otherwise impossible to achieve using purely classical resources. Often in unclonable cryptography, one has to deal with (non-local) adversarial parties which cannot communicate but may share entanglement. Just as search-to-decision reductions have been immensely useful for classical cryptography, search-to-decision reductions have also proven to be essential for unclonable cryptography~\cite{CLLZ21,AKLLZ22,AKL23,AKY25}. However, entanglement presents a significant obstacle to search-to-decision reductions for non-local parties; in particular, entanglement presents technical challenges when applying a one-party reduction separately to both non-local parties. Before we move further, we first state the non-local search-to-decision problem considered in this work.

\paragraph{Non-Local Search-to-Decision.} Consider the following problem: a dealer encodes a uniformly random string
$x\in\mathbb{F}^n$, for a field $\mathbb{F}$, into a bipartite state and gives one part to Bob and
the other to Charlie.  The two non-communicating parties then receive the same
uniformly random vector $r\in\mathbb{F}^n$ and they must {\em simultaneously} output the inner product 
$\langle r,x\rangle$. The trivial strategy that Bob and Charlie can employ is to output a uniformly random element in $\mathbb{F}$. The non-local
search-to-decision question asks whether a noticeable improvement over this
trivial strategy can occur unless there also exist local measurements with which Bob and
Charlie both {\em simultaneously} recover the entire hidden string $x$ with noticeable probability.
\par Formally, let us first consider the case when $\mathbb{F}=\mathbb{F}_2$. Let $x \xleftarrow{\$} \mathbb{F}_2^n$ and let a quantum channel $\Phi$ encode the
computational basis state $\ketbra{x}{x}$ as the following bipartite state:
\[
 \rho_x:=\Phi(\ketbra{x}{x})
\]
on registers $(\regB,\regC)$ held by two non-communicating parties, Bob and Charlie. 

\begin{figure}[H]
\centering
\begin{tikzpicture}[node distance=0.9cm and 1.15cm]
  \node[nlsdchallenge] (sample) {$x\xleftarrow{\$}\ftwo^n$};
  \node[nlsdstate,right=of sample] (input) {$\ketbra{x}{x}$};
  \node[nlsdstate,right=of input,minimum width=1.15cm,minimum height=0.9cm]
    (channel) {$\Phi$};
  \node[nlsdparty,above right=0.48cm and 1.45cm of channel] (bob) {$\bob$};
  \node[nlsdparty,below right=0.48cm and 1.45cm of channel] (charlie) {$\charlie$};
  \draw[nlsdarrow] (sample)--(input);
  \draw[nlsdarrow] (input)--(channel);
  \draw[nlsdarrow] (channel)--node[above,sloped,font=\small]
    {register $\regB$} (bob);
  \draw[nlsdarrow] (channel)--node[below,sloped,font=\small]
    {register $\regC$} (charlie);
  \node[below=0.62cm of channel,text=NLSDGray,font=\small]
    {$\rho_x=\Phi(\ketbra{x}{x})$};
\end{tikzpicture}
\caption{The splitting stage.  A uniformly random string is
encoded by an arbitrary channel into the bipartite state $\rho_x$.  Bob and
Charlie receive disjoint registers and cannot communicate.}
\label{fig:nl-s2d-distribution}
\end{figure}
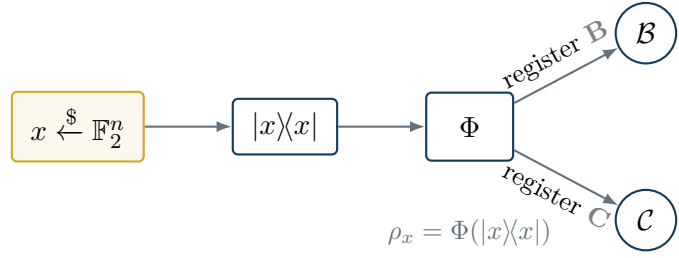

\noindent We consider two experiments: 
\begin{itemize}
    \item {\bf Search experiment}: in this experiment, Bob and Charlie succeed if they simultaneously output $x$. Denote the optimal success probability in this experiment by $p_{\srch}$. 
    \item {\bf Identical-challenge decision experiment}: in this experiment, both Bob and Charlie are additionally given $r \xleftarrow{\$} \mathbb{F}^n$. They win if they simultaneously predict $\langle r,x \rangle$. Denote the optimal success probability in this experiment by $p_{\pred}$.  
\end{itemize}
The non-local search-to-decision problem asks whether the following holds: if $p_{\srch}$ is negligible, is $p_{\pred}$ at most $\frac{1}{2} + \varepsilon$, where $\varepsilon$ is negligible?

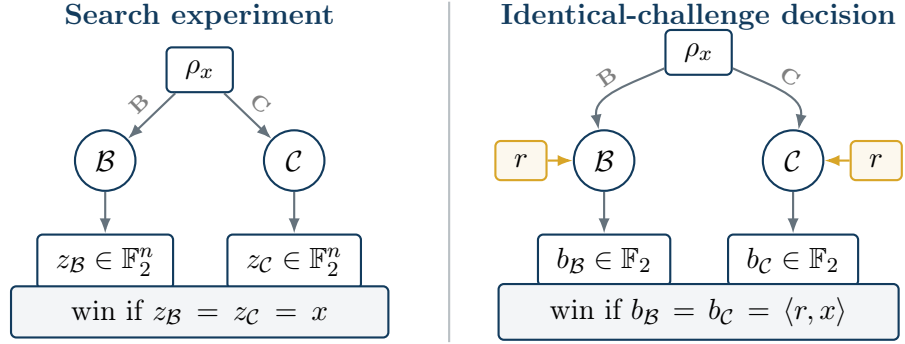
\begin{figure}[H]
\centering
\begin{tikzpicture}[x=1cm,y=1cm]
  \begin{scope}[shift={(-3.3,0)}]
    \node[font=\bfseries,text=NLSDBlue] at (0,2.75) {Search experiment};
    \node[nlsdstate] (searchstate) at (0,2.05) {$\rho_x$};
    \node[nlsdparty] (searchB) at (-1.25,0.85) {$\bob$};
    \node[nlsdparty] (searchC) at (1.25,0.85) {$\charlie$};
    \draw[nlsdarrow] (searchstate)--node[above,sloped,font=\scriptsize]
      {$\regB$} (searchB);
    \draw[nlsdarrow] (searchstate)--node[above,sloped,font=\scriptsize]
      {$\regC$} (searchC);
    \node[nlsdstate,below=0.55cm of searchB] (zB) {$z_{\bob}\in\ftwo^n$};
    \node[nlsdstate,below=0.55cm of searchC] (zC) {$z_{\charlie}\in\ftwo^n$};
    \draw[nlsdarrow] (searchB)--(zB);
    \draw[nlsdarrow] (searchC)--(zC);
    \node[nlsdresult,text width=4.5cm] at (0,-1.15)
      {win if $z_{\bob}=z_{\charlie}=x$};
  \end{scope}

  \draw[NLSDGray!45,thick] (0,-1.55)--(0,2.95);

  \begin{scope}[shift={(3.3,0)}]
    \node[font=\bfseries,text=NLSDBlue] at (0,2.75)
      {Identical-challenge decision};
    \node[nlsdstate] (predstate) at (0,2.28) {$\rho_x$};
    \node[nlsdparty] (predB) at (-1.25,0.85) {$\bob$};
    \node[nlsdparty] (predC) at (1.25,0.85) {$\charlie$};
    \node[nlsdchallenge] (rleft) at (-2.35,0.85) {$r$};
    \node[nlsdchallenge] (rright) at (2.35,0.85) {$r$};
    \draw[nlsdarrow] (predstate) to[out=205,in=100]
      node[above,sloped,font=\scriptsize] {$\regB$} (predB);
    \draw[nlsdarrow] (predstate) to[out=335,in=80]
      node[above,sloped,font=\scriptsize] {$\regC$} (predC);
    \draw[-{Latex[length=2mm]},thick,draw=NLSDGold]
      (rleft)--(predB);
    \draw[-{Latex[length=2mm]},thick,draw=NLSDGold]
      (rright)--(predC);
    \node[nlsdstate,below=0.55cm of predB] (bB) {$b_{\bob}\in\ftwo$};
    \node[nlsdstate,below=0.55cm of predC] (bC) {$b_{\charlie}\in\ftwo$};
    \draw[nlsdarrow] (predB)--(bB);
    \draw[nlsdarrow] (predC)--(bC);
    \node[nlsdresult,text width=4.8cm] at (0,-1.15)
      {win if $b_{\bob}=b_{\charlie}=\langle r,x\rangle$};
  \end{scope}
\end{tikzpicture}
\caption{The search and the identical-challenge decision experiments. }
\label{fig:nl-s2d-games}
\end{figure}

Note that if, instead of two parties, there were exactly one party, then the celebrated Goldreich-Levin theorem~
\cite{GL89} (and its quantum generalizations~\cite{AC02,CLLZ21}) answers the above question in the affirmative. The non-triviality of the above question stems from the fact that there are two parties which are entangled with each other. If the two parties do not share any entanglement (and instead share correlated randomness) then the above question has already been answered in the affirmative~\cite{CakanGoyal24,CKNY25,CLX26}. Indeed, the above question was left as an important open problem by many works~\cite{KT25,AB24,CHV23,AKY25,CLX26}.

\paragraph{Our Results.}
Our main theorem addresses the above non-local search-to-decision problem. We show the following: 

\begin{theorem}[Informal]
\label{thm:intro:nlstodftwo}
If $p_{\srch}$ is the optimal success probability in the search experiment and $p_{\pred}$ is the optimal success probability in the identical-challenge decision experiment (as defined above), then: 
$$p_{\pred} \leq  \min\left\{1,\frac12+
 5p_{\srch}^{1/22}\right\}$$
As a consequence, if $p_{\srch}$ is negligible (in $n$) then $p_{\pred}$ is at most $\frac{1}{2}+$ negligible (in $n$).  
\end{theorem}

\noindent We note that the exponent $1/22$ is not optimal (as demonstrated by the concurrent work; more on this later). However, for applications, we typically consider the setting when $p_{\srch}$ is negligible and in this setting, the constant exponent does {\bf not} significantly affect the qualitative result statement.  
\par The above result is concerned with the non-local search-to-decision problem over $\mathbb{F}_2$. We also extend the above result to arbitrary finite fields. In this case, however, we obtain non-trivial (but sub-optimal) bounds for large-sized finite fields.

\begin{theorem}[Informal]
\label{thm:intro:finitefield}
If $p_{\srch}$ and $p_{\pred}$ are optimal success probabilities of the extended (i.e., over a field $\mathbb{F}$) search and identical-challenge decision\footnote{Even though in the extended decision experiment, the outputs of $\bob$ and $\charlie$ are no longer boolean, for consistency purposes, we abuse terminology and call this a decision experiment.} experiments, then: 
$$p_{\pred}\leq \min\!\left\{1,\;6\left(\frac1q+2\left(\frac{q-1}{q}\right)^{2/3}p_{\srch}^{1/3}\right)^{1/3}\right\},$$
where $q=|\mathbb{F}|$. As a consequence, for $q > 216$, if $p_{\srch}$ is negligible (in $n$) then $p_{\pred}$ is at most $\frac{6}{q^{\frac{1}{3}}}+$ negligible (in $n$); see the graph below.  
\end{theorem}

\par\medskip
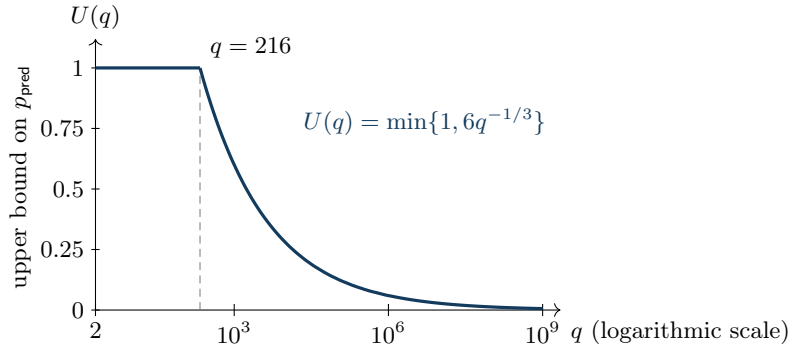
\begin{center}\begin{tikzpicture}[x=0.68cm,y=3.2cm,font=\footnotesize]
  \def\xleft{0.30103}
  \def\xthreshold{2.33445}
  \def\xright{9}
  \draw[->] (\xleft,0) -- (9.35,0)
    node[below right] {$q$ (logarithmic scale)};
  \draw[->] (\xleft,0) -- (\xleft,1.12)
    node[above] {$U(q)$};
  \foreach \y/\lab in {0/{0},0.25/{0.25},0.5/{0.5},0.75/{0.75},1/{1}}
    \draw (\xleft-0.08,\y) -- (\xleft+0.08,\y)
      node[left=2pt] {$\lab$};
  \foreach \x/\lab in {0.30103/{2},3/{10^3},6/{10^6},9/{10^9}}
    \draw (\x,-0.025) -- (\x,0.025)
      node[below=3pt] {$\lab$};
  \draw[densely dashed,gray] (\xthreshold,0) -- (\xthreshold,1);
  \node[above right] at (\xthreshold,1) {$q=216$};
  \draw[NLSDBlue,line width=1.2pt]
    (\xleft,1) -- (\xthreshold,1);
  \draw[NLSDBlue,line width=1.2pt,domain=2.33445:9,samples=120,smooth,variable=\t]
    plot ({\t},{6*exp(-ln(10)*\t/3)});
  \node[anchor=west,NLSDBlue] at (4.15,0.78)
    {$U(q)=\min\{1,6q^{-1/3}\}$};
  \node[rotate=90] at (-1.15,0.55)
    {upper bound on $p_{\pred}$};\end{tikzpicture}\par\smallskip
{\small Here $U(q):=\min\{1,6q^{-1/3}\}$ is the upper bound on $p_{\pred}$ after suppressing the negligible term; it equals $1$ through
$q=216$ and then decreases as $q^{-1/3}$. Here, we do not treat $q$ as the size of a finite field and in particular, allow it to take any positive value.}
\end{center}

\par\medskip

\noindent As a result of Theorem~\ref{thm:intro:nlstodftwo}, we obtain the following corollary to unclonable encryption~\cite{BL20,AK21}. Informally speaking, unclonable encryption is a one-time encryption scheme in which the ciphertexts are unclonable quantum states. There are two security notions considered for unclonable encryption, namely, search security and indistinguishability security. In the search security, the non-local adversaries (after the splitting stage) are expected to predict the message encrypted in the unclonable challenge ciphertext whereas in the indistinguishability security, the non-local adversaries are expected to distinguish whether the message encrypted in the challenge ciphertext is one of two adversarially chosen messages.

\begin{corollary}[Application to unclonable encryption]
\label{cor:ue-search-to-ind}
Suppose a one-time secret-key scheme for unclonable encryption of uniform
messages in $\ftwo^{n(\lambda)}$ has information-theoretic search security.
Then there is a one-time secret-key scheme for one-bit messages whose
identical-challenge distinguishing success is at most
$1/2+\operatorname{negl}(\lambda)$.
\end{corollary}

\noindent The above corollary follows in two steps: 
\begin{itemize}
    \item Ananth, Kaleoglu, and
Liu~\cite[Section~5.4.1, Corollaries~4--5]{AKL23} prove the above search-to-indistinguishability corollary but for a different primitive, namely one-time
single-decryptor encryption secure with respect to identical challenges\footnote{Technically, in~\cite{AKL23}, the single-decryptor encryption is proven secure with respect to independent challenges. However, the same result can be readily extended to work with identical challenges as well.}. Roughly speaking, a single-decryptor encryption is an encryption scheme where the decryption keys are unclonable quantum states. Here, identical challenges refer to the fact that in the security experiment of single-decryptor encryption, the non-local parties receive the same challenge ciphertext. 
\item Moreover, Georgiou and
Zhandry~\cite[Theorem~2]{GZ20} show the equivalence between one-time single-decryptor encryption with respect to identical challenges and one-time unclonable encryption. 
\end{itemize}
The resulting construction of indistinguishability-secure unclonable encryption retains the same ciphertext structure as the search-secure unclonable encryption. As a result, we obtain indistinguishability-secure unclonable encryption based on a family of unclonable states, including BB84 states and coset states. Prior works designed unclonable encryption based on Pauli-based states~\cite{AS26,Rag26} or based on Clifford-based states~\cite{BBC26}.

\subsection{Technical Overview}

\subsubsection{Search-to-Decision over $\mathbb{F}_2$}
We follow the notation set up in Section~\ref{sec:intro}. That is, we define $\rho_x$ to be the bipartite state received by $\bob$ and $\charlie$, $p_{\pred}$ to be the probability of simultaneously predicting the inner product and $p_{\srch}$ to be the optimal probability of simultaneously guessing the entire string. 
\par Additionally, we define the following notation: 
\begin{itemize}
    \item Let $\{\predpovm_{\bob}^{r,0},\predpovm_{\bob}^{r,1}\}$ be a family of POVM measurements, parameterized by $r \in \mathbb{F}_2^n$, performed by $\bob$.  
    \item Similarly, let $\{\predpovm_{\charlie}^{r,0},\predpovm_{\charlie}^{r,1}\}$ be a family of POVM measurements, parameterized by $r \in \mathbb{F}_2^n$, performed by $\charlie$.
\end{itemize}
\noindent Thus, we can express the probability of predicting the inner product by $\bob$ and $\charlie$ as follows: 
\begin{equation}
 p_{\pred}
 :=
 \expect_{x,r}\left[
 \Tr\left(
  \left(
   \predpovm_{\bob}^{r,\langle r,x\rangle}
   \otimes
   \predpovm_{\charlie}^{r,\langle r,x\rangle}
  \right)\rho_x
 \right)
 \right],
 \label{eqn:intro:ppred}
\end{equation}
We first purify $\rho_x$ to obtain a state of the form $\ket{\psi^x}_{\regB \regC \regE}$, where $\regB$ is the register measured by $\bob$, $\regC$ is the register measured by $\charlie$ and $\regE$ is the external reference register. Moreover, we define $\bfG_x$ to be $\expect_{r} \left[ \predpovm_{\bob}^{r,\langle r,x \rangle} \otimes \predpovm_{\charlie}^{r,\langle r,x \rangle} \right] \otimes I_{\regE}$. Thus, we can rewrite $p_{\pred}$ as:
$$p_{\pred} = \expect_x\left[ \bra{\psi^x}   \bfG_x \ket{\psi^x} \right] $$

\noindent \underline{\textsc{Approach in a Nutshell}}: In the rest of the proof, we decompose $\ket{\psi^x}$ into many parts and then we analyze the action of $\bfG_x$ on both diagonal and cross terms that arise from these parts. In a bit more detail, we express $\ket{\psi^x}$ as follows: 
$$\ket{\psi^x} = \ket{\psi^x_{\HH}} + \ket{\psi^x_{\HL}} + \ket{\psi^x_{\LH}} + \ket{\psi^x_{\LL}} $$
Informally speaking,
\begin{itemize}
    \item $\ket{\psi^x_{\HH}}$: on this component, both $\bob$ and $\charlie$ succeed with probability above a certain threshold $\tau$. 
    \item $\ket{\psi^x_{\LL}}$: on this component, both $\bob$ and $\charlie$ succeed with probability below $\tau$. 
    \item $\ket{\psi^x_{\HL}}$: on this component, $\bob$ succeeds with probability above $\tau$ whereas $\charlie$ succeeds with probability below $\tau$. 
    \item $\ket{\psi^x_{\LH}}$: on this component, $\bob$ succeeds with probability below $\tau$ whereas $\charlie$ succeeds with probability above $\tau$. 
\end{itemize}
Thus, we can write $p_{\pred}$ as: 
$$p_{\pred} = \underbrace{\sum_{U \in \{\HH,\HL,\LH,\LL\}} \expect_x \left[ \bra{\psi^{x}_{U}} \bfG_x \ket{\psi^x_{U}} \right]}_{\text{diagonal terms}} + \underbrace{\sum_{\substack{U,V \in \{\HH,\HL,\LH,\LL\}\\ U \neq V}} \expect_x \left[ \bra{\psi^{x}_{U}} \bfG_x \ket{\psi^x_{V}} \right]}_{\text{cross terms}} $$
The rest of the proof is to upper bound each of these terms. There are two relatively simple cases: 
\begin{itemize}
    \item A term is of the form $\bra{\psi^{x}_{U}} G_x \ket{\psi^x_{V}}$ when either $U=\HH$ or $V=\HH$: we can upper bound these terms by designing a non-local adversary that can simultaneously predict $x$ given $\ket{\psi^x}$ whose success probability is in turn upper bounded by\footnote{We are being loose here; technically, we suffer from a multiplicative loss.} $p_{\srch}$. 
    \item The averaged norm of any of the terms $\ket{\psi^x_{\HL}},\ket{\psi^x_{\LH}}$ or $\ket{\psi^x_{\LL}}$ is low: in this case, it can be shown that $\bra{\psi^{x}_{U}} G_x \ket{\psi^x_{V}}$ has to be low as well. In this case, we can safely ignore these terms in the above analysis. 
\end{itemize}
We then strip away the terms from $\ket{\psi^x}$ that capture the above cases and then we are left with the state $\ket{\phi^x}$. It now suffices to determine $p_{\pred}$ in terms of $\bfG_x$. We in turn express $\ket{\phi^x}$ as the following sum:
$$\ket{\phi^x} = \ket{\phi^x_{\HL}} +\ket{\phi^x_{\LH}} + \ket{\phi^x_{\LL}},$$
where $\ket{\phi^x_{U}}$ is either $\ket{\psi_U^x}$ or 0, depending on the second bullet point above.
\par Instead of determining an upper bound for $\expect_x\left[ \bra{\phi^x} \bfG_x \ket{\phi^x} \right]$, we instead define a new operator $S_x$: $S_x:=\left(\frac12+\eta\right)I-\bfG_x$, where $\eta$ is some function of $p_{\srch}$. For instance, if $p_{\srch}$ is negligible then so is $\eta$. Now, the new goal\footnote{Although the approach works, the author found the step of analyzing the action of $S_x$ instead of $\bfG_x$ to not be a natural one.} is to show that $\expect_x \left[ \bra{\phi^x} S_x \ket{\phi^x} \right] \geq 0$. Since $\| \phi^x \|^2 \leq 1$, this implies that $\expect_x\left[ \bra{\phi^x} \bfG_x \ket{\phi^x} \right] \leq \frac{1}{2} + \eta$. 
\par Let us expand out $\expect_x\left[ \bra{\phi^x} S_x \ket{\phi^x} \right]$:
$$\expect_x\left[ \bra{\phi^x} S_x \ket{\phi^x} \right] = \sum_{U,V \in \{\HL,\LH,\LL\}} \expect_x\left[ \bra{\phi^x_{U}} S_x \ket{\phi^x_{V}} \right] $$
Informally speaking, the terms involving $\ket{\phi^x_{\LL}}$ should not contribute much to the sum; this is because $\ket{\phi^x_{\LL}}$ is precisely the component wherein both $\bob$ and $\charlie$ fail to succeed above a certain threshold. This suggests that we could minimize over all possible $\ket{\phi^x_{\LL}}$ wherein $\bob$ and $\charlie$'s success probability is below $\tau$. If we perform the minimization, we end up showing that:
$$\expect_x\left[ \bra{\phi^x} S_x \ket{\phi^x} \right] \geq \left(\eta - \frac{\eta^3}{2} \right) \left( \underbrace{\expect_x \left[ \| \ket{\phi^x_{\HL}} \|^2 \right]}_{\text{weight of }\HL \text{ term}} + \underbrace{\expect_x \left[ \| \ket{\phi^x_{\LH}} \|^2 \right]}_{\text{weight of }\LH \text{ term}} \right)  - 2 |\expect_x \left[ c_x \right]|,$$
where $c_x$ is a cross term involving $\ket{\phi^x_{\HL}}$ and $\ket{\phi^x_{\LH}}$. Thus, we effectively eliminated cross terms involving $\ket{\phi^x_{\LL}}$ and are left with analyzing cross terms involving $\ket{\phi^x_{\HL}}$ and $\ket{\phi^x_{\LH}}$. We clarify that $c_x$ is of the  form $\bra{\phi^x_{\HL}} \bfM \ket{\phi^x_{\LH}}$, where $\bfM$ is not necessarily $S_x$. 
\par Finally, it remains to show that $2 |\expect_x \left[ c_x \right]| \leq \left(\eta - \frac{\eta^3}{2} \right) \left( \expect_x \left[ \| \ket{\phi^x_{\HL}} \|^2 \right] + \expect_x \left[ \| \ket{\phi^x_{\LH}} \|^2 \right] \right)$; showing this would imply that $\expect_x\left[ \bra{\phi^x} S_x \ket{\phi^x} \right] \geq 0$. If not, then we use the fact that $2 |\expect_x \left[ c_x \right]|$ is large enough to come up with a non-local adversary that can simultaneously predict $x$ with probability larger than $p_{\srch}$. This then would complete the proof. 
\par For the rest of the overview, we break down the above template into distinct steps and also provide the appropriate pointers to the technical sections.

\paragraph{Step 0 [Section~\ref{sec:nlsd-main} and Section~\ref{sec:nlsd-decompose}]: Expanding out $\bfG_x$.} The first task is to expand out the description of the operator $\bfG_x$. Recall that:
$$\bfG_x = \expect_{r} \left[ \predpovm_{\bob}^{r,\langle r,x \rangle} \otimes \predpovm_{\charlie}^{r,\langle r,x \rangle} \right] \otimes I_{\regE}$$ 
For notational simplicity, we ignore $I_{\regE}$ from here onward. We define the following operators: $\Delta_{\bob}^r$ to be $\predpovm_{\bob}^{r,0} - \predpovm_{\bob}^{r,1}$; similarly, define $\Delta_{\charlie}^r$ to be $\predpovm_{\charlie}^{r,0} - \predpovm_{\charlie}^{r,1}$. Thus, we can express $\bfG_x$  as follows: 
\begin{eqnarray*}
\bfG_x & = & \expect_{r} \left[ \predpovm_{\bob}^{r,\langle r,x \rangle} \otimes \predpovm_{\charlie}^{r,\langle r,x \rangle} \right] \\
& = & \expect_{r} \left[ \frac{\left( I_{\regB} + (-1)^{\langle r,x \rangle} \Delta^r_{\bob}  \right)}{2}  \otimes \frac{\left( I_{\regC} + (-1)^{\langle r,x \rangle} \Delta^r_{\charlie}  \right)}{2} \right] \\
& = & \frac{1}{4} \left( I_{\regB \regC} + \underbrace{\expect_{r} \left[ (-1)^{\langle r,x \rangle} \Delta^r_{\bob} \right]}_{\bfX^x} \otimes I_{\regC} + I_{\regB} \otimes \underbrace{\expect_{r} \left[ (-1)^{\langle r,x \rangle} \Delta^r_{\charlie} \right]}_{\bfY^x} + \underbrace{\expect_{r} \left[ \Delta^r_{\bob} \otimes \Delta^r_{\charlie} \right]}_{\bfZ} \right) \\ 
& = & \frac14\left(
   I_{\regB\regC}
   +\bfX^x\otimes I_{\regC}
   +I_{\regB}\otimes\bfY^x
   +\bfZ \right)
\end{eqnarray*}
\noindent The reason behind expanding out $\bfG_x$ this way will be clear in the next step: this will help us decompose $\ket{\psi^x}$ into the four components corresponding to $\HH,\HL,\LH$ and $\LL$, as described earlier.

\paragraph{Step 1 [Section~\ref{sec:nlsd-decompose}]: Decomposing $\ket{\psi^x}$ into four components $\HH,\HL,\LH$ and $\LL$.} Since $\bfX^x$ is Hermitian, we can look at the spectral decomposition of $\bfX^x$. Let $\{\ket{u_i}\}$ be an orthonormal basis such that: 
$$\bfX^x = \sum_{i:\lambda_i} \lambda_i \ketbra{u_i}{u_i} $$
We ignore the terms where $|\lambda_i| < \tau$, where $\tau$ is some function defined as a function of $p_{\srch}$. Define the projector $\Pi_x$ as follows: $\Pi_x = \sum_{|\lambda_i| \geq \tau} \ketbra{u_i}{u_i}$. Similarly, define $\Theta_x = \sum_{|\gamma_i| \geq \tau} \ketbra{v_i}{v_i}$, where $\bfY^x = \sum_{i} \gamma_i \ketbra{v_i}{v_i} $. Intuitively, $\Pi_x$ projects $\bob$'s state onto $\bob$'s eigenstates with absolute eigenvalues at least $\tau$ and similarly, $\Theta_x$ projects $\charlie$'s state onto $\charlie$'s eigenstates with absolute eigenvalues at least $\tau$. Now, define the complement projectors, namely $N_x = I - \Pi_x$ and $O_x = I - \Theta_x$. Using these projectors, we can correspondingly define the four product projectors: 
\[
 \begin{array}{ll}
 P_{\HH}^x=\Pi_x\otimes\Theta_x,
 &P_{\HL}^x=\Pi_x\otimes O_x,\\[0.25em]
 P_{\LH}^x=N_x\otimes\Theta_x,
 &P_{\LL}^x=N_x\otimes O_x.
 \end{array}
\]
We are now ready to define the different components of $\ket{\psi^x}$:
\begin{itemize}
    \item $\ket{\psi^x_{\HH}} = \left( \Pi_x\otimes\Theta_x \right) \ket{\psi^x},$
    \item $\ket{\psi^x_{\HL}} = \left( \Pi_x\otimes O_x \right) \ket{\psi^x},$
    \item $\ket{\psi^x_{\LH}} = \left( N_x\otimes \Theta_x \right) \ket{\psi^x},$
    \item $\ket{\psi^x_{\LL}} = \left( N_x\otimes O_x \right) \ket{\psi^x},$
\end{itemize}
Since the above
orthogonal projectors sum to the identity, we have: 
$\ket{\psi^x}=\ket{\psi_{\HH}^x}+\ket{\psi_{\HL}^x}
+\ket{\psi_{\LH}^x}+\ket{\psi_{\LL}^x}$.

Substituting this decomposition into the prediction probability produces
four diagonal terms and six cross terms:
\[
 \begin{aligned}
 p_{\pred}
 ={}&
 \sum_{U\in\{\HH,\HL,\LH,\LL\}}
 \expect_x\bra{\psi_U^x}\bfG_x\ket{\psi_U^x}\\
 &+2\operatorname{Re}\expect_x\Bigl[
 \bra{\psi_{\HH}^x}\bfG_x\ket{\psi_{\HL}^x}
 +\bra{\psi_{\HH}^x}\bfG_x\ket{\psi_{\LH}^x}
 +\bra{\psi_{\HH}^x}\bfG_x\ket{\psi_{\LL}^x}\\
 &\hspace{8.5em}
 +\bra{\psi_{\HL}^x}\bfG_x\ket{\psi_{\LH}^x}
 +\bra{\psi_{\HL}^x}\bfG_x\ket{\psi_{\LL}^x}
 +\bra{\psi_{\LH}^x}\bfG_x\ket{\psi_{\LL}^x}
 \Bigr].
 \end{aligned}
\]

\paragraph{Step 2 [Sections~\ref{sec:nlsd-omit}--\ref{sec:nlsd-retained}]: Removing the easy and negligible sectors.} We first explain why the $\HH$ component must have small average weight. It can be shown that the set of operators $\{\tau^2\Pi_x\}_x$ has the following properties: each of $\tau^2\Pi_x$ is positive semidefinite; moreover, $\sum_x\tau^2\Pi_x \leq I_{\regB}$. Thus, we can define a POVM $\Gamma_{\bob} = \left\{ \left\{\tau^2\Pi_x \right\}_x,\ I_{\regB}-\sum_x\tau^2\Pi_x \right\}$ that can be applied to $\bob$'s state. Similarly, we can define a POVM $\Gamma_{\charlie} = \left\{ \left\{\tau^2\Theta_x \right\}_x,\ I_{\regC}-\sum_x\tau^2\Theta_x \right\}$ that can be applied to $\charlie$'s state. 
\par This yields a simultaneous search strategy.  Bob measures his register using $\Gamma_{\bob}$ and outputs the measurement outcome; Charlie does
the same with his own measurement $\Gamma_{\charlie}$.  When the hidden string is $x$, the probability that
both conclusive outcomes equal $x$ is
\[
\begin{aligned}
  \bra{\psi^x}
    \bigl(\tau^2\Pi_x\otimes\tau^2\Theta_x\otimes I_{\regE}\bigr)
  \ket{\psi^x}
  &=\tau^4
    \bigl\|(\Pi_x\otimes\Theta_x\otimes I_{\regE})
      \ket{\psi^x}\bigr\|^2 \\
  &=\tau^4\bigl\|\ket{\psi_{\HH}^x}\bigr\|^2.
\end{aligned}
\]
Thus its average success probability is at least
$\tau^4\expect_x\|\ket{\psi_{\HH}^x}\|^2$.  Since $p_{\srch}$ is the optimal success probability of any
simultaneous full-string search strategy, this quantity is at most
$p_{\srch}$.  Therefore: 
\[
  \expect_x\bigl\|\ket{\psi_{\HH}^x}\bigr\|^2
  \leq \frac{p_{\srch}}{\tau^4}
\]

We now set aside $\ket{\psi_{\HH}^x}$ for the remainder of the calculation and charge every diagonal or cross term containing it to an error term.  We also set aside any one of the $\HL,\LH,\LL$ components whose average squared norm is below the same threshold $\tau$.\footnote{For readability, the overview uses one symbol for both cutoffs.  The technical proof uses two parameters: $\tau$ cuts off the magnitude of the local eigenvalues, whereas $\delta$ cuts off the average weight of an entire sector family.  Keeping them separate allows the two resulting error terms to be optimized independently.}  This weight cutoff ensures that every component retained for the later extraction argument has enough average mass; otherwise normalizing that component could incur an uncontrolled loss.  These are only mathematical deletions used in the proof; Bob and Charlie do not project or modify their state.  Since $0\preceq\bfG_x\preceq I$, Cauchy--Schwarz bounds the total change caused by all deleted terms in terms of their total squared norm, which is at most $p_{\srch}/\tau^4+3\tau$ in the simplified notation of this overview. Let \[  \ket{\phi^x}=\ket{\phi_{\HL}^x}  +\ket{\phi_{\LH}^x}+\ket{\phi_{\LL}^x}\] be the retained vector, and write \[  W_U:=\expect_x\bigl\|\ket{\phi_U^x}\bigr\|^2  \qquad(U\in\{\HL,\LH,\LL\}).\] Each $W_U$ is either zero or at least $\tau$.  If $W_{\HL}=0$, every retained component is low for Bob, so $\bra{\phi^x}\bfG_x\ket{\phi^x}$ is at most $((1+\tau)/2)\|\phi^x\|^2$.  Averaging and using $\expect_x\left[\|\phi^x\|^2 \right] \leq1$ gives the desired bound.  Similarly, if $W_{\LH}=0$, Charlie's marginal gives the same conclusion.  It remains to study the case $W_{\HL},W_{\LH}\geq\tau$. \paragraph{Step 3 [Section~\ref{sec:nlsd-retained}]: Accounting for the $\LL$ cross terms.}
Our target for the retained vector is
\[
  \expect_x\bra{\phi^x}\bfG_x\ket{\phi^x}\leq\frac12+\eta.
\]
It is convenient to measure the shortfall from this target.  Define
\[
 S_x:=\left(\frac12+\eta\right)I-\bfG_x.
\]
It suffices to prove
$\expect_x\bra{\phi^x}S_x\ket{\phi^x}\geq0$: this bounds the retained
success by $(1/2+\eta)\expect_x\|\phi^x\|^2\leq1/2+\eta$.  Moreover, whenever Charlie
is in his low subspace, the probability that both parties are correct is at
most Charlie's marginal probability, which is at most $(1+\tau)/2$.  Thus
$S_x$ has a positive lower bound on $\HL\oplus\LL$.  The analogous statement
holds on $\LH\oplus\LL$ because Bob is low there.

The cross terms $\HL$--$\LL$ and $\LH$--$\LL$ need not be small by
themselves.  Bounding their absolute values separately would lose the
positive $\LL$ diagonal term that controls them.  Instead, fix the $\HL$ and
$\LH$ components and collect every term that depends on
$\ket{\ell}:=\ket{\phi_{\LL}^x}$.  It has the form
\[
  \bra{\ell}M_x\ket{\ell}
  -2\operatorname{Re}\langle\ell\mid u_x\rangle,
\]
where $M_x$ is the restriction of $S_x$ to the $\LL$ subspace and
$\ket{u_x}$ collects the two couplings from $\LL$ to $\HL$ and $\LH$.
The positive margin just discussed makes $M_x$ invertible on this subspace.
Completing the square gives
\[
  (\bra{\ell}-\bra{u_x}M_x^{-1})M_x
  (\ket{\ell}-M_x^{-1}\ket{u_x})
  -\bra{u_x}M_x^{-1}\ket{u_x}.
\]
The first term is nonnegative.  The most unfavorable possible $\LL$
component is therefore obtained at $\ket{\ell}=M_x^{-1}\ket{u_x}$.  This
explains both why the $\LL$ vector disappears from the lower bound and why
$M_x^{-1}$ appears.  This is the matrix version of completing the square,
also called taking a Schur complement.

Expanding the correction $\bra{u_x}M_x^{-1}\ket{u_x}$ leaves positive
diagonal contributions from $\HL$ and $\LH$, together with one effective
interaction $c_x$ between them.  Claim~\ref{clm:td-retained-reduction} makes
this precise:
\[
 \expect_x\bra{\phi^x}S_x\ket{\phi^x}
 \ge
 \left(\eta-\frac{\eta^3}{2}\right)
 (W_{\HL}+W_{\LH})
 -2\left|\expect_x c_x\right|.
\]
Thus the original three-sector problem has been reduced to controlling one
averaged $\HL$--$\LH$ interaction.

\paragraph{Step 4 [Section~\ref{sec:nlsd-expand}]: Rewriting the remaining interaction as a sum of finite paths.}
Set $\alpha:=\eta-\eta^3/2$.  If
\[
  \left|\expect_xc_x\right|
  \leq\alpha\sqrt{W_{\HL}W_{\LH}},
\]
then
\[
  2\left|\expect_xc_x\right|
  \leq2\alpha\sqrt{W_{\HL}W_{\LH}}
  \leq\alpha(W_{\HL}+W_{\LH}),
\]
where the last inequality is $2\sqrt{ab}\leq a+b$.  This cancels the only
negative term in the lower bound from Step~3.  It therefore remains to prove
the displayed estimate on $c_x$.

Recall that
\[
 \bfG_x=\frac14\bigl(I+\bfX^x\otimes I+I\otimes\bfY^x+\bfZ\bigr).
\]
The high and low subspaces were defined from the spectral decompositions of
$\bfX^x$ and $\bfY^x$.  Hence $\bfX^x$ cannot move Bob's state from his
high subspace to his low subspace, and $\bfY^x$ cannot do so on Charlie's
side.  The identity changes neither label.  Consequently, whenever a matrix
element of $\bfG_x$ connects two different sectors among $\HL,\LH,\LL$,
the identity and the two one-party terms vanish by orthogonality; only
$\bfZ/4$ remains.  Thus the direct $\HL$--$\LH$ interaction and the two
couplings through $\LL$ can all be written using $\bfZ$.

The inverse $M_x^{-1}$ from Step~3 also contains the restrictions of
$\bfX^x$ and $\bfY^x$ to the $\LL$ subspace.  These restrictions have norm
at most $\tau$, precisely because $\LL$ is low on both sides.  A standard
inverse-perturbation estimate lets us replace them at a controlled cost.
More precisely, writing
\[
  A_x:=(1+4\eta)P_{\LL}^x-Q_x,
  \qquad Q_x:=P_{\LL}^x\bfZ P_{\LL}^x,
\]
the technical proof has $4M_x=A_x-R_x$ for a small perturbation $R_x$ and
shows that $M_x^{-1}$ is close to $4A_x^{-1}$.
Here $Q_x$ is simply the part of $\bfZ$ that starts in the $\LL$ subspace
and ends there again.  Since $\|\bfZ\|\leq1$, we have $\|Q_x\|\leq1$.
The familiar scalar geometric series
\[
  \frac1{a-t}=\frac1a+\frac{t}{a^2}+\frac{t^2}{a^3}+\cdots
  \qquad(|t|<a)
\]
has the following operator analogue on the $\LL$ subspace:
\[
 \bigl((1+4\eta)P_{\LL}^x-Q_x\bigr)^{-1}
 =\sum_{k\geq0}(1+4\eta)^{-(k+1)}Q_x^k.
\]
This is the Neumann, or operator geometric, expansion.  Its purpose here is
concrete: every power $Q_x^k$ describes $k$ successive applications of
$\bfZ$, with the state required to return to $\LL$ between applications.

We record these finite paths by scalars $\Gamma_m$.  The first one is the
direct transition from $\LH$ to $\HL$ through one copy of $\bfZ$:
\[
 \Gamma_1:=\expect_x\bra{\phi_{\HL}^x}\bfZ\ket{\phi_{\LH}^x}.
\]
For $m\geq2$, define
\[
 \Gamma_m:=\expect_x\bra{\phi_{\HL}^x}
 \bfZ P_{\LL}^x(P_{\LL}^x\bfZ P_{\LL}^x)^{m-2}
 P_{\LL}^x\bfZ\ket{\phi_{\LH}^x}.
\]
This expression starts in $\LH$, returns to $\LL$ after every intermediate
application of $\bfZ$, and reaches $\HL$ after the $m$-th copy.  Claim
~\ref{clm:td-geometric-expansion} says that, up to a remainder of size at
most $\frac{\eta}{16}\sqrt{W_{\HL}W_{\LH}}$, the scalar
$\expect_xc_x$ is a geometrically weighted sum of the $\Gamma_m$'s.  We have
therefore replaced a difficult inverse by finite expressions built from
copies of the original common-challenge operator $\bfZ$, separated by
$\LL$ projections.

\paragraph{Step 5 [Sections~\ref{sec:nlsd-witness}--\ref{sec:nlsd-complete}]: Bounding every finite path by search security.}
Fix $m$ and consider $\Gamma_m$.  Replacing every occurrence of
\[
  \bfZ=\expect_r[\Delta_{\bob}^r\otimes\Delta_{\charlie}^r]
\]
by its definition expresses $\Gamma_m$ as an average over a word of
challenges $w=(r_1,\ldots,r_m)$.  Once $w$ is fixed, the corresponding
operator factors as
\[
  V^{\bob}_{x,w}\otimes V^{\charlie}_{x,w}.
\]
The first factor acts only on Bob's Hilbert space and the second only on
Charlie's.  Each $r_i$ is a common challenge for the two tensor factors,
while the different $r_i$'s are sampled independently.  The word $w$ is
independent of the hidden string and may be fixed as shared randomness before
the parties separate.  The label $x$ will become a possible measurement
outcome; it is not supplied to either party.

The maps $V^{\bob}_{x,w}$ and $V^{\charlie}_{x,w}$ are not themselves
measurement effects, but their adjoint-products are positive.  The
high-projector estimate from Step~2, together with an expansion of the
intermediate low projectors, gives
\[
  \sum_x(V^{\bob}_{x,w})^\dagger V^{\bob}_{x,w}
  \preceq(m+1)^2\tau^{-2}I,
  \qquad
  \sum_x(V^{\charlie}_{x,w})^\dagger V^{\charlie}_{x,w}
  \preceq\tau^{-2}I.
\]
The first bound contains $m+1$ because a telescoping identity removes Bob's
$m$ intermediate low projectors one at a time.  This produces one base term
and $m$ correction terms; each is bounded using one high projector $\Pi_x$
and the estimate from Step~2.  On Charlie's side, the high projector is on
the input side: $V^{\charlie}_{x,w}=K^{\charlie}_w\Theta_x$ for a contraction
$K^{\charlie}_w$.  Hence
$(V^{\charlie}_{x,w})^\dagger V^{\charlie}_{x,w}\preceq\Theta_x$ directly,
so no analogous $m+1$ loss appears.

It follows that
\[
 F_{\bob,w}^x:=\frac{\tau^2}{(m+1)^2}
 (V^{\bob}_{x,w})^\dagger V^{\bob}_{x,w},
 \qquad
 F_{\charlie,w}^x:=\tau^2
 (V^{\charlie}_{x,w})^\dagger V^{\charlie}_{x,w}
\]
form local sub-POVMs indexed by candidate strings $x$; residual effects
complete them to POVMs.  Search security bounds the probability that these
two measurements both output the true hidden string.  This controls the
squared norm of the path applied to $\ket{\psi^x}$.  Since $\Gamma_m$ is an
amplitude rather than a squared norm, a final Cauchy--Schwarz inequality
takes a square root and gives
\[
  |\Gamma_m|\leq\frac{(m+1)\sqrt{p_{\srch}}}{\tau^2}.
\]
Substituting these bounds into the geometrically weighted sum from Step~4
controls $\expect_xc_x$.  The inequality
$2\sqrt{W_{\HL}W_{\LH}}\leq W_{\HL}+W_{\LH}$ then makes the lower bound in
Step~3 nonnegative.  Finally, adding back the error incurred when the small
sectors were removed in Step~2 proves the theorem after choosing
$\eta,\tau,$ and $\delta$ as specified in the technical proof.

\subsubsection{Search-to-Decision over Large Finite Fields}

The large-field result follows a different route from the four-sector
argument over $\mathbb F_2$.  We first compare the common-challenge experiment
with an experiment in which Bob and Charlie receive independent challenge
vectors.  We then give a direct finite-field Fourier procedure that recovers
the hidden vectors from predictors for the independent-challenge experiment.
The Fourier procedure works over every finite field; the assumption that the
field is large is needed only at the end, when the independent theorem's
$1/q$ baseline passes through the cube-root common-to-independent comparison.
In this subsection, $p_{\mathrm{com}}$ denotes the
finite-field analogue of the quantity called $p_{\pred}$ above.

\paragraph{Step 1 [Section~\ref{sec:qary-common-finite}]: From a common challenge to independent challenges.}
Let $p_{\mathrm{ind}}$ be the success probability of the same local response
measurements when Bob and Charlie receive independent uniform challenges.
The correlated-to-independent inequality of Ananth, Kaleoglu, and
Liu~\cite[Lemma~18]{AKL23} gives
\[
  p_{\mathrm{com}}\leq6p_{\mathrm{ind}}^{1/3},
  \qquad\text{or equivalently}\qquad
  p_{\mathrm{ind}}\geq p_{\mathrm{com}}^3/216.
\]
Thus sufficiently large success with a common challenge forces some success
when the two challenges are sampled independently.  This step changes only
the challenge distribution in the analysis; the local response measurements
remain the same.  It is also the source of the cube-root loss in the final
common-challenge bound.

\paragraph{Step 2 [Section~\ref{sec:independent-qary}]: What the independent-challenge extractor detects.}
It is useful to state this part in the slightly more general two-secret form
used in Section~\ref{sec:independent-qary}, where the hidden vectors have
dimension $\ell$ over $\mathbb F_q$.  Suppose Bob's hidden vector is
$x$, Charlie's is $y$, their independent challenges are
$v,w\in\mathbb F_q^\ell$, and their answers are $g,h\in\mathbb F_q$.  Define
their additive errors by
\[
  E^{\mathsf B}:=g-\langle x,v\rangle,
  \qquad
  E^{\mathsf C}:=h-\langle y,w\rangle.
\]
If both answers are correct, then $E^{\mathsf B}=E^{\mathsf C}=0$, and in
particular $E^{\mathsf B}+E^{\mathsf C}=0$.  The proof establishes the
stronger statement that an advantage in satisfying this last relation already
implies joint recovery of $x$ and $y$.  This strengthening is useful because
the sum of the two errors is exactly what the Fourier calculation detects.

To encode a field element as a complex phase, fix a nontrivial additive
character $\chi:\mathbb F_q\to\mathbb C$.  It satisfies
$|\chi(s)|=1$ and $\chi(s+t)=\chi(s)\chi(t)$, and therefore plays the role
of the binary phase $(-1)^s$.  When $q=p^t$, one concrete choice is
\[
  \chi(s)=\exp\!\left(
    \frac{2\pi i}{p}\operatorname{Tr}_{\mathbb F_q/\mathbb F_p}(s)
  \right).
\]
Fix a nonzero multiplier $a\in\mathbb F_q^*$.  Bob prepares a uniform
superposition over challenges $v$, runs his predictor coherently, gives a
branch with answer $g$ the phase $\chi(ag)$, runs the predictor backward, and
applies the inverse additive Fourier transform.  If the predictor were
perfect, it would insert the phase
\[
  \chi(a\langle x,v\rangle)=\chi(\langle ax,v\rangle).
\]
The inverse Fourier transform maps this entire phase pattern to the label
$ax$; multiplying by $a^{-1}$ returns $x$.  Charlie performs the same local
procedure and obtains $y$.  This is the finite-field analogue of Hadamard-basis
extraction in the quantum Goldreich--Levin argument~\cite{AC02}; the
finite-field Fourier transform itself is standard~\cite{BCW02}.

\paragraph{Step 3 [Section~\ref{sec:independent-qary}]: Why the multiplier is shared and averaged.}
When the predictors make errors, Bob's phase is the ideal phase pointing to
$ax$, multiplied by $\chi(aE^{\mathsf B})$, and Charlie obtains the analogous
factor $\chi(aE^{\mathsf C})$.  Because both parties use the same multiplier,
their error phases combine as
\[
  \chi(aE^{\mathsf B})\chi(aE^{\mathsf C})
  =\chi\!\left(a(E^{\mathsf B}+E^{\mathsf C})\right).
\]
For a fixed hidden pair $z=(x,y)$, let
\[
  A_z(s):=\Pr[E^{\mathsf B}+E^{\mathsf C}=s\mid z]
\]
be the conditional total-error distribution.  In the extraction analysis,
the overlap of the all-zero-ancilla successful-output component with the
original purification is
\[
  \Phi_z(a):=
  \mathbb E\!\left[
    \chi\!\left(a(E^{\mathsf B}+E^{\mathsf C})\right)\mid z
  \right]
  =\sum_{s\in\mathbb F_q}A_z(s)\chi(as).
\]
Thus $\Phi_z(a)$ is the additive Fourier coefficient of the total-error
distribution in direction $a$.  The extraction success conditioned on
$(z,a)$ is at least $|\Phi_z(a)|^2$.

Over $\mathbb F_2$ there is only one nonzero multiplier.  Over a larger
field, however, different nonzero errors may contribute different complex
phases that cancel for a particular $a$.  We avoid depending on one such
coefficient by averaging over every $a\in\mathbb F_q^*$.  Additive-character
orthogonality gives the exact identity
\[
  \mathbb E_{a\in\mathbb F_q^*}|\Phi_z(a)|^2
  =\frac{q\sum_{s\in\mathbb F_q}A_z(s)^2-1}{q-1}.
\]
The sum $\sum_sA_z(s)^2$ is the collision probability of the total error.  It
equals $1/q$ when that error is uniform, and increases when extra probability
is placed on zero.  In particular, if
$A_z(0)\geq1/q+\epsilon/2$, then the last display is at least
$q^2\epsilon^2/(4(q-1)^2)$.  The hypothesis that the overall probability of
$E^{\mathsf B}+E^{\mathsf C}=0$ is $1/q+\epsilon$ implies that such good hidden pairs have total
probability at least $\epsilon/2$.  Averaging over those pairs and over $a$
therefore gives joint recovery probability at least
\[
  \frac{q^2}{8(q-1)^2}\epsilon^3.
\]
In particular, some single multiplier works for the entire hidden-pair
distribution and can be fixed before the parties separate.  The nonzero
multiplier phase kick and final multiplication by its inverse are adapted
from large-modulus one-party extraction~\cite{STHY23}; using the same
multiplier on both sides and averaging over it is the new two-party step.
The independent-challenge framework follows Kundu and Tan~\cite{KT25}.

\paragraph{Step 4 [Section~\ref{sec:qary-common-finite}]: Combining the two stages.}
Applying the independent-challenge extraction bound after Step~1 yields
\[
 p_{\mathrm{search}}^\star
 \geq
 \frac{q^2}{8(q-1)^2}
 \left(\frac{p_{\mathrm{com}}^3}{216}-\frac1q\right)_+^3,
\]
or, equivalently,
\[
 p_{\mathrm{com}}
 \leq
 6\left(
 \frac1q+2\left(\frac{q-1}{q}\right)^{2/3}
 (p_{\mathrm{search}}^\star)^{1/3}
 \right)^{1/3}.
\]
Here $p_{\mathrm{search}}^\star$ is the finite-field analogue of
$p_{\srch}$.  If $p_{\mathrm{search}}^\star$ is negligible, the explicit
upper bound is $6q^{-1/3}+\operatorname{negl}(n)$.  Hence it becomes
negligible when $q$ is exponential in the security parameter.  This pinpoints
where the large-field assumption is used: the Fourier extraction works over
every finite field, but the common-to-independent comparison leaves the
residual term $6q^{-1/3}$.  For fixed or moderately sized fields, the same
quantitative theorem remains valid, but it does not control prediction
probabilities close to the natural $1/q$ baseline.

Returning to the binary case, Figure~\ref{fig:nlsd-proof-flow} summarizes the logical flow of the binary
proof and indicates where each step is carried out in the technical section.

\begin{figure}[H]
\centering
\begin{tikzpicture}[
  node distance=4.5mm and 4.5mm,
  proofstep/.style={
    nlsdstate,
    text width=.28\textwidth,
    minimum height=2.45cm,
    inner xsep=5pt,
    inner ysep=4pt,
    font=\small
  },
  flowlabel/.style={
    fill=white,
    inner sep=1pt,
    font=\scriptsize\itshape,
    text=NLSDGray
  },
  sectiontag/.style={
    font=\scriptsize,
    text=NLSDGray
  }
]
\node[proofstep] (step0) {
  {\bfseries\color{NLSDBlue}Step 0}\\[-2pt]
  {\scriptsize\color{NLSDGray}[Sections~\ref{sec:nlsd-main}, \ref{sec:nlsd-decompose}]}\\[-1pt]
  \emph{Operator form}\\
  Expand joint success into two one-party terms and one common-challenge term.
};
\node[proofstep,right=of step0] (step1) {
  {\bfseries\color{NLSDBlue}Step 1}\\[-2pt]
  {\scriptsize\color{NLSDGray}[Section~\ref{sec:nlsd-decompose}]}\\[-1pt]
  \emph{Four spectral sectors}\\
  Separate the state according to whether each one-party correlation is large or small.
};
\node[proofstep,right=of step1] (step2) {
  {\bfseries\color{NLSDBlue}Step 2}\\[-2pt]
  {\scriptsize\color{NLSDGray}[Sections~\ref{sec:nlsd-omit}--\ref{sec:nlsd-retained}]}\\[-1pt]
  \emph{Prune and split cases}\\
  Bound $\HH$ using Fourier energy; set aside only families with small total weight.
};
\node[proofstep,below=of step2] (step3) {
  {\bfseries\color{NLSDBlue}Step 3}\\[-2pt]
  {\scriptsize\color{NLSDGray}[Section~\ref{sec:nlsd-retained}]}\\[-1pt]
  \emph{Account for $\LL$}\\
  Complete the square over $\LL$; one effective $\HL$--$\LH$ interaction remains.
};
\node[proofstep,left=of step3] (step4) {
  {\bfseries\color{NLSDBlue}Step 4}\\[-2pt]
  {\scriptsize\color{NLSDGray}[Section~\ref{sec:nlsd-expand}]}\\[-1pt]
  \emph{Expand $c_x$}\\
  Rewrite the inverse as a weighted sum of finite common-challenge paths.
};
\node[proofstep,left=of step4] (step5) {
  {\bfseries\color{NLSDBlue}Step 5}\\[-2pt]
  {\scriptsize\color{NLSDGray}[Sections~\ref{sec:nlsd-witness}--\ref{sec:nlsd-complete}]}\\[-1pt]
  \emph{Invoke search security}\\
  Normalize each path into local full-string measurements, then sum the bounds.
};

\draw[nlsdarrow] (step0)--(step1);
\draw[nlsdarrow] (step1)--(step2);
\draw[nlsdarrow] (step2)--node[flowlabel,right]{hard case}(step3);
\draw[nlsdarrow] (step3)--(step4);
\draw[nlsdarrow] (step4)--(step5);
\end{tikzpicture}
\caption{Flow of the proof.  Steps~0--2 reformulate the prediction experiment
and isolate the only hard case.  Steps~3--5 use positivity of the $\LL$ block
to turn the surviving interaction into common-challenge paths bounded by
joint search security.}
\label{fig:nlsd-proof-flow}
\end{figure}
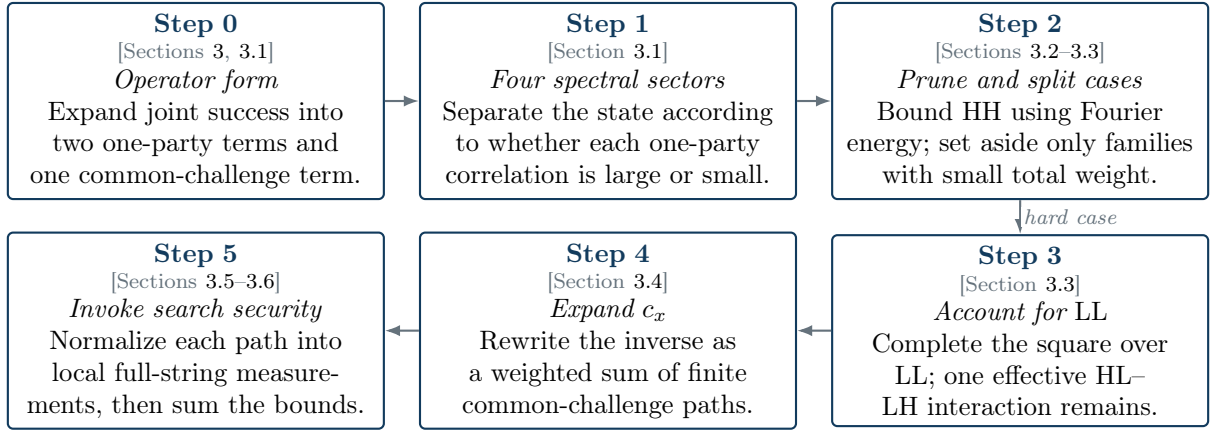

\subsection{Related Work} 

\paragraph{Different Avatars of Non-Local Search-to-Decision.} The non-local search-to-decision problem considered in this work can be cast as an instance of the generalized non-local search-to-decision problem, described in Figure~\ref{fig:nl-s2d-games-generalized}. In the $f$-generalized non-local decision problem, $\bob$ and $\charlie$ respectively receive $r_{\bob}$ and $r_{\charlie}$, where $(r_{\bob},r_{\charlie})$ are sampled from a distribution. Moreover, the winning condition is characterized by a predicate $f$ that takes as input $(b_{\bob},b_{\charlie},\langle r_{\bob},x \rangle,\langle r_{\charlie},x \rangle)$. If $f$ outputs 0 then we say that $\bob$ and $\charlie$ win the game. 

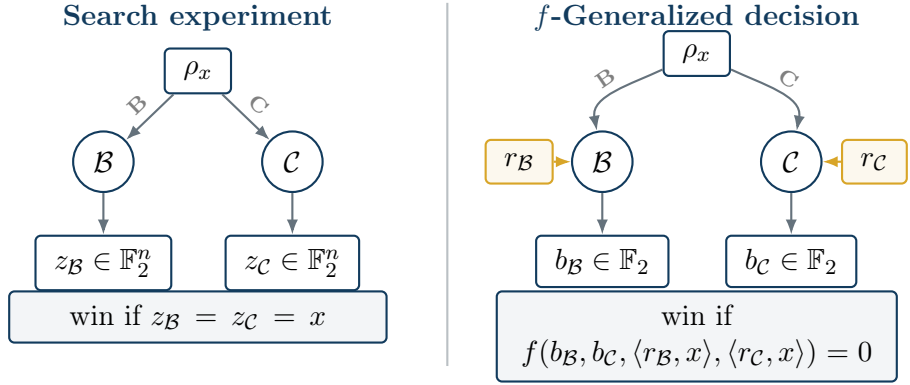
\begin{figure}[H]
\centering
\begin{tikzpicture}[x=1cm,y=1cm]
  \begin{scope}[shift={(-3.3,0)}]
    \node[font=\bfseries,text=NLSDBlue] at (0,2.75) {Search experiment};
    \node[nlsdstate] (searchstate) at (0,2.05) {$\rho_x$};
    \node[nlsdparty] (searchB) at (-1.25,0.85) {$\bob$};
    \node[nlsdparty] (searchC) at (1.25,0.85) {$\charlie$};
    \draw[nlsdarrow] (searchstate)--node[above,sloped,font=\scriptsize]
      {$\regB$} (searchB);
    \draw[nlsdarrow] (searchstate)--node[above,sloped,font=\scriptsize]
      {$\regC$} (searchC);
    \node[nlsdstate,below=0.55cm of searchB] (zB) {$z_{\bob}\in\ftwo^n$};
    \node[nlsdstate,below=0.55cm of searchC] (zC) {$z_{\charlie}\in\ftwo^n$};
    \draw[nlsdarrow] (searchB)--(zB);
    \draw[nlsdarrow] (searchC)--(zC);
    \node[nlsdresult,text width=4.5cm,anchor=north] at (0,-0.85)
      {win if $z_{\bob}=z_{\charlie}=x$};
  \end{scope}

  \draw[NLSDGray!45,thick] (0,-1.85)--(0,2.95);

  \begin{scope}[shift={(3.3,0)}]
    \node[font=\bfseries,text=NLSDBlue] at (0,2.75)
      {$f$-Generalized decision};
    \node[nlsdstate] (predstate) at (0,2.28) {$\rho_x$};
    \node[nlsdparty] (predB) at (-1.25,0.85) {$\bob$};
    \node[nlsdparty] (predC) at (1.25,0.85) {$\charlie$};
    \node[nlsdchallenge] (rleft) at (-2.35,0.85) {$r_{\bob}$};
    \node[nlsdchallenge] (rright) at (2.35,0.85) {$r_{\charlie}$};
    \draw[nlsdarrow] (predstate) to[out=205,in=100]
      node[above,sloped,font=\scriptsize] {$\regB$} (predB);
    \draw[nlsdarrow] (predstate) to[out=335,in=80]
      node[above,sloped,font=\scriptsize] {$\regC$} (predC);
    \draw[-{Latex[length=2mm]},thick,draw=NLSDGold]
      (rleft)--(predB);
    \draw[-{Latex[length=2mm]},thick,draw=NLSDGold]
      (rright)--(predC);
    \node[nlsdstate,below=0.55cm of predB] (bB) {$b_{\bob}\in\ftwo$};
    \node[nlsdstate,below=0.55cm of predC] (bC) {$b_{\charlie}\in\ftwo$};
    \draw[nlsdarrow] (predB)--(bB);
    \draw[nlsdarrow] (predC)--(bC);
    \node[nlsdresult,text width=4.8cm,anchor=north] at (0,-0.85)
      {win if $f(b_{\bob},b_{\charlie},\langle r_{\bob},x\rangle,\langle r_{\charlie},x\rangle)=0$};
  \end{scope}
\end{tikzpicture}
\caption{The search and the identical-challenge decision experiments. }
\label{fig:nl-s2d-games-generalized}
\end{figure}

\noindent We note other instances of the above generalized problem considered in the literature: 
\begin{itemize}
    \item {\bf Independent}: in this case, $(r_{\bob},r_{\charlie})$ are sampled from a product distribution on $\mathbb{F}_2^n$. Both $\bob$ and $\charlie$ win if $b_{\bob}=\langle r_{\bob},x \rangle$ and $b_{\charlie}=\langle r_{\charlie},x \rangle$. This setting was explored in the works~\cite{AKL23,KT25}. Specifically, they present a simultaneous Goldreich-Levin reduction that works for non-local adversaries in the independent challenge setting. 
    \item {\bf Correlated}: similar to the above case, $(r_{\bob},r_{\charlie})$ are sampled from a product distribution on $\mathbb{F}_2^n$. However, the winning condition is a relaxation of the above condition: $\bob$ and $\charlie$ win if $b_{\bob} \oplus b_{\charlie} = \langle r_{\bob},x \rangle \oplus \langle r_{\charlie},x \rangle$. This setting was explored in~\cite{AKY25}. 
\end{itemize}
Other variants of the $f$-generalized decision problem have also been considered in the literature. Instead of requiring $\bob$ and $\charlie$ to predict the inner product, they participate in a distinguishing game. That is, $\bob$ (and similarly, $\charlie$) receives either $\langle r_{\bob},x \rangle$ or a uniformly random element in $\mathbb{F}_2$ and they need to distinguish. The work~\cite{CG24} considered a multi-sample distinguishing variant, wherein both $\bob$ and $\charlie$ receive $o(n)$ samples of the random inner product $\langle r,x \rangle$. The work~\cite{AB24} considered the distinguishing version over an arbitrary finite field, instead of just $\mathbb{F}_2$. 

\paragraph{Applications.} Search-to-decision problems have been useful in different contexts in unclonable cryptography. Apart from unclonable encryption, they have also been used in the proofs of quantum copy-protection~\cite{AB24}. Roughly speaking, copy-protection is a compiler that converts a program into a functionality-preserving unclonable quantum state. Other applications include leakage-resilience~\cite{CakanGoyal24} and telegraphing schemes~\cite{CKNY25}.

\paragraph{Concurrent Works.}
Concurrent and independent work of Coladangelo, Liu, and Xie~\cite{CLX26Common} proves a stronger theorem for the $\mathbb{F}_2$ setting. Writing $\Delta=2p_{\pred}-1$, their reduction gives the sharper guarantee $p_{\srch}\geq \Delta^4$ and, when supplied with an efficiently computable lower bound $0<\gamma\leq\Delta$, admits a uniform implementation whose running time is polynomial in $1/\gamma$ and in the running time of the prediction strategy; without such a lower bound, they also give a uniform expected-time variant. Nonetheless, for the known unclonable cryptographic applications, our qualitatively weaker theorem suffices.

The two papers nevertheless make different technical contributions. The proof of Coladangelo, Liu, and Xie follows a filtering framework closely related to that of Ananth and Sahai~\cite{AS26}. By contrast, our proof decomposes the bipartite state received by Bob and Charlie and analyzes the action of Bob's and Charlie's measurements on these different regions using comparatively elementary ingredients. Thus, although our quantitative bound is weaker, our argument provides a structurally different and relatively direct explanation of the non-local search-to-decision phenomenon. As a result, we can view our work as providing an alternate proof of the existence of unclonable encryption using completely different techniques. We consider this alternative proof strategy as an independent conceptual contribution of the present work.

Concurrent work of Behera, \c{C}akan, and Goyal~\cite{BCG26} takes a complementary route. They prove an identical-challenge decisional monogamy theorem for coset states and use computational variants to obtain correlated-challenge single-decryptor encryption and copy-protection schemes. Unlike our generic implication from search security to decision security, their theorem exploits the algebraic structure of coset states to obtain new constructions of a variety of unclonable primitives from indistinguishability obfuscation and other cryptographic assumptions. 

The results beyond $\mathbb F_2$ in this paper remain complementary to both concurrent works: specifically, Theorem~\ref{thm:intro:finitefield} is unique to our work and does not appear in concurrent works.

\section{Preliminaries}

\subsection{Basic notation}

For a positive integer $n$, write $[n]:=\{1,\ldots,n\}$.  If $S$ is a
finite set, then $s\xleftarrow{\$}S$ denotes a uniformly random element of
$S$.  An expectation whose subscript consists only of random variables,
such as $\expect_x$ or $\expect_{x,r}$, is taken with respect to the uniform
and mutually independent choices of those variables unless another
distribution is stated explicitly.  A function
$\mu:\mathbb N\to\mathbb R_{\ge0}$ is negligible if, for every positive
polynomial $p$, there is an $N_p$ such that $\mu(n)<1/p(n)$ for every
$n\ge N_p$.

For a linear operator $M$, its operator norm is
\begin{equation}
 \norm{M}_\infty
 :=\sup_{\norm{\ket{v}}=1}\norm{M\ket{v}}.
 \label{eq:prelim-operator-norm}
\end{equation}
It is submultiplicative, and
$|\bra{u}M\ket{v}|\le\norm{M}_\infty
\norm{\ket{u}}\norm{\ket{v}}$.  For a Hermitian operator, the operator norm
is the largest absolute value of an eigenvalue.  We write
$\operatorname{ran}(M)$ for its range.  We use the vector
Cauchy--Schwarz inequality and its averaged form
\[
 \left|\expect_x\langle u_x|v_x\rangle\right|
 \le
 \left(\expect_x\norm{\ket{u_x}}^2\right)^{1/2}
 \left(\expect_x\norm{\ket{v_x}}^2\right)^{1/2}
\]
without further comment.

\subsection{Quantum Information}

A register $\regR$ is associated with a finite-dimensional Hilbert space,
also denoted $\regR$.  We write $\linear(\regR)$ for the linear operators on
$\regR$ and $\density(\regR)$ for the density operators.  The identity on
$\regR$ is $I_{\regR}$; identities on registers unaffected by an operator
are often suppressed.  An operator $M$ is positive semidefinite, written
$M\succeq0$, if $\bra{v}M\ket{v}\ge0$ for every vector $\ket{v}$.  For
Hermitian $M,N$, the Loewner order $M\preceq N$ means $N-M\succeq0$.

A state $\rho\in\density(\regR)$ satisfies $\rho\succeq0$ and
$\Tr(\rho)=1$.  A unit vector $\ket{\psi}$ represents the pure state
$\ketbra{\psi}{\psi}$.  For a bipartite state $\rho_{\regR\regS}$, its
reduced state on $\regR$ is
$\rho_{\regR}=\Tr_{\regS}(\rho_{\regR\regS})$.  A purification of
$\rho_{\regR}$ is a unit vector
$\ket{\psi}_{\regR\regE}$ satisfying
$\Tr_{\regE}(\ketbra{\psi}{\psi})=\rho_{\regR}$.

A quantum channel
$\Phi:\linear(\regR)\to\linear(\regS)$ is a completely positive,
trace-preserving linear map.  A measurement with outcome set $\mathcal X$ is
a POVM $\{M^z\}_{z\in\mathcal X}$ with $M^z\succeq0$ and
$\sum_zM^z=I$.  On state $\rho$, the Born rule assigns outcome $z$
probability $\Tr(M^z\rho)$.  A sub-POVM only requires
$\sum_zM^z\preceq I$; it can be completed to a POVM by assigning the
positive residual effect $I-\sum_zM^z$ to any outcome.  We also use cyclicity
of the trace, $\Tr(AB)=\Tr(BA)$ whenever the products are defined.

A binary POVM $\{M^0,M^1\}$ is equivalently described by the Hermitian
contraction $D:=M^0-M^1$, for which $-I\preceq D\preceq I$ and
\begin{equation}
 M^b=\frac{I+(-1)^bD}{2},
 \qquad b\in\{0,1\}.
 \label{eq:prelim-binary-observable}
\end{equation}
A projection $P$ satisfies $P=P^\dagger=P^2$.  When an inverse of a
compressed operator $PMP$ appears, it is always taken on
$\operatorname{ran}(P)$ and, when convenient, extended by zero on the
orthogonal complement.  In particular, if
$PMP\succeq aP$ for $a>0$, then this restricted inverse exists and has
operator norm at most $1/a$ (equivalently,
$(PMP)^{-1}\preceq a^{-1}P$ on $\operatorname{ran}(P)$).

\subsection{Characters}

We identify $\{0,1\}^n$ with the vector space $\ftwo^n$.  Addition and the
inner product
\begin{equation}
 \langle r,x\rangle:=\sum_{j=1}^nr_jx_j\pmod 2
 \label{eq:prelim-f2-inner-product}
\end{equation}
are over $\ftwo$.  For $x,r\in\ftwo^n$, define the character
\begin{equation}
 \chi_x(r):=(-1)^{\langle r,x\rangle}.
 \label{eq:prelim-character}
\end{equation}
The identities
\begin{equation}
 \chi_x(r+r')=\chi_x(r)\chi_x(r'),
 \qquad
 \sum_{x\in\ftwo^n}\chi_x(s)
 =\begin{cases}2^n,&s=0,\\0,&s\ne0,\end{cases}
 \label{eq:prelim-character-orthogonality}
\end{equation}
are the character orthogonality relations used in the Parseval calculation.

\section{Non-Local Search-to-Decision over \texorpdfstring{$\ftwo$}{F2}}
\label{sec:nlsd-main}

\noindent We prove the main theorem. 

\newcommand{\E}{\expect}
\newcommand{\F}{\mathbb F}
\newcommand{\ran}{\operatorname{ran}}
\newcommand{\id}{I}
\newcommand{\idB}{I_{\regB}}
\newcommand{\idC}{I_{\regC}}
\newcommand{\idBC}{I_{\regB\regC}}

\begin{theorem}
\label{thm:nl-s2d-f2}
Let $n$ be a positive integer, and let
\[
 \boldsymbol{\rho}:=(\rho_x)_{x\in\ftwo^n},
 \qquad
 \rho_x\in\density(\regB\otimes\regC),
\]
be a bipartite ensemble indexed by a uniformly distributed
$x\in\ftwo^n$.  Define its optimal joint local recovery probability by
\begin{equation}
 p_{\srch}(\boldsymbol{\rho})
 :=
 \sup_{\Lambda_{\bob}^{\srch},\Lambda_{\charlie}^{\srch}}
 \expect_x\left[
 \Tr\left(
  (\extractor_{\bob}^x\otimes\extractor_{\charlie}^x)\rho_x
 \right)
 \right].
 \label{eqn:thm:ftwo:psrch}
\end{equation}
Here the supremum is over local full-string POVMs
$\Lambda_{\bob}^{\srch}=\{\extractor_{\bob}^z\}_{z\in\ftwo^n}$ on
$\regB$ and
$\Lambda_{\charlie}^{\srch}=\{\extractor_{\charlie}^z\}_{z\in\ftwo^n}$
on $\regC$.  The parties can always output the same fixed string, and hence
\[
 2^{-n}\le p_{\srch}(\boldsymbol{\rho})\le1.
\]

For every collection of binary-outcome POVMs
\[
 \Lambda_{\bob}^{\pred}[r]
 =\{\predpovm_{\bob}^{r,0},\predpovm_{\bob}^{r,1}\},
 \qquad
 \Lambda_{\charlie}^{\pred}[r]
 =\{\predpovm_{\charlie}^{r,0},\predpovm_{\charlie}^{r,1}\},
\]
one for each $r\in\ftwo^n$, define
\begin{equation}
 p_{\pred}
 :=\expect_{r,x}\left[
 \Tr\left(
 (\predpovm_{\bob}^{r,\langle r,x\rangle}\otimes
  \predpovm_{\charlie}^{r,\langle r,x\rangle})\rho_x
 \right)
 \right],
 \label{eqn:thm:ftwo:ppred}
\end{equation}
where $x$ and $r$ are sampled independently and uniformly from $\ftwo^n$.
Then
\begin{equation}
 p_{\pred}\le
 \min\left\{1,\frac12+
 5p_{\srch}(\boldsymbol{\rho})^{1/22}\right\}.
 \label{eqn:thm:ftwo:main-bound}
\end{equation}

Consequently, let
$\{\boldsymbol{\rho}^{(n)}
=(\rho_x^{(n)})_{x\in\ftwo^n}\}_{n\in\mathbb N}$ be any family of such
ensembles, where the local register dimensions may depend on $n$.  If
$p_{\srch}(\boldsymbol{\rho}^{(n)})$ is negligible in $n$, then, for every
family of local prediction POVMs, there is a negligible function $\nu$ such
that
\[
 p_{\pred}^{(n)}\le\frac12+\nu(n).
\]
\end{theorem}

\begin{proof}
For brevity, write $p_{\srch}:=p_{\srch}(\boldsymbol{\rho})$.
Define
\begin{equation}
 \Delta_{\bob}^r:=\predpovm_{\bob}^{r,0}-\predpovm_{\bob}^{r,1},
 \qquad
 \Delta_{\charlie}^r:=\predpovm_{\charlie}^{r,0}-\predpovm_{\charlie}^{r,1}.
 \label{eq:td-difference-observables}
\end{equation}
These are Hermitian contractions, and
\[
 \predpovm_{\bob}^{r,b}=\frac{I_{\regB}+(-1)^b\Delta_{\bob}^r}{2},
 \qquad
 \predpovm_{\charlie}^{r,b}=\frac{I_{\regC}+(-1)^b\Delta_{\charlie}^r}{2}.
\]
Let
\[
 \chi_x(r):=(-1)^{\langle r,x\rangle},
\]
and define
\begin{equation}
 \bfX^x:=\expect_r[\chi_x(r)\Delta_{\bob}^r],
 \qquad
 \bfY^x:=\expect_r[\chi_x(r)\Delta_{\charlie}^r],
 \qquad
 \bfZ:=\expect_r[\Delta_{\bob}^r\otimes\Delta_{\charlie}^r].
 \label{eq:td-fourier-coefficients}
\end{equation}
For every $x$, define the effect corresponding to the event that both parties answer the parity correctly:
\begin{align}
 \bfG_x
 &:=\expect_r\!
 \left[
   \predpovm_{\bob}^{r,\langle r,x\rangle}
   \otimes
   \predpovm_{\charlie}^{r,\langle r,x\rangle}
 \right]
 \notag\\
 &=\frac14\left(
   I_{\regB\regC}
   +\bfX^x\otimes I_{\regC}
   +I_{\regB}\otimes\bfY^x
   +\bfZ
 \right).
 \label{eq:td-Gx}
\end{align}
Hence
\begin{equation}
 p_{\pred}=\expect_x[\Tr(\bfG_x\rho_x)].
 \label{eq:td-p-as-G}
\end{equation}
Since joint correctness implies Bob's correctness and also Charlie's correctness,
\begin{equation}
 0\preceq\bfG_x\preceq
 \frac{I_{\regB\regC}+\bfX^x\otimes I_{\regC}}2,
 \qquad
 0\preceq\bfG_x\preceq
 \frac{I_{\regB\regC}+I_{\regB}\otimes\bfY^x}2.
 \label{eq:td-marginal-bounds}
\end{equation}
In particular, $0\preceq\bfG_x\preceq I_{\regB\regC}$.

Fix
\begin{equation}
 0<\eta\le\frac14,
 \qquad
 \tau:=\eta^3,
 \qquad
 \delta:=\frac{\eta^2}{3072},
 \label{eq:td-parameters}
\end{equation}
and temporarily assume
\begin{equation}
 p_{\srch}\le\frac{\eta^{22}}{3072^2}.
 \label{eq:td-temporary-assumption}
\end{equation}
We will prove the conditional bound
\begin{equation}
 p_{\pred}<\frac12+1.064\eta.
 \label{eq:td-conditional-goal}
\end{equation}
The theorem will then follow by choosing $\eta$ as a function of $p_{\srch}$.

\subsection{Start from the desired expectation and decompose the state}
\label{sec:nlsd-decompose}

Choose a purification $\ket{\psi^x}_{\regB\regC\regE}$ of each $\rho_x$.
Every operator on $\regB\regC$ below is understood to act as the identity on $\regE$.
Equation~\eqref{eq:td-p-as-G} becomes
\begin{equation}
 p_{\pred}
 =\expect_x\!
 \left[
   \bra{\psi^x}\bfG_x\ket{\psi^x}
 \right].
 \label{eq:td-main-target}
\end{equation}
Our entire aim is to upper-bound the expectation on the right by $1/2$ plus a negligible quantity.

Consider spectral decompositions
\[
 \bfX^x=\sum_i\lambda_i\ketbra{u_i}{u_i},
 \qquad
 \bfY^x=\sum_i\mu_i\ketbra{v_i}{v_i}.
\]
Define Bob's high and low projections and Charlie's high and low projections by
\begin{equation}
 \Pi_x:=\sum_{i:\,|\lambda_i|\ge\tau}\ketbra{u_i}{u_i},
 \qquad N_x:=I_{\regB}-\Pi_x,
 \qquad
 \Theta_x:=\sum_{i:\,|\mu_i|\ge\tau}\ketbra{v_i}{v_i},
 \qquad O_x:=I_{\regC}-\Theta_x.
 \label{eq:td-high-low-projections}
\end{equation}
Thus $N_x\bfX^xN_x$ has operator norm at most $\tau$, while
$O_x\bfY^xO_x$ has operator norm at most $\tau$.

For each fixed $x$, define four mutually orthogonal projections
\begin{align}
 P_{\HH}^x&:=\Pi_x\otimes\Theta_x\otimes I_{\regE},
 &P_{\HL}^x&:=\Pi_x\otimes O_x\otimes I_{\regE},
 \label{eq:td-HH-HL-projections}\\
 P_{\LH}^x&:=N_x\otimes\Theta_x\otimes I_{\regE},
 &P_{\LL}^x&:=N_x\otimes O_x\otimes I_{\regE}.
 \label{eq:td-LH-LL-projections}
\end{align}
The first letter refers to Bob and the second to Charlie.  For example, $\HL$ means that Bob is in his high subspace and Charlie is in his low subspace.
These projections sum to the identity.  Define
\begin{equation}
 \ket{\psi_U^x}:=P_U^x\ket{\psi^x},
 \qquad U\in\{\HH,\HL,\LH,\LL\}.
 \label{eq:td-components}
\end{equation}
Then
\begin{equation}
 \ket{\psi^x}
 =\ket{\psi_{\HH}^x}
 +\ket{\psi_{\HL}^x}
 +\ket{\psi_{\LH}^x}
 +\ket{\psi_{\LL}^x},
 \label{eq:td-state-decomposition}
\end{equation}
and the four vectors on the right are mutually orthogonal.
For later use, put
\begin{equation}
 \operatorname{wt}(U):=\expect_x\norm{\ket{\psi_U^x}}^2.
 \label{eq:td-weights}
\end{equation}
Their weights sum to one.

We now open the target quantity in~\eqref{eq:td-main-target} using~\eqref{eq:td-state-decomposition}:
\begin{align}
 p_{\pred}
 &=\sum_{U\in\{\HH,\HL,\LH,\LL\}}
   \expect_x\bra{\psi_U^x}\bfG_x\ket{\psi_U^x}
 \notag\\
 &\quad
 +2\operatorname{Re}
 \sum_{U<V}
   \expect_x\bra{\psi_U^x}\bfG_x\ket{\psi_V^x}.
 \label{eq:td-diagonal-cross-expansion}
\end{align}
Here $U<V$ denotes any fixed ordering of the four labels.
The first line contains the four diagonal terms.  The second line contains the six cross terms.

This expansion reveals the proof strategy.

\begin{enumerate}[label=\textbf{\arabic*.},leftmargin=2.2em]
 \item The diagonal terms supported on $\HL$, $\LH$, and $\LL$ are easy: on each of these subspaces at least one party is low, so the marginal bounds in~\eqref{eq:td-marginal-bounds} put the value below $(1+\tau)/2$ times the squared norm.

 \item The $\HH$ diagonal term does not have this easy bound.  Instead, we will show that the average squared norm of $\ket{\psi_{\HH}^x}$ is small (by reducing to search security).  The same small-norm estimate will control all cross terms involving the $\HH$ component at once.

 \item After the $\HH$ component and any other component of tiny average weight are removed from the estimate, the only potentially difficult interaction is among the $\HL$, $\LH$, and $\LL$ components.  The cross terms involving $\LL$ should not be estimated one by one: their signs are uncontrolled.  We instead study
 \[
   \left(\frac12+\eta\right)I-\bfG_x
 \]
 and minimize over the $\LL$ component.  This combines the $\HL$--$\LL$ and $\LH$--$\LL$ cross terms with the $\LL$ diagonal term and leaves one effective $\HL$--$\LH$ cross scalar $c_x$.

 \item Finally, we expand the average of $c_x$ into contributions associated with fixed sequences of common challenges.  Each such contribution yields a valid pair of local full-string extraction measurements and is therefore bounded by $p_{\srch}$.
\end{enumerate}

The rest of the proof implements these four points in this order.

\subsection{The high--high component and all terms involving omitted vectors}
\label{sec:nlsd-omit}

The first lemma supplies the measurement needed to control the $\HH$ component.

\begin{lemma}[Parseval and the average weight of the $\HH$ component]
We have
\begin{equation}
 \sum_x\Pi_x\preceq\tau^{-2}I_{\regB},
 \qquad
 \sum_x\Theta_x\preceq\tau^{-2}I_{\regC},
 \label{eq:td-threshold-sums}
\end{equation}
and consequently
\begin{equation}
 \operatorname{wt}(\HH)
 \le\frac{p_{\srch}}{\tau^4}.
 \label{eq:td-HH-weight}
\end{equation}
\end{lemma}

\begin{proof}
Character orthogonality gives
\begin{align}
 \sum_x(\bfX^x)^2
 &=\sum_x\expect_{r,r'}
   \chi_x(r+r')\Delta_{\bob}^r\Delta_{\bob}^{r'}
 \notag\\
 &=\expect_r(\Delta_{\bob}^r)^2
 \preceq I_{\regB}.
 \label{eq:td-parseval-B}
\end{align}
Likewise,
\begin{equation}
 \sum_x(\bfY^x)^2\preceq I_{\regC}.
 \label{eq:td-parseval-C}
\end{equation}
On the range of $\Pi_x$, the positive operator $(\bfX^x)^2$ is at least $\tau^2I$.  Hence
$\Pi_x\preceq\tau^{-2}(\bfX^x)^2$.
Summing over $x$ proves the first inequality in~\eqref{eq:td-threshold-sums}; the second is identical.

It follows that $\{\tau^2\Pi_x\}_x$ and $\{\tau^2\Theta_x\}_x$ are local sub-POVMs.  Assigning each residual effect to an arbitrary fixed output string completes them to POVMs.  Their joint probability of outputting the correct full string is at least
\begin{align*}
 \expect_x\bra{\psi^x}
 (\tau^2\Pi_x\otimes\tau^2\Theta_x\otimes I_{\regE})
 \ket{\psi^x}
 &=\tau^4\expect_x\norm{\ket{\psi_{\HH}^x}}^2\\
 &=\tau^4\operatorname{wt}(\HH).
\end{align*}
By the definition of $p_{\srch}$, this is at most $p_{\srch}$, proving~\eqref{eq:td-HH-weight}.
\end{proof}

Rather than bounding separately the $\HH$ diagonal term and the three cross terms involving $\HH$ in~\eqref{eq:td-diagonal-cross-expansion}, we remove the entire $\HH$ vector from the mathematical estimate.  We also remove any of the other three indexed component families whose average weight is below $\delta$.  This second removal will later ensure that whenever both $\HL$ and $\LH$ remain, their average weights are bounded away from zero.

For $U\in\{\HL,\LH,\LL\}$, define
\begin{equation}
 \kappa_U:=
 \begin{cases}
  1,&\operatorname{wt}(U)\ge\delta,\\
  0,&\operatorname{wt}(U)<\delta.
 \end{cases}
 \label{eq:td-kappa}
\end{equation}
Set
\begin{equation}
 \ket{\phi_U^x}:=\kappa_U\ket{\psi_U^x},
 \qquad
 \ket{\phi^x}:=
 \ket{\phi_{\HL}^x}
 +\ket{\phi_{\LH}^x}
 +\ket{\phi_{\LL}^x},
 \label{eq:td-retained-vector}
\end{equation}
and define
\begin{equation}
 \ket{e^x}:=\ket{\psi^x}-\ket{\phi^x},
 \qquad
 q_{\mathrm{omit}}:=\expect_x\norm{\ket{e^x}}^2.
 \label{eq:td-omitted-vector}
\end{equation}
No party performs this removal.  It is only a decomposition used in the estimate.

\begin{claim}[Cost of all omitted diagonal and cross terms]
We have
\begin{equation}
 q_{\mathrm{omit}}\le\frac{p_{\srch}}{\tau^4}+3\delta,
 \label{eq:td-q-bound}
\end{equation}
and
\begin{equation}
 p_{\pred}
 \le
 \expect_x\bra{\phi^x}\bfG_x\ket{\phi^x}
 +2\sqrt{q_{\mathrm{omit}}}+q_{\mathrm{omit}}.
 \label{eq:td-deletion-bound}
\end{equation}
\end{claim}

\begin{proof}
The four original components are orthogonal.  The omitted $\HH$ component has average weight at most $p_{\srch}/\tau^4$ by~\eqref{eq:td-HH-weight}; each additional omitted component has average weight below $\delta$.  This proves~\eqref{eq:td-q-bound}.

Using $\ket{\psi^x}=\ket{\phi^x}+\ket{e^x}$,
\begin{align*}
 p_{\pred}
 &=\expect_x\bra{\phi^x}\bfG_x\ket{\phi^x}
 +2\operatorname{Re}\expect_x\bra{\phi^x}\bfG_x\ket{e^x}
 +\expect_x\bra{e^x}\bfG_x\ket{e^x}.
\end{align*}
The last term is at most $q_{\mathrm{omit}}$ because
$0\preceq\bfG_x\preceq I$.  Moreover,
\begin{align*}
 \left|\expect_x\bra{\phi^x}\bfG_x\ket{e^x}\right|
 &=\left|\expect_x
 \langle \bfG_x^{1/2}\phi^x\mid\bfG_x^{1/2}e^x\rangle\right|\\
 &\le
 \left(\expect_x\bra{\phi^x}\bfG_x\ket{\phi^x}\right)^{1/2}
 \left(\expect_x\bra{e^x}\bfG_x\ket{e^x}\right)^{1/2}\\
 &\le\sqrt{q_{\mathrm{omit}}}.
\end{align*}
This proves~\eqref{eq:td-deletion-bound}.
\end{proof}

The claim has now accounted for the $\HH$ diagonal term, every cross term involving $\HH$, and every term involving any other omitted component.  It remains to show
\begin{equation}
 \expect_x\bra{\phi^x}\bfG_x\ket{\phi^x}
 \le\frac12+\eta.
 \label{eq:td-retained-target}
\end{equation}

\subsection{Diagonal and cross terms of the retained vector}
\label{sec:nlsd-retained}

Define the retained average weights
\begin{equation}
 W_U:=\expect_x\norm{\ket{\phi_U^x}}^2,
 \qquad U\in\{\HL,\LH,\LL\}.
 \label{eq:td-retained-weights}
\end{equation}
Each $W_U$ is either zero or at least $\delta$.

We first dispose of the case in which one of the two one-high components is absent.
If $W_{\HL}=0$, then $\ket{\phi_{\HL}^x}=0$ for every $x$, and every retained vector lies in Bob's low subspace.  The first inequality in~\eqref{eq:td-marginal-bounds} gives
\begin{equation}
 \bra{\phi^x}\bfG_x\ket{\phi^x}
 \le\frac{1+\tau}{2}\norm{\ket{\phi^x}}^2
 \le\left(\frac12+\eta\right)\norm{\ket{\phi^x}}^2.
 \label{eq:td-easy-B-low}
\end{equation}
If $W_{\LH}=0$, then every retained vector lies in Charlie's low subspace, and the second inequality in~\eqref{eq:td-marginal-bounds} gives the same result.  Thus~\eqref{eq:td-retained-target} already holds whenever
\begin{equation}
 W_{\HL}=0\quad\text{or}\quad W_{\LH}=0.
 \label{eq:td-easy-case}
\end{equation}

We henceforth assume
\begin{equation}
 W_{\HL}\ge\delta,
 \qquad
 W_{\LH}\ge\delta,
 \qquad
 \sqrt{W_{\HL}W_{\LH}}\ge\delta.
 \label{eq:td-hard-case-weights}
\end{equation}
In this case $\ket{\phi_{\HL}^x}=\ket{\psi_{\HL}^x}$ and
$\ket{\phi_{\LH}^x}=\ket{\psi_{\LH}^x}$ for every $x$.

The remaining diagonal and cross terms can now be seen directly by writing
\begin{align}
 &\bra{\phi^x}
 \left(\left(\frac12+\eta\right)I-\bfG_x\right)
 \ket{\phi^x}
 \notag\\
 &=\sum_{U\in\{\HL,\LH,\LL\}}
 \bra{\phi_U^x}
 \left(\left(\frac12+\eta\right)I-\bfG_x\right)
 \ket{\phi_U^x}
 \notag\\
 &\quad
 -2\operatorname{Re}\bra{\phi_{\HL}^x}\bfG_x\ket{\phi_{\LH}^x}
 -2\operatorname{Re}\bra{\phi_{\HL}^x}\bfG_x\ket{\phi_{\LL}^x}
 -2\operatorname{Re}\bra{\phi_{\LH}^x}\bfG_x\ket{\phi_{\LL}^x}.
 \label{eq:td-retained-diagonal-cross}
\end{align}
We work with the operator
$\left(\frac12+\eta\right)I-\bfG_x$ because it is sufficient for
~\eqref{eq:td-retained-target} to prove that the average of
~\eqref{eq:td-retained-diagonal-cross} is nonnegative.  Indeed, that
nonnegativity gives
\begin{align*}
 \expect_x\bra{\phi^x}\bfG_x\ket{\phi^x}
 &\le
 \left(\frac12+\eta\right)
 \expect_x\norm{\ket{\phi^x}}^2
 \le\frac12+\eta.
\end{align*}
More importantly, on any subspace where Bob or Charlie is low, the quadratic
form of $\left(\frac12+\eta\right)I-\bfG_x$ has a positive lower
bound.  This lets us combine the $\LL$ diagonal term with both cross terms
involving $\LL$ by completing the square.

For each $x$, define the compressed $\LL$-block operator
\begin{equation}
 \boxed{
 \begin{aligned}
 M_x
 &:=
 \left.
 P_{\LL}^x
 \left(\left(\frac12+\eta\right)I-\bfG_x\right)
 P_{\LL}^x
 \right|_{\operatorname{ran}(P_{\LL}^x)}.
 \end{aligned}
 }
 \label{eq:td-claim24-Mx}
\end{equation}
Here $M_x$ is regarded as an operator from
$\operatorname{ran}(P_{\LL}^x)$ to itself.
Equation~\eqref{eq:td-LL-operator-margin-expanded} establishes that
$M_x\succeq(\eta-\eta^3/2)I$ on this subspace, so $M_x$ is invertible
there and $M_x^{-1}$ denotes this inverse.  By contrast, the sandwiched
operator in~\eqref{eq:td-claim24-Mx}, if regarded as an operator on the
entire ambient space, vanishes on $\ker(P_{\LL}^x)$ and is therefore
generally not invertible there.  If $\operatorname{ran}(P_{\LL}^x)$ is
zero-dimensional, every term containing $M_x^{-1}$ is defined to be zero.

For each $x$, define the scalar
\begin{equation}
\begin{aligned}
 c_x:={}&
 \bra{\phi_{\HL}^x}\bfG_x\ket{\phi_{\LH}^x}
 \\[0.2em]
 &+\bra{\phi_{\HL}^x}
 \bfG_xP_{\LL}^x
 M_x^{-1}
 P_{\LL}^x\bfG_x
 \ket{\phi_{\LH}^x}.
\end{aligned}
\label{eq:td-cx}
\end{equation}
The first term in $c_x$ is the direct $\HL$--$\LH$ cross term in~\eqref{eq:td-retained-diagonal-cross}.  The second is the additional $\HL$--$\LH$ cross term created when the two cross terms involving $\LL$ are minimized together with the $\LL$ diagonal term.

\begin{claim}[The retained expectation reduces to the average of $c_x$]
\label{clm:td-retained-reduction}
\begin{equation}
\begin{aligned}
 &\expect_x\bra{\phi^x}
 \left(\left(\frac12+\eta\right)I-\bfG_x\right)
 \ket{\phi^x}
 \\
 &\qquad\ge
 \left(\eta-\frac{\eta^3}{2}\right)
 (W_{\HL}+W_{\LH})
 -2\left|\expect_xc_x\right|.
\end{aligned}
\label{eq:td-main-reduction}
\end{equation}
\end{claim}

\begin{proof}
Fix $x$.  If
$\ket{z}\in\operatorname{ran}(P_{\HL}^x+P_{\LL}^x)$, then the
definitions of the two projections give
\begin{align}
 P_{\HL}^x+P_{\LL}^x
 &=\Pi_x\otimes O_x\otimes I_{\regE}
   +N_x\otimes O_x\otimes I_{\regE}
 \notag\\
 &=(\Pi_x+N_x)\otimes O_x\otimes I_{\regE}
 \notag\\
 &=I_{\regB}\otimes O_x\otimes I_{\regE}.
 \label{eq:td-Charlie-low-projection}
\end{align}
Consequently,
\begin{equation}
 (I_{\regB}\otimes O_x\otimes I_{\regE})\ket{z}=\ket{z}.
 \label{eq:td-Charlie-low-support}
\end{equation}
Since $O_x$ is the spectral projection onto the eigenvalues of
$\bfY^x$ whose absolute values are smaller than $\tau$,
\begin{align}
 &\bra{z}
 (I_{\regB}\otimes\bfY^x\otimes I_{\regE})
 \ket{z}
 \notag\\
 &\qquad=
 \bra{z}
 (I_{\regB}\otimes O_x\bfY^xO_x\otimes I_{\regE})
 \ket{z}
 \notag\\
 &\qquad\le\tau\norm{\ket{z}}^2.
 \label{eq:td-Charlie-low-insertion}
\end{align}
The second marginal inequality in~\eqref{eq:td-marginal-bounds} now gives
\begin{align}
 \bra{z}\bfG_x\ket{z}
 &\le
 \frac12\bra{z}
 \left(I_{\regB\regC}+I_{\regB}\otimes\bfY^x\right)
 \ket{z}
 \notag\\
 &=\frac12\norm{\ket{z}}^2
   +\frac12\bra{z}
   \left(I_{\regB}\otimes\bfY^x\right)
   \ket{z}
 \notag\\
 &\le \frac{1+\tau}{2}\norm{\ket{z}}^2
 =\frac{1+\eta^3}{2}\norm{\ket{z}}^2,
 \label{eq:td-Charlie-low-G-bound-expanded}
\end{align}
and therefore
\begin{align}
 &\bra{z}
 \left(\left(\frac12+\eta\right)I-\bfG_x\right)
 \ket{z}
 \notag\\
 &\qquad=
 \left(\frac12+\eta\right)\norm{\ket{z}}^2
 -\bra{z}\bfG_x\ket{z}
 \notag\\
 &\qquad\ge
 \left(\frac12+\eta-\frac{1+\eta^3}{2}\right)
 \norm{\ket{z}}^2
 \notag\\
 &\qquad=
 \left(\eta-\frac{\eta^3}{2}\right)
 \norm{\ket{z}}^2.
 \label{eq:td-Charlie-low-margin}
\end{align}
If
$\ket{z}\in\operatorname{ran}(P_{\LH}^x+P_{\LL}^x)$, then similarly
\begin{align}
 P_{\LH}^x+P_{\LL}^x
 &=N_x\otimes\Theta_x\otimes I_{\regE}
   +N_x\otimes O_x\otimes I_{\regE}
 \notag\\
 &=N_x\otimes(\Theta_x+O_x)\otimes I_{\regE}
 \notag\\
 &=N_x\otimes I_{\regC}\otimes I_{\regE}.
 \label{eq:td-Bob-low-projection}
\end{align}
Thus
\begin{equation}
 (N_x\otimes I_{\regC}\otimes I_{\regE})\ket{z}=\ket{z}.
 \label{eq:td-Bob-low-support}
\end{equation}
Since $N_x$ is the spectral projection onto the eigenvalues of
$\bfX^x$ whose absolute values are smaller than $\tau$,
\begin{align}
 &\bra{z}
 (\bfX^x\otimes I_{\regC}\otimes I_{\regE})
 \ket{z}
 \notag\\
 &\qquad=
 \bra{z}
 (N_x\bfX^xN_x\otimes I_{\regC}\otimes I_{\regE})
 \ket{z}
 \notag\\
 &\qquad\le\tau\norm{\ket{z}}^2.
 \label{eq:td-Bob-low-insertion}
\end{align}
The first marginal inequality in~\eqref{eq:td-marginal-bounds} therefore gives
\begin{align}
 \bra{z}\bfG_x\ket{z}
 &\le
 \frac12\bra{z}
 \left(I_{\regB\regC}+\bfX^x\otimes I_{\regC}\right)
 \ket{z}
 \notag\\
 &=\frac12\norm{\ket{z}}^2
   +\frac12\bra{z}
   \left(\bfX^x\otimes I_{\regC}\right)
   \ket{z}
 \notag\\
 &\le \frac{1+\tau}{2}\norm{\ket{z}}^2
 =\frac{1+\eta^3}{2}\norm{\ket{z}}^2,
 \label{eq:td-Bob-low-G-bound-expanded}
\end{align}
and therefore
\begin{align}
 &\bra{z}
 \left(\left(\frac12+\eta\right)I-\bfG_x\right)
 \ket{z}
 \notag\\
 &\qquad=
 \left(\frac12+\eta\right)\norm{\ket{z}}^2
 -\bra{z}\bfG_x\ket{z}
 \notag\\
 &\qquad\ge
 \left(\frac12+\eta-\frac{1+\eta^3}{2}\right)
 \norm{\ket{z}}^2
 \notag\\
 &\qquad=
 \left(\eta-\frac{\eta^3}{2}\right)
 \norm{\ket{z}}^2.
 \label{eq:td-Bob-low-margin}
\end{align}
In particular, for every nonzero
$\ket{z}\in\operatorname{ran}(P_{\LL}^x)$,
\begin{align}
 &\bra{z}
 P_{\LL}^x
 \left(\left(\frac12+\eta\right)I-\bfG_x\right)
 P_{\LL}^x
 \ket{z}
 \notag\\
 &\qquad=
 \bra{z}
 \left(\left(\frac12+\eta\right)I-\bfG_x\right)
 \ket{z}
 \notag\\
 &\qquad\ge
 \left(\eta-\frac{\eta^3}{2}\right)
 \norm{\ket{z}}^2
 >0.
 \label{eq:td-LL-positive-expanded}
\end{align}
The inequality before the final strict inequality also holds for
$\ket{z}=0$.  Hence, for every vector $\ket{v}$ in the full space, we may
apply it to $\ket{z}=P_{\LL}^x\ket{v}$ and obtain
\begin{align}
 &\bra{v}
 P_{\LL}^x
 \left(\left(\frac12+\eta\right)I-\bfG_x\right)
 P_{\LL}^x
 \ket{v}
 \notag\\
 &\qquad\ge
 \left(\eta-\frac{\eta^3}{2}\right)
 \bra{v}P_{\LL}^x\ket{v}.
 \label{eq:td-LL-quadratic-form-order}
\end{align}
By the definition of positive-semidefinite order through quadratic forms,
this is precisely the operator inequality
\begin{equation}
 P_{\LL}^x
 \left(\left(\frac12+\eta\right)I-\bfG_x\right)
 P_{\LL}^x
 \succeq
 \left(\eta-\frac{\eta^3}{2}\right)P_{\LL}^x.
 \label{eq:td-LL-operator-margin-expanded}
\end{equation}
Because
$\eta-\eta^3/2=\eta(1-\eta^2/2)>0$, the restriction of the
left-hand side to $\operatorname{ran}(P_{\LL}^x)$ is strictly positive and
therefore invertible there.  This proves that the inverse in
~\eqref{eq:td-cx} exists on that subspace.

Using
$\ket{\phi^x}=\ket{\phi_{\HL}^x}+\ket{\phi_{\LH}^x}+\ket{\phi_{\LL}^x}$
and the pairwise orthogonality of these three vectors,
\begin{align}
 &\bra{\phi^x}
 \left(\left(\frac12+\eta\right)I-\bfG_x\right)
 \ket{\phi^x}
 \notag\\
 &=
 \bra{\phi_{\HL}^x}
 \left(\left(\frac12+\eta\right)I-\bfG_x\right)
 \ket{\phi_{\HL}^x}
 \notag\\
 &\quad+
 \bra{\phi_{\LH}^x}
 \left(\left(\frac12+\eta\right)I-\bfG_x\right)
 \ket{\phi_{\LH}^x}
 \notag\\
 &\quad+
 \bra{\phi_{\LL}^x}
 \left(\left(\frac12+\eta\right)I-\bfG_x\right)
 \ket{\phi_{\LL}^x}
 \notag\\
 &\quad-
 2\operatorname{Re}
 \bra{\phi_{\HL}^x}\bfG_x\ket{\phi_{\LH}^x}
 \notag\\
 &\quad-
 2\operatorname{Re}
 \bra{\phi_{\HL}^x}\bfG_x\ket{\phi_{\LL}^x}
 \notag\\
 &\quad-
 2\operatorname{Re}
 \bra{\phi_{\LH}^x}\bfG_x\ket{\phi_{\LL}^x}.
 \label{eq:td-claim24-full-expansion}
\end{align}
We now combine the $\LL$ diagonal term with the two cross terms involving
$\ket{\phi_{\LL}^x}$, namely
\begin{align}
 &\bra{\phi_{\LL}^x}
 \left(\left(\frac12+\eta\right)I-\bfG_x\right)
 \ket{\phi_{\LL}^x}
 \notag\\
 &\quad-
 2\operatorname{Re}
 \bra{\phi_{\LL}^x}\bfG_x\ket{\phi_{\HL}^x}
 \notag\\
 &\quad-
 2\operatorname{Re}
 \bra{\phi_{\LL}^x}\bfG_x\ket{\phi_{\LH}^x}.
 \label{eq:td-claim24-LL-terms}
\end{align}
Here we used the Hermiticity of $\bfG_x$ to reverse the two cross scalars
inside their real parts.

Recall from~\eqref{eq:td-claim24-Mx} that $M_x$ is the compression of
$\left(\frac12+\eta\right)I-\bfG_x$ to
$\operatorname{ran}(P_{\LL}^x)$ and hence is the coefficient of the
quadratic term in $\ket{\phi_{\LL}^x}$.  By
~\eqref{eq:td-LL-operator-margin-expanded}, $M_x$ is positive definite, so
$M_x^{-1}$ exists and is Hermitian on this subspace.

For this same calculation, also define
\begin{equation}
 \ket{u_x}
 :=
 P_{\LL}^x\bfG_x
 \left(\ket{\phi_{\HL}^x}+\ket{\phi_{\LH}^x}\right).
 \label{eq:td-claim24-ux}
\end{equation}
Taking the adjoint reverses the operator order.  In particular,
\begin{equation}
 (P_{\LL}^x\bfG_x)^\dagger
 =\bfG_xP_{\LL}^x,
 \label{eq:td-claim24-adjoint-order}
\end{equation}
and hence
\begin{equation}
 \bra{u_x}
 =
 \left(\bra{\phi_{\HL}^x}+\bra{\phi_{\LH}^x}\right)
 \bfG_xP_{\LL}^x.
 \label{eq:td-claim24-ux-adjoint}
\end{equation}
No commutation between $P_{\LL}^x$ and $\bfG_x$ is being used.

Because
$P_{\LL}^x\ket{\phi_{\LL}^x}=\ket{\phi_{\LL}^x}$, the three terms in
~\eqref{eq:td-claim24-LL-terms} first satisfy the intermediate equality
\begin{align}
 &\bra{\phi_{\LL}^x}
 \left(\left(\frac12+\eta\right)I-\bfG_x\right)
 \ket{\phi_{\LL}^x}
 \notag\\
 &\quad-
 2\operatorname{Re}
 \bra{\phi_{\LL}^x}\bfG_x
 \left(\ket{\phi_{\HL}^x}+\ket{\phi_{\LH}^x}\right)
 \notag\\
 &=
 \bra{\phi_{\LL}^x}M_x\ket{\phi_{\LL}^x}
 -2\operatorname{Re}\bra{\phi_{\LL}^x}\ket{u_x}.
 \label{eq:td-claim24-LL-Mu}
\end{align}
This equality accounts for every inserted copy of $P_{\LL}^x$.

\medskip
\begin{itshape}
\noindent\textbf{Remark (completing the square and why $M^{-1}$ appears).}
The quadratic expression under consideration is precisely the part of
~\eqref{eq:td-claim24-LL-Mu} that depends on the $\LL$ vector.  Hold the
$\HL$ and $\LH$ components fixed, and let
$\ket{\ell}\in\operatorname{ran}(P_{\LL}^x)$ stand for a possible
$\LL$ component.  The relevant expression is
\[
 q_x(\ell)
 :=
 \bra{\ell}M_x\ket{\ell}
 -2\operatorname{Re}\bra{\ell}\ket{u_x}.
\]
Its value for the state in the proof is
$q_x(\phi_{\LL}^x)$.  To obtain a lower bound, we may minimize
$q_x(\ell)$ over all vectors in the $\LL$ sector, because the actual
vector $\ket{\phi_{\LL}^x}$ is one of the allowed choices.

Here is the general completion-of-the-square calculation.  Let $M$ be positive
definite and set
\[
 f(\ell):=
 \bra{\ell}M\ket{\ell}
 -2\operatorname{Re}\bra{\ell}\ket{u}.
\]
Since
\[
 2\operatorname{Re}\bra{\ell}\ket{u}
 =
 \bra{\ell}\ket{u}+\bra{u}\ket{\ell},
\]
we add and subtract $\bra{u}M^{-1}\ket{u}$.  Expanding the resulting
product gives
\begin{align*}
 &\left(\bra{\ell}-\bra{u}M^{-1}\right)
 M
 \left(\ket{\ell}-M^{-1}\ket{u}\right)
 \notag\\
 &\qquad=
 \bra{\ell}M\ket{\ell}
 -\bra{\ell}\ket{u}
 -\bra{u}\ket{\ell}
 +\bra{u}M^{-1}\ket{u}.
\end{align*}
Therefore
\[
 f(\ell)
 =
 \left(\bra{\ell}-\bra{u}M^{-1}\right)
 M
 \left(\ket{\ell}-M^{-1}\ket{u}\right)
 -\bra{u}M^{-1}\ket{u}.
\]
This is the completion of the square.  The first term is nonnegative and
vanishes exactly when
$M\ket{\ell}=\ket{u}$, equivalently
$\ket{\ell}=M^{-1}\ket{u}$.  Thus minimizing the eliminated sector
produces the correction $-\bra{u}M^{-1}\ket{u}$; this is the corresponding
Schur-complement correction.  In~\eqref{eq:td-claim24-complete-square-full},
$M=M_x$, $\ket{u}=\ket{u_x}$, and
$\ket{\ell}=\ket{\phi_{\LL}^x}$.
\end{itshape}
\medskip

Completing the square now gives the shorter identity
\begin{align}
 &\bra{\phi_{\LL}^x}M_x\ket{\phi_{\LL}^x}
 -2\operatorname{Re}\bra{\phi_{\LL}^x}\ket{u_x}
 \notag\\
 &=
 \left(\bra{\phi_{\LL}^x}-\bra{u_x}M_x^{-1}\right)
 M_x
 \left(\ket{\phi_{\LL}^x}-M_x^{-1}\ket{u_x}\right)
 -\bra{u_x}M_x^{-1}\ket{u_x}.
 \label{eq:td-claim24-complete-square-full}
\end{align}
The first term on the right is nonnegative.  Therefore the three terms in
~\eqref{eq:td-claim24-LL-terms} are at least
\begin{equation}
 -\bra{u_x}M_x^{-1}\ket{u_x}.
 \label{eq:td-claim24-complete-square-lower}
\end{equation}

Only after completing the square do we expand this scalar.  The definition
of $\ket{u_x}$ and the adjoint formula~\eqref{eq:td-claim24-ux-adjoint}
give
\begin{align}
 &\bra{u_x}M_x^{-1}\ket{u_x}
 \notag\\
 &=
 \bra{\phi_{\HL}^x}
 \bfG_xP_{\LL}^x
 M_x^{-1}
 P_{\LL}^x\bfG_x
 \ket{\phi_{\HL}^x}
 \notag\\
 &\quad+
 \bra{\phi_{\LH}^x}
 \bfG_xP_{\LL}^x
 M_x^{-1}
 P_{\LL}^x\bfG_x
 \ket{\phi_{\LH}^x}
 \notag\\
 &\quad+
 2\operatorname{Re}
 \bra{\phi_{\HL}^x}
 \bfG_xP_{\LL}^x
 M_x^{-1}
 P_{\LL}^x\bfG_x
 \ket{\phi_{\LH}^x}.
 \label{eq:td-claim24-induced-expansion}
\end{align}
The left copy of $\bfG_x$ in every scalar comes from $\bra{u_x}$, and the
right copy comes from $\ket{u_x}$; both copies are required.
By the definition of $M_x$, the operator
$P_{\LL}^xM_x^{-1}P_{\LL}^x$ is exactly the inverse in~\eqref{eq:td-cx},
embedded in the full space by the surrounding projections.  We now combine
~\eqref{eq:td-claim24-full-expansion},
~\eqref{eq:td-claim24-complete-square-lower},
~\eqref{eq:td-claim24-induced-expansion}, and the definition of $c_x$:
\begin{align}
 &\bra{\phi^x}
 \left(\left(\frac12+\eta\right)I-\bfG_x\right)
 \ket{\phi^x}
 \notag\\
 &\ge
 \Bigg[
 \bra{\phi_{\HL}^x}
 \left(\left(\frac12+\eta\right)I-\bfG_x\right)
 \ket{\phi_{\HL}^x}
 \notag\\
 &\qquad\quad-
 \bra{\phi_{\HL}^x}
 \bfG_xP_{\LL}^x
 M_x^{-1}
 P_{\LL}^x\bfG_x
 \ket{\phi_{\HL}^x}
 \Bigg]
 \notag\\
 &\quad+
 \Bigg[
 \bra{\phi_{\LH}^x}
 \left(\left(\frac12+\eta\right)I-\bfG_x\right)
 \ket{\phi_{\LH}^x}
 \notag\\
 &\qquad\quad-
 \bra{\phi_{\LH}^x}
 \bfG_xP_{\LL}^x
 M_x^{-1}
 P_{\LL}^x\bfG_x
 \ket{\phi_{\LH}^x}
 \Bigg]
 \notag\\
 &\quad-2\operatorname{Re}c_x.
 \label{eq:td-claim24-decomposition-lower-bound}
\end{align}

For every $\ket{z}\in\operatorname{ran}(P_{\LL}^x)$,
\begin{align}
 &(\bra{\phi_{\HL}^x}+\bra{z})
 \left(\left(\frac12+\eta\right)I-\bfG_x\right)
 (\ket{\phi_{\HL}^x}+\ket{z})
 \notag\\
 &=
 \bra{\phi_{\HL}^x}
 \left(\left(\frac12+\eta\right)I-\bfG_x\right)
 \ket{\phi_{\HL}^x}
 \notag\\
 &\quad+
 \bra{z}M_x\ket{z}
 \notag\\
 &\quad-
 2\operatorname{Re}
 \bra{z}P_{\LL}^x\bfG_x\ket{\phi_{\HL}^x}
 \notag\\[0.3em]
 &=
 \left(
 \bra{z}
 -
 \bra{\phi_{\HL}^x}\bfG_xP_{\LL}^x
 M_x^{-1}
 \right)
 M_x
 \left(
 \ket{z}
 -
 M_x^{-1}
 P_{\LL}^x\bfG_x\ket{\phi_{\HL}^x}
 \right)
 \notag\\
 &\quad+
 \bra{\phi_{\HL}^x}
 \left(\left(\frac12+\eta\right)I-\bfG_x\right)
 \ket{\phi_{\HL}^x}
 \notag\\
 &\quad-
 \bra{\phi_{\HL}^x}
 \bfG_xP_{\LL}^x
 M_x^{-1}
 P_{\LL}^x\bfG_x
 \ket{\phi_{\HL}^x}.
 \label{eq:td-claim24-HL-square}
\end{align}
Taking the minimum over $\ket{z}\in\operatorname{ran}(P_{\LL}^x)$ in
~\eqref{eq:td-claim24-HL-square} yields
\begin{align}
 &\bra{\phi_{\HL}^x}
 \left(\left(\frac12+\eta\right)I-\bfG_x\right)
 \ket{\phi_{\HL}^x}
 \notag\\
 &\quad-
 \bra{\phi_{\HL}^x}
 \bfG_xP_{\LL}^x
 M_x^{-1}
 P_{\LL}^x\bfG_x
 \ket{\phi_{\HL}^x}
 \notag\\
 &=
 \min_{\ket{z}\in\operatorname{ran}(P_{\LL}^x)}
 (\bra{\phi_{\HL}^x}+\bra{z})
 \left(\left(\frac12+\eta\right)I-\bfG_x\right)
 (\ket{\phi_{\HL}^x}+\ket{z})
 \notag\\
 &\ge
 \min_{\ket{z}\in\operatorname{ran}(P_{\LL}^x)}
 \left(\eta-\frac{\eta^3}{2}\right)
 \norm{\ket{\phi_{\HL}^x}+\ket{z}}^2
 \notag\\
 &=
 \min_{\ket{z}\in\operatorname{ran}(P_{\LL}^x)}
 \left(\eta-\frac{\eta^3}{2}\right)
 \left(
  \norm{\ket{\phi_{\HL}^x}}^2+\norm{\ket{z}}^2
 \right)
 \notag\\
 &=
 \left(\eta-\frac{\eta^3}{2}\right)
 \norm{\ket{\phi_{\HL}^x}}^2.
 \label{eq:td-claim24-HL-bound-expanded}
\end{align}
Likewise, for every $\ket{z}\in\operatorname{ran}(P_{\LL}^x)$,
\begin{align}
 &(\bra{\phi_{\LH}^x}+\bra{z})
 \left(\left(\frac12+\eta\right)I-\bfG_x\right)
 (\ket{\phi_{\LH}^x}+\ket{z})
 \notag\\
 &=
 \bra{\phi_{\LH}^x}
 \left(\left(\frac12+\eta\right)I-\bfG_x\right)
 \ket{\phi_{\LH}^x}
 \notag\\
 &\quad+
 \bra{z}M_x\ket{z}
 \notag\\
 &\quad-
 2\operatorname{Re}
 \bra{z}P_{\LL}^x\bfG_x\ket{\phi_{\LH}^x}
 \notag\\[0.3em]
 &=
 \left(
 \bra{z}
 -
 \bra{\phi_{\LH}^x}\bfG_xP_{\LL}^x
 M_x^{-1}
 \right)
 M_x
 \left(
 \ket{z}
 -
 M_x^{-1}
 P_{\LL}^x\bfG_x\ket{\phi_{\LH}^x}
 \right)
 \notag\\
 &\quad+
 \bra{\phi_{\LH}^x}
 \left(\left(\frac12+\eta\right)I-\bfG_x\right)
 \ket{\phi_{\LH}^x}
 \notag\\
 &\quad-
 \bra{\phi_{\LH}^x}
 \bfG_xP_{\LL}^x
 M_x^{-1}
 P_{\LL}^x\bfG_x
 \ket{\phi_{\LH}^x}.
 \label{eq:td-claim24-LH-square}
\end{align}
Taking the minimum over $\ket{z}\in\operatorname{ran}(P_{\LL}^x)$ in
~\eqref{eq:td-claim24-LH-square} yields
\begin{align}
 &\bra{\phi_{\LH}^x}
 \left(\left(\frac12+\eta\right)I-\bfG_x\right)
 \ket{\phi_{\LH}^x}
 \notag\\
 &\quad-
 \bra{\phi_{\LH}^x}
 \bfG_xP_{\LL}^x
 M_x^{-1}
 P_{\LL}^x\bfG_x
 \ket{\phi_{\LH}^x}
 \notag\\
 &=
 \min_{\ket{z}\in\operatorname{ran}(P_{\LL}^x)}
 (\bra{\phi_{\LH}^x}+\bra{z})
 \left(\left(\frac12+\eta\right)I-\bfG_x\right)
 (\ket{\phi_{\LH}^x}+\ket{z})
 \notag\\
 &\ge
 \min_{\ket{z}\in\operatorname{ran}(P_{\LL}^x)}
 \left(\eta-\frac{\eta^3}{2}\right)
 \norm{\ket{\phi_{\LH}^x}+\ket{z}}^2
 \notag\\
 &=
 \min_{\ket{z}\in\operatorname{ran}(P_{\LL}^x)}
 \left(\eta-\frac{\eta^3}{2}\right)
 \left(
  \norm{\ket{\phi_{\LH}^x}}^2+\norm{\ket{z}}^2
 \right)
 \notag\\
 &=
 \left(\eta-\frac{\eta^3}{2}\right)
 \norm{\ket{\phi_{\LH}^x}}^2.
 \label{eq:td-claim24-LH-bound-expanded}
\end{align}
Combining~\eqref{eq:td-claim24-decomposition-lower-bound},
~\eqref{eq:td-claim24-HL-bound-expanded}, and
~\eqref{eq:td-claim24-LH-bound-expanded},
\begin{align}
 &\bra{\phi^x}
 \left(\left(\frac12+\eta\right)I-\bfG_x\right)
 \ket{\phi^x}
 \notag\\
 &\qquad\ge
 \left(\eta-\frac{\eta^3}{2}\right)
 \left(
  \norm{\ket{\phi_{\HL}^x}}^2
  +\norm{\ket{\phi_{\LH}^x}}^2
 \right)
 -2\operatorname{Re}c_x.
 \label{eq:td-pointwise-reduction}
\end{align}
Finally,
\begin{align}
 &\expect_x\bra{\phi^x}
 \left(\left(\frac12+\eta\right)I-\bfG_x\right)
 \ket{\phi^x}
 \notag\\
 &\qquad\ge
 \left(\eta-\frac{\eta^3}{2}\right)
 \expect_x\left[
  \norm{\ket{\phi_{\HL}^x}}^2
  +\norm{\ket{\phi_{\LH}^x}}^2
 \right]
 -2\expect_x[\operatorname{Re}c_x]
 \notag\\
 &\qquad=
 \left(\eta-\frac{\eta^3}{2}\right)
 (W_{\HL}+W_{\LH})
 -2\operatorname{Re}\expect_x[c_x]
 \notag\\
 &\qquad\ge
 \left(\eta-\frac{\eta^3}{2}\right)
 (W_{\HL}+W_{\LH})
 -2\left|\expect_xc_x\right|.
 \label{eq:td-claim24-average-expanded}
\end{align}
This is~\eqref{eq:td-main-reduction}.
\end{proof}

By the arithmetic--geometric mean inequality, the claim shows that~\eqref{eq:td-retained-target} follows once we prove
\begin{equation}
 \left|\expect_xc_x\right|
 \le
 \left(\eta-\frac{\eta^3}{2}\right)
 \sqrt{W_{\HL}W_{\LH}}.
 \label{eq:td-cx-target}
\end{equation}
Indeed,
$2\sqrt{W_{\HL}W_{\LH}}\le W_{\HL}+W_{\LH}$.
Thus the entire proof has now been reduced to one averaged cross scalar.

\subsection{Expand the effective cross scalar into common-challenge contributions}
\label{sec:nlsd-expand}

The definition of $c_x$ still contains the inverse of the restriction of
$((1/2+\eta)I-\bfG_x)$ to the $\LL$ subspace.  The next claim rewrites this inverse as a geometric series involving the common-challenge operator $\bfZ$.  This is useful because every copy of
\[
 \bfZ=\expect_r[\Delta_{\bob}^r\otimes\Delta_{\charlie}^r]
\]
uses the same challenge on Bob's and Charlie's sides.

Define
\begin{equation}
 \Gamma_1
 :=\expect_x\bra{\phi_{\HL}^x}\bfZ\ket{\phi_{\LH}^x},
 \label{eq:td-Gamma1}
\end{equation}
and, for $m\ge2$,
\begin{equation}
 \Gamma_m
 :=\expect_x\bra{\phi_{\HL}^x}
 \bfZ P_{\LL}^x
 (P_{\LL}^x\bfZ P_{\LL}^x)^{m-2}
 P_{\LL}^x\bfZ
 \ket{\phi_{\LH}^x}.
 \label{eq:td-Gammam}
\end{equation}

\begin{claim}[Geometric expansion of the average effective cross term]
\label{clm:td-geometric-expansion}
\begin{equation}
 \left|
  \expect_xc_x
  -\frac14\left[
   \Gamma_1+
   \sum_{k=0}^{\infty}(1+4\eta)^{-(k+1)}\Gamma_{k+2}
  \right]
 \right|
 \le
 \frac{\eta}{16}\sqrt{W_{\HL}W_{\LH}}.
 \label{eq:td-series-approximation}
\end{equation}
\end{claim}

\begin{proof}
Fix $x$.

We first isolate the part of $\bfG_x$ that changes a high/low label.  In
~\eqref{eq:td-Gx}, the identity preserves each of the four high/low
subspaces.  Moreover, $\bfX^x$ commutes with its spectral projections
$\Pi_x,N_x$ and acts trivially on Charlie, while $\bfY^x$ commutes with
$\Theta_x,O_x$ and acts trivially on Bob.  In particular,
\begin{equation*}
 \Pi_xN_x=N_x\Pi_x=0,
 \qquad
 O_x\Theta_x=\Theta_xO_x=0,
\end{equation*}
and the spectral subspaces reduce the corresponding Fourier operators:
\begin{equation*}
 \Pi_x\bfX^xN_x=N_x\bfX^x\Pi_x=0,
 \qquad
 O_x\bfY^x\Theta_x=\Theta_x\bfY^xO_x=0.
\end{equation*}
We suppress the identity on $\regE$.  Using the support identities
$P_U^x\ket{\phi_U^x}=\ket{\phi_U^x}$ and the decomposition
of $\bfG_x$ in~\eqref{eq:td-Gx}, the $\HL$--$\LH$ block expands as
\begin{align}
 \bra{\phi_{\HL}^x}\bfG_x\ket{\phi_{\LH}^x}
 &=\bra{\phi_{\HL}^x}
 P_{\HL}^x\bfG_xP_{\LH}^x
 \ket{\phi_{\LH}^x}
 \notag\\
 &=\frac14\bra{\phi_{\HL}^x}
 \Bigl[
 P_{\HL}^xP_{\LH}^x
 +P_{\HL}^x(\bfX^x\otimes I_{\regC})P_{\LH}^x
 \notag\\
 &\hspace{8em}
 +P_{\HL}^x(I_{\regB}\otimes\bfY^x)P_{\LH}^x
 +P_{\HL}^x\bfZ P_{\LH}^x
 \Bigr]
 \ket{\phi_{\LH}^x}
 \notag\\
 &=\frac14\bra{\phi_{\HL}^x}
 \Bigl[
 \underbrace{(\Pi_xN_x)\otimes(O_x\Theta_x)}_{=\,0}
 +\underbrace{(\Pi_x\bfX^xN_x)\otimes(O_x\Theta_x)}_{=\,0}
 \notag\\
 &\hspace{8em}
 +\underbrace{(\Pi_xN_x)\otimes(O_x\bfY^x\Theta_x)}_{=\,0}
 +P_{\HL}^x\bfZ P_{\LH}^x
 \Bigr]
 \ket{\phi_{\LH}^x}
 \notag\\
 &=\frac14
 \bra{\phi_{\HL}^x}\bfZ\ket{\phi_{\LH}^x}.
 \label{eq:td-direct-Z}
\end{align}
Likewise, the block from $\HL$ into $\LL$ expands as
\begin{align}
 P_{\LL}^x\bfG_x\ket{\phi_{\HL}^x}
 &=P_{\LL}^x\bfG_xP_{\HL}^x\ket{\phi_{\HL}^x}
 \notag\\
 &=\frac14\Bigl[
 P_{\LL}^xP_{\HL}^x
 +P_{\LL}^x(\bfX^x\otimes I_{\regC})P_{\HL}^x
 \notag\\
 &\hspace{7em}
 +P_{\LL}^x(I_{\regB}\otimes\bfY^x)P_{\HL}^x
 +P_{\LL}^x\bfZ P_{\HL}^x
 \Bigr]\ket{\phi_{\HL}^x}
 \notag\\
 &=\frac14\Bigl[
 \underbrace{(N_x\Pi_x)\otimes O_x}_{=\,0}
 +\underbrace{(N_x\bfX^x\Pi_x)\otimes O_x}_{=\,0}
 \notag\\
 &\hspace{7em}
 +\underbrace{(N_x\Pi_x)\otimes(O_x\bfY^xO_x)}_{=\,0}
 +P_{\LL}^x\bfZ P_{\HL}^x
 \Bigr]\ket{\phi_{\HL}^x}
 \notag\\
 &=\frac14P_{\LL}^x\bfZ\ket{\phi_{\HL}^x}.
 \label{eq:td-HL-LL-Z}
\end{align}
Finally, the block from $\LH$ into $\LL$ expands as
\begin{align}
 P_{\LL}^x\bfG_x\ket{\phi_{\LH}^x}
 &=P_{\LL}^x\bfG_xP_{\LH}^x\ket{\phi_{\LH}^x}
 \notag\\
 &=\frac14\Bigl[
 P_{\LL}^xP_{\LH}^x
 +P_{\LL}^x(\bfX^x\otimes I_{\regC})P_{\LH}^x
 \notag\\
 &\hspace{7em}
 +P_{\LL}^x(I_{\regB}\otimes\bfY^x)P_{\LH}^x
 +P_{\LL}^x\bfZ P_{\LH}^x
 \Bigr]\ket{\phi_{\LH}^x}
 \notag\\
 &=\frac14\Bigl[
 \underbrace{N_x\otimes(O_x\Theta_x)}_{=\,0}
 +\underbrace{(N_x\bfX^xN_x)\otimes(O_x\Theta_x)}_{=\,0}
 \notag\\
 &\hspace{7em}
 +\underbrace{N_x\otimes(O_x\bfY^x\Theta_x)}_{=\,0}
 +P_{\LL}^x\bfZ P_{\LH}^x
 \Bigr]\ket{\phi_{\LH}^x}
 \notag\\
 &=\frac14P_{\LL}^x\bfZ\ket{\phi_{\LH}^x}.
 \label{eq:td-LH-LL-Z}
\end{align}
Using the definition of $M_x$ in~\eqref{eq:td-claim24-Mx}, these three
identities rewrite the full scalar $c_x$ as
\begin{align}
 c_x
 &=\frac14\bra{\phi_{\HL}^x}\bfZ\ket{\phi_{\LH}^x}
 \notag\\
 &\quad+\frac1{16}\bra{\phi_{\HL}^x}
 \bfZ P_{\LL}^x M_x^{-1}P_{\LL}^x\bfZ
 \ket{\phi_{\LH}^x}.
 \label{eq:td-cx-inverse-form}
\end{align}

We next approximate $M_x^{-1}$ by an inverse involving only $\bfZ$.
On $\operatorname{ran}(P_{\LL}^x)$, introduce
\begin{align}
 Q_x&:=P_{\LL}^x\bfZ P_{\LL}^x,\notag\\
 R_x&:=P_{\LL}^x\bigl(\bfX^x\otimes I_{\regC}
              +I_{\regB}\otimes\bfY^x\bigr)P_{\LL}^x,\notag\\
 A_x&:=(1+4\eta)P_{\LL}^x-Q_x.
 \label{eq:td-claim25-QRA}
\end{align}
Here and below the identity on $\regE$ is suppressed.  Multiplying the
definition of $M_x$ by four and using~\eqref{eq:td-Gx} gives the exact
operator identity
\begin{align}
 4M_x
 &=P_{\LL}^x
 \Bigl[(2+4\eta)I_{\regB\regC}-4\bfG_x\Bigr]
 P_{\LL}^x
 \notag\\
 &=P_{\LL}^x
 \Bigl[
  (2+4\eta)I_{\regB\regC}
  -\bigl(
    I_{\regB\regC}
    +\bfX^x\otimes I_{\regC}
    +I_{\regB}\otimes\bfY^x
    +\bfZ
   \bigr)
 \Bigr]
 P_{\LL}^x
 \notag\\
 &=(1+4\eta)P_{\LL}^x
   -P_{\LL}^x\bfZ P_{\LL}^x
 \notag\\
 &\quad
   -P_{\LL}^x
    \bigl(
      \bfX^x\otimes I_{\regC}
      +I_{\regB}\otimes\bfY^x
    \bigr)
    P_{\LL}^x
 \notag\\
 &=(1+4\eta)P_{\LL}^x-Q_x-R_x
 =A_x-R_x.
 \label{eq:td-inverse-expanded}
\end{align}

Since $\Delta_{\bob}^r$ and $\Delta_{\charlie}^r$ are contractions and
$\bfZ$ is Hermitian, we have
\begin{align}
 \norm{\bfZ}_\infty
 &\le\expect_r
   \norm{\Delta_{\bob}^r\otimes\Delta_{\charlie}^r}_\infty
 =\expect_r\!\left[
   \norm{\Delta_{\bob}^r}_\infty
   \norm{\Delta_{\charlie}^r}_\infty
  \right]
 \le1
 \notag\\
 &\Longrightarrow
 -P_{\LL}^x
 \preceq
 Q_x=P_{\LL}^x\bfZ P_{\LL}^x
 \preceq
 P_{\LL}^x,
 \notag\\
 A_x
 &=(1+4\eta)P_{\LL}^x-Q_x
 \succeq
 (1+4\eta)P_{\LL}^x-P_{\LL}^x
 =4\eta P_{\LL}^x,
 \notag\\
 0\preceq A_x^{-1}
 &\preceq\frac1{4\eta}P_{\LL}^x,
 \qquad
 \norm{A_x^{-1}}_\infty\le\frac1{4\eta}.
 \label{eq:td-claim25-A-margin}
\end{align}
By the definition of the low spectral projectors,
$\norm{N_x\bfX^xN_x}_\infty\le\tau$ and
$\norm{O_x\bfY^xO_x}_\infty\le\tau$.  Since
$P_{\LL}^x=N_x\otimes O_x$ (with the identity on $\regE$ suppressed),
the two compressed operators satisfy
\begin{align}
 P_{\LL}^x(\bfX^x\otimes I_{\regC})P_{\LL}^x
 &=(N_x\bfX^xN_x)\otimes O_x,
 \notag\\
 \norm{P_{\LL}^x(\bfX^x\otimes I_{\regC})P_{\LL}^x}_\infty
 &=\norm{N_x\bfX^xN_x}_\infty\norm{O_x}_\infty
 \le\tau=\eta^3,
 \notag\\
 P_{\LL}^x(I_{\regB}\otimes\bfY^x)P_{\LL}^x
 &=N_x\otimes(O_x\bfY^xO_x),
 \notag\\
 \norm{P_{\LL}^x(I_{\regB}\otimes\bfY^x)P_{\LL}^x}_\infty
 &=\norm{N_x}_\infty\norm{O_x\bfY^xO_x}_\infty
 \le\tau=\eta^3,
 \label{eq:td-low-XY}
\end{align}
so
\begin{equation}
 \norm{R_x}_\infty\le2\eta^3.
 \label{eq:td-claim25-R-bound}
\end{equation}
Finally,~\eqref{eq:td-LL-operator-margin-expanded} and
~\eqref{eq:td-inverse-expanded} imply
\begin{equation}
 A_x-R_x=4M_x\succeq(4\eta-2\eta^3)P_{\LL}^x,
 \qquad
 \norm{(A_x-R_x)^{-1}}_\infty
 \le\frac1{4\eta-2\eta^3}.
 \label{eq:td-claim25-M-margin}
\end{equation}

\medskip
\begin{itshape}
\noindent\textbf{Remark (the resolvent identity used here).}
Let $\mathcal H_x:=\operatorname{ran}(P_{\LL}^x)$, whose identity
operator is the restriction of $P_{\LL}^x$.  For an operator $T$ on
$\mathcal H_x$ and a complex number $z$ for which
$T-zI_{\mathcal H_x}$ is invertible, use the convention
\[
 \mathcal R_T(z):=(T-zI_{\mathcal H_x})^{-1}.
\]
The second resolvent identity in this convention is
standard; see, for example,~\cite[Lemma~6.5, Eq.~(6.4)]{Teschl2014}:
\[
 \mathcal R_S(z)-\mathcal R_T(z)
 =\mathcal R_S(z)(T-S)\mathcal R_T(z).
\]
The bounds in~\eqref{eq:td-claim25-A-margin} and
~\eqref{eq:td-claim25-M-margin} show that $A_x$ and $A_x-R_x$ are
invertible on $\mathcal H_x$.  Thus
\[
 A_x^{-1}=\mathcal R_{A_x}(0),
 \qquad
 (A_x-R_x)^{-1}=\mathcal R_{A_x-R_x}(0).
\]
Taking $S=A_x-R_x$, $T=A_x$, and $z=0$ in the second resolvent
identity gives
\[
 (A_x-R_x)^{-1}-A_x^{-1}
 =(A_x-R_x)^{-1}R_xA_x^{-1}.
\]
Thus the two inverse operators to which the second resolvent identity is
applied, $A_x^{-1}$ and $(A_x-R_x)^{-1}$, are resolvents at zero.
Neither $A_x$ nor $R_x$ is itself a resolvent; $R_x$ is the perturbation
relating $A_x-R_x$ to $A_x$.
\end{itshape}
\medskip

Using~\eqref{eq:td-inverse-expanded} together with this identity gives
\begin{align}
 M_x^{-1}-4A_x^{-1}
 &=4\bigl((A_x-R_x)^{-1}-A_x^{-1}\bigr)\notag\\
 &=4(A_x-R_x)^{-1}R_xA_x^{-1}.
 \label{eq:td-claim25-resolvent}
\end{align}
Taking norms and substituting
~\eqref{eq:td-claim25-A-margin}--\eqref{eq:td-claim25-M-margin},
\begin{equation}
 \norm{M_x^{-1}-4A_x^{-1}}_\infty
 \le
 \frac{8\eta^3}{(4\eta)(4\eta-2\eta^3)}
 =\frac{\eta}{2-\eta^2}
 \le\eta,
 \label{eq:td-inverse-error}
\end{equation}
where the last inequality holds in particular for $0<\eta\le1/4$.

Define the scalar obtained by replacing $M_x^{-1}$ with $4A_x^{-1}$:
\begin{align}
 \widetilde c_x
 &:={}
 \frac14\bra{\phi_{\HL}^x}\bfZ\ket{\phi_{\LH}^x}
 \notag\\
 &\quad+
 \frac14\bra{\phi_{\HL}^x}
 \bfZ P_{\LL}^x A_x^{-1}P_{\LL}^x\bfZ
 \ket{\phi_{\LH}^x}.
 \label{eq:td-claim25-c-tilde}
\end{align}
Equations~\eqref{eq:td-cx-inverse-form} and~\eqref{eq:td-inverse-error},
together with $\norm{\bfZ}_\infty\le1$, give the pointwise estimate
\begin{align}
 |c_x-\widetilde c_x|
 &\le\frac1{16}
 \norm{\ket{\phi_{\HL}^x}}
 \norm{M_x^{-1}-4A_x^{-1}}_\infty
 \norm{\ket{\phi_{\LH}^x}}\notag\\
 &\le\frac{\eta}{16}
 \norm{\ket{\phi_{\HL}^x}}
 \norm{\ket{\phi_{\LH}^x}}.
 \label{eq:td-claim25-pointwise-error}
\end{align}
Averaging and applying Cauchy--Schwarz therefore yields
\begin{equation}
 \left|\expect_x(c_x-\widetilde c_x)\right|
 \le\frac{\eta}{16}\sqrt{W_{\HL}W_{\LH}}.
 \label{eq:td-claim25-average-error}
\end{equation}

\medskip
\begin{itshape}
\noindent\textbf{Remark (Neumann expansion).}
The Neumann expansion is the operator analogue of the scalar geometric
series.  If $B$ is an operator with $\norm{B}_\infty<1$, then
\[
 (I-B)^{-1}=\sum_{k=0}^{\infty}B^k.
\]
Indeed, its finite partial sums satisfy
\[
 (I-B)\sum_{k=0}^{m}B^k=I-B^{m+1},
\]
and $\norm{B^{m+1}}_\infty\le\norm{B}_\infty^{m+1}\to0$.  Thus the
series converges in operator norm and its limit is the inverse of
$I-B$.  In our setting the underlying space is
$\operatorname{ran}(P_{\LL}^x)$, its identity is $P_{\LL}^x$, and
\[
 B=\frac{Q_x}{1+4\eta}.
\]
The first line of~\eqref{eq:td-geometric-inverse} verifies the required
condition $\norm{B}_\infty<1$.  See, for example, Horn and
Johnson~\cite[p.~365]{HJ12}.
\end{itshape}
\medskip

It remains to expand $A_x^{-1}$.  On $\operatorname{ran}(P_{\LL}^x)$,
the projection $P_{\LL}^x$ is the identity.  The norm condition and the
resulting Neumann expansion are
\begin{align}
 \norm{\frac{Q_x}{1+4\eta}}_\infty
 &\le\frac{\norm{Q_x}_\infty}{1+4\eta}
 \le\frac{\norm{\bfZ}_\infty}{1+4\eta}
 \le\frac1{1+4\eta}
 <1,
 \notag\\
 A_x^{-1}
 &=\Bigl[(1+4\eta)P_{\LL}^x-Q_x\Bigr]^{-1}
 \notag\\
 &=\frac1{1+4\eta}
   \left[P_{\LL}^x-\frac{Q_x}{1+4\eta}\right]^{-1}
 \notag\\
 &=\frac1{1+4\eta}
   \sum_{k=0}^{\infty}
   \left(\frac{Q_x}{1+4\eta}\right)^k
 \notag\\
 &=\frac{P_{\LL}^x}{1+4\eta}
   +\frac{Q_x}{(1+4\eta)^2}
   +\frac{Q_x^2}{(1+4\eta)^3}
   +\cdots
 \notag\\
 &=\sum_{k=0}^{\infty}(1+4\eta)^{-(k+1)}Q_x^k.
 \label{eq:td-geometric-inverse}
\end{align}
Substituting~\eqref{eq:td-geometric-inverse} into
~\eqref{eq:td-claim25-c-tilde}, and recalling
$Q_x=P_{\LL}^x\bfZ P_{\LL}^x$, shows that
\begin{align}
 \expect_x\widetilde c_x
 &=\frac14\left[
   \Gamma_1+
   \sum_{k=0}^{\infty}(1+4\eta)^{-(k+1)}\Gamma_{k+2}
  \right].
 \label{eq:td-claim25-average-series}
\end{align}
The exchange of the expectation and the series is justified by norm
convergence; more explicitly, the absolute value of the $k$th averaged
matrix element, including its coefficient, is at most
$\sqrt{W_{\HL}W_{\LH}}(1+4\eta)^{-(k+1)}$, whose sum is finite.
Combining~\eqref{eq:td-claim25-average-error} and
~\eqref{eq:td-claim25-average-series} proves
~\eqref{eq:td-series-approximation}.
\end{proof}

The scalar $\Gamma_m$ has a direct algebraic interpretation.  Starting with the $\LH$ component, one applies $m$ copies of $\bfZ$.  After each of the first $m-1$ copies, one keeps only the $\LL$ component; after the final copy, one takes the overlap with the $\HL$ component.  The projections in~\eqref{eq:td-Gammam} only select vector components in the calculation.  They are not measurements performed in the original prediction protocol.

\subsection{Each common-challenge contribution defines an extractor}
\label{sec:nlsd-witness}

We now prove the estimate on $\Gamma_m$ that is needed to sum the geometric series.  Expanding the $m$ copies of $\bfZ$ produces an average over a challenge sequence
$w=(r_1,\ldots,r_m)$.  For each fixed sequence, the local products that appear on Bob's and Charlie's sides can be normalized into Kraus operators of valid local measurements whose outcome is a candidate string $x$.  If the corresponding scalar contribution were large, these measurements would jointly recover the full string with probability larger than $p_{\srch}$.

\begin{lemma}[Every common-challenge contribution is controlled by full-string extraction]
\label{lem:td-gamma-search}
For every integer $m\ge1$,
\begin{equation}
 |\Gamma_m|
 \le\frac{(m+1)\sqrt{p_{\srch}}}{\tau^2}.
 \label{eq:td-Gamma-bound}
\end{equation}
\end{lemma}

\begin{proof}
Let $r_1,\ldots,r_m$ be independent and uniform in $\ftwo^n$, and write
$w=(r_1,\ldots,r_m)$ and $\expect_w=\expect_{r_1,\ldots,r_m}$.
For a fixed word $w$, define
\begin{align}
 V_{x,w}^{B}
 &:={\Pi_x\Delta_{\bob}^{r_m}N_x
 \Delta_{\bob}^{r_{m-1}}N_x\cdots
 \Delta_{\bob}^{r_1}N_x},
 \label{eq:td-VB}\\
 V_{x,w}^{C}
 &:={O_x\Delta_{\charlie}^{r_m}O_x
 \Delta_{\charlie}^{r_{m-1}}O_x\cdots
 \Delta_{\charlie}^{r_1}\Theta_x}.
 \label{eq:td-VC}
\end{align}
For $m=1$, these mean
$V_{x,(r_1)}^B=\Pi_x\Delta_{\bob}^{r_1}N_x$ and
$V_{x,(r_1)}^C=O_x\Delta_{\charlie}^{r_1}\Theta_x$.

We verify explicitly how these products arise.  Because we are in the
case~\eqref{eq:td-hard-case-weights}, the retained $\HL$ and $\LH$
vectors equal the corresponding components of $\ket{\psi^x}$.  For
$m\ge2$, expand each copy of
$\bfZ=\expect_r[\Delta_{\bob}^r\otimes\Delta_{\charlie}^r]$ in
~\eqref{eq:td-Gammam} and insert the endpoint projections.  This gives
\begin{align}
 \Gamma_m
 &=\expect_w\expect_x\bra{\psi^x}
 P_{\HL}^x
 (\Delta_{\bob}^{r_m}\otimes\Delta_{\charlie}^{r_m})
 P_{\LL}^x\cdots P_{\LL}^x
 (\Delta_{\bob}^{r_1}\otimes\Delta_{\charlie}^{r_1})
 P_{\LH}^x
 \ket{\psi^x}.
 \label{eq:td-Gamma-projected-word}
\end{align}
For $m=1$, the same formula has no $P_{\LL}^x$ and follows from
~\eqref{eq:td-Gamma1}.  Substituting the tensor-product formulas for
$P_{\HL}^x,P_{\LH}^x,P_{\LL}^x$ and collecting the Bob and Charlie
factors yields
\begin{equation}
 P_{\HL}^x
 (\Delta_{\bob}^{r_m}\otimes\Delta_{\charlie}^{r_m})
 P_{\LL}^x\cdots P_{\LL}^x
 (\Delta_{\bob}^{r_1}\otimes\Delta_{\charlie}^{r_1})
 P_{\LH}^x
 =V_{x,w}^{B}\otimes V_{x,w}^{C}\otimes I_{\regE}.
 \label{eq:td-Gamma-factorization}
\end{equation}
Reading from right to left confirms the order: Bob starts low, remains
low after each intermediate challenge operator, and ends high; Charlie
starts high, enters the low subspace after the first challenge, and
remains low.  Therefore~\eqref{eq:td-Gamma-projected-word} becomes
\begin{equation}
 \Gamma_m
 =\expect_w\expect_x
 \bra{\psi^x}
 (V_{x,w}^{B}\otimes V_{x,w}^{C}\otimes I_{\regE})
 \ket{\psi^x}.
 \label{eq:td-Gamma-fixed-word}
\end{equation}
The maps $V_{x,w}^{B},V_{x,w}^{C}$ need not be positive and are not
themselves POVM effects.  Their adjoint-products will become effects
after the normalization below.

We next show that the maps can be normalized into measurements.  For every fixed $w$,
\begin{equation}
 \sum_x(V_{x,w}^{B})^\dagger V_{x,w}^{B}
 \preceq(m+1)^2\tau^{-2}I_{\regB}.
 \label{eq:td-Bob-sum}
\end{equation}
To prove this, fix a unit vector $|\zeta\rangle$.  Repeatedly insert $N_x=I-\Pi_x$ to obtain
\begin{align}
 &\Delta_{\bob}^{r_m}N_x\cdots\Delta_{\bob}^{r_1}N_x
 -\Delta_{\bob}^{r_m}\cdots\Delta_{\bob}^{r_1}
 \notag\\
 &\quad=-\sum_{j=1}^{m}
 \Delta_{\bob}^{r_m}N_x\cdots
 \Delta_{\bob}^{r_{j+1}}N_x
 \Delta_{\bob}^{r_j}\Pi_x
 \Delta_{\bob}^{r_{j-1}}\cdots\Delta_{\bob}^{r_1}.
 \label{eq:td-telescope}
\end{align}
Multiplying~\eqref{eq:td-telescope} on the left by $\Pi_x$ gives the
operator identity
\begin{align}
 V_{x,w}^{B}
 &=\Pi_x\Delta_{\bob}^{r_m}\cdots\Delta_{\bob}^{r_1}
 \notag\\
 &\quad-\sum_{j=1}^{m}
 \Pi_x\Delta_{\bob}^{r_m}N_x\cdots
 \Delta_{\bob}^{r_{j+1}}N_x
 \Delta_{\bob}^{r_j}\Pi_x
 \Delta_{\bob}^{r_{j-1}}\cdots\Delta_{\bob}^{r_1}.
 \label{eq:td-Bob-telescope-after-Pi}
\end{align}
For a unit vector $\ket{\zeta}$, define
\begin{equation}
 \ket{\zeta_0}
 :=\Delta_{\bob}^{r_m}\cdots\Delta_{\bob}^{r_1}\ket{\zeta},
 \qquad
 \ket{\zeta_j}
 :=\Delta_{\bob}^{r_{j-1}}\cdots\Delta_{\bob}^{r_1}\ket{\zeta}
 \quad(1\le j\le m),
 \label{eq:td-zeta-vectors}
\end{equation}
where the product defining $\ket{\zeta_1}$ is empty and therefore equals
the identity.  These vectors do not depend on $x$, and each has norm at
most one.  In the $j$th summand of
~\eqref{eq:td-Bob-telescope-after-Pi}, every operator to the left of the
displayed inner copy of $\Pi_x$ is a contraction.  The triangle
inequality therefore gives
\begin{equation}
 \norm{V_{x,w}^{B}\ket{\zeta}}
 \le\sum_{j=0}^{m}\norm{\Pi_x\ket{\zeta_j}},
 \label{eq:td-pointwise-VB}
\end{equation}

\medskip
\begin{itshape}
\noindent\textbf{Remark (Minkowski's inequality).}
For any nonnegative numbers $a_{j,x}$, Minkowski's inequality for the
$\ell_2$ norm states that
\[
 \left(\sum_x\left(\sum_{j=0}^{m}a_{j,x}\right)^2\right)^{1/2}
 \le
 \sum_{j=0}^{m}\left(\sum_x a_{j,x}^2\right)^{1/2}.
\]
This is simply the triangle inequality in Euclidean space.  For each $j$,
regard $a_j=(a_{j,x})_x$ as a vector indexed by $x$; the displayed inequality
is $\norm{\sum_j a_j}_2\le\sum_j\norm{a_j}_2$.  In our application,
$a_{j,x}=\norm{\Pi_x\ket{\zeta_j}}$.  Equation~\eqref{eq:td-pointwise-VB}
first bounds each $\norm{V_{x,w}^{B}\ket{\zeta}}$ by
$\sum_j a_{j,x}$; monotonicity of the $\ell_2$ norm and the displayed
inequality then give the first line below.
\end{itshape}
\medskip

Taking the $\ell_2$ norm over $x$, applying Minkowski's inequality, and
using~\eqref{eq:td-threshold-sums},
\begin{align}
 \left(\sum_x\norm{V_{x,w}^{B}\ket{\zeta}}^2\right)^{1/2}
 &\le\sum_{j=0}^{m}
 \left(\sum_x\norm{\Pi_x\ket{\zeta_j}}^2\right)^{1/2}
 \notag\\
 &=\sum_{j=0}^{m}
 \left(
  \bra{\zeta_j}\left(\sum_x\Pi_x\right)\ket{\zeta_j}
 \right)^{1/2}
 \notag\\
 &\le\frac{m+1}{\tau}.
 \label{eq:td-Bob-Minkowski}
\end{align}
Squaring~\eqref{eq:td-Bob-Minkowski} shows that every unit vector
$\ket{\zeta}$ satisfies
\begin{align}
 \bra{\zeta}
 \left(\sum_x(V_{x,w}^{B})^\dagger V_{x,w}^{B}\right)
 \ket{\zeta}
 &=\sum_x\norm{V_{x,w}^{B}\ket{\zeta}}^2
 \le\frac{(m+1)^2}{\tau^2}.
 \label{eq:td-Bob-quadratic-form}
\end{align}
The quadratic-form characterization of positive-semidefinite order now
proves~\eqref{eq:td-Bob-sum}.

For Charlie, write
\begin{equation}
 V_{x,w}^{C}=K_{x,w}^{C}\Theta_x,
 \qquad
 K_{x,w}^{C}
 :=O_x\Delta_{\charlie}^{r_m}O_x
 \Delta_{\charlie}^{r_{m-1}}O_x\cdots
 \Delta_{\charlie}^{r_1}.
 \label{eq:td-Charlie-K}
\end{equation}
Every factor in $K_{x,w}^{C}$ is a contraction, so
$\norm{K_{x,w}^{C}}_\infty\le1$.  Hence
\begin{align}
 (V_{x,w}^{C})^\dagger V_{x,w}^{C}
 &=\Theta_x(K_{x,w}^{C})^\dagger K_{x,w}^{C}\Theta_x
 \preceq\Theta_x.
 \label{eq:td-Charlie-effect}
\end{align}
Summing this inequality and invoking~\eqref{eq:td-threshold-sums} gives
\begin{equation}
 \sum_x(V_{x,w}^{C})^\dagger V_{x,w}^{C}
 \preceq\tau^{-2}I_{\regC}.
 \label{eq:td-Charlie-sum}
\end{equation}
For this fixed $w$, define effects indexed by candidate output strings:
\begin{equation}
 F_{\bob,w}^x
 :=\frac{\tau^2}{(m+1)^2}
 (V_{x,w}^{B})^\dagger V_{x,w}^{B},
 \qquad
 F_{\charlie,w}^x
 :=\tau^2(V_{x,w}^{C})^\dagger V_{x,w}^{C}.
 \label{eq:td-fixed-word-effects}
\end{equation}
Equations~\eqref{eq:td-Bob-sum} and~\eqref{eq:td-Charlie-sum}
imply
\begin{equation}
 \sum_xF_{\bob,w}^x\preceq I_{\regB},
 \qquad
 \sum_xF_{\charlie,w}^x\preceq I_{\regC},
 \label{eq:td-fixed-word-subpovm}
\end{equation}
so the two families are local sub-POVMs.  Fix any string $x_\star$.
Adding the residual effects
\begin{equation}
 I_{\regB}-\sum_xF_{\bob,w}^x,
 \qquad
 I_{\regC}-\sum_xF_{\charlie,w}^x
 \label{eq:td-fixed-word-residuals}
\end{equation}
to the outcomes labeled $x_\star$ completes the families to full POVMs.
Every completed effect dominates the corresponding sub-POVM effect, so
this completion cannot decrease their joint correct-label
probability.

The completed measurements may depend on the fixed auxiliary word $w$,
which is hardwired independently of the hidden string.  This is allowed
because $p_{\srch}$ upper-bounds every pair of local full-string POVMs.  The index $x$ in~\eqref{eq:td-fixed-word-effects} is
only a candidate output label: the constructed measurement is not given
the true hidden string.  Applying the definition of $p_{\srch}$ to these
completed POVMs and then using the subeffects gives
\begin{align}
 p_{\srch}
 &\ge\expect_x
 \Tr\left[
  (F_{\bob,w}^x\otimes F_{\charlie,w}^x)\rho_x
 \right]
 \notag\\
 &=\frac{\tau^4}{(m+1)^2}\expect_x
 \Tr\left[
  \bigl(
   (V_{x,w}^{B})^\dagger V_{x,w}^{B}
   \otimes
   (V_{x,w}^{C})^\dagger V_{x,w}^{C}
  \bigr)\rho_x
 \right]
 \notag\\
 &=\frac{\tau^4}{(m+1)^2}\expect_x
 \norm{
  (V_{x,w}^{B}\otimes V_{x,w}^{C}\otimes I_{\regE})
  \ket{\psi^x}
 }^2.
 \label{eq:td-fixed-word-success}
\end{align}
Rearranging~\eqref{eq:td-fixed-word-success} yields
\begin{equation}
 \expect_x
 \norm{
 (V_{x,w}^{B}\otimes V_{x,w}^{C}\otimes I_{\regE})
 \ket{\psi^x}
 }^2
 \le\frac{(m+1)^2p_{\srch}}{\tau^4}.
 \label{eq:td-fixed-word-square}
\end{equation}
For this fixed word, put
\begin{equation}
 s_w:=\expect_x\bra{\psi^x}
 (V_{x,w}^{B}\otimes V_{x,w}^{C}\otimes I_{\regE})
 \ket{\psi^x}.
 \label{eq:td-fixed-word-scalar}
\end{equation}
Since every $\ket{\psi^x}$ is a unit vector, the triangle inequality,
Cauchy--Schwarz for each inner product, and Cauchy--Schwarz over $x$
give
\begin{align}
 |s_w|
 &\le\expect_x\left|
 \bra{\psi^x}
 (V_{x,w}^{B}\otimes V_{x,w}^{C}\otimes I_{\regE})
 \ket{\psi^x}
 \right|
 \notag\\
 &\le\expect_x
 \norm{
  (V_{x,w}^{B}\otimes V_{x,w}^{C}\otimes I_{\regE})
  \ket{\psi^x}
 }
 \notag\\
 &\le
 \left(
  \expect_x
  \norm{
   (V_{x,w}^{B}\otimes V_{x,w}^{C}\otimes I_{\regE})
   \ket{\psi^x}
  }^2
 \right)^{1/2}
 \notag\\
 &\le\frac{(m+1)\sqrt{p_{\srch}}}{\tau^2},
 \label{eq:td-fixed-word-amplitude}
\end{align}
where the last step is~\eqref{eq:td-fixed-word-square}.  Finally,
~\eqref{eq:td-Gamma-fixed-word} and the triangle inequality over the
uniform word $w$ yield
\begin{equation}
 |\Gamma_m|
 =\left|\expect_ws_w\right|
 \le\expect_w|s_w|
 \le\frac{(m+1)\sqrt{p_{\srch}}}{\tau^2}.
 \label{eq:td-Gamma-final-average}
\end{equation}
This is~\eqref{eq:td-Gamma-bound}.
\end{proof}

\subsection{Complete the top-down estimate}
\label{sec:nlsd-complete}

We now return to the single remaining target~\eqref{eq:td-cx-target}.
Combining~\eqref{eq:td-series-approximation} and~\eqref{eq:td-Gamma-bound},
\begin{align}
 \left|\expect_xc_x\right|
 &\le
 \frac14\left[
  2+\sum_{k=0}^{\infty}(k+3)(1+4\eta)^{-(k+1)}
 \right]
 \frac{\sqrt{p_{\srch}}}{\tau^2}
 \notag\\
 &\qquad
 +\frac{\eta}{16}\sqrt{W_{\HL}W_{\LH}}.
 \label{eq:td-cx-before-sum}
\end{align}
The geometric-series identities
\[
 \sum_{k=0}^{\infty}(1+4\eta)^{-(k+1)}=\frac1{4\eta},
 \qquad
 \sum_{k=0}^{\infty}(k+1)(1+4\eta)^{-(k+1)}
 =\frac{1+4\eta}{16\eta^2}
\]
give
\begin{equation}
 \frac14\left[
  2+\sum_{k=0}^{\infty}(k+3)(1+4\eta)^{-(k+1)}
 \right]
 =\frac12+\frac1{64\eta^2}+\frac3{16\eta}
 \le\frac3{32\eta^2},
 \label{eq:td-geometric-bound}
\end{equation}
where the final inequality uses $0<\eta\le1/4$.
Thus
\begin{equation}
 \left|\expect_xc_x\right|
 \le
 \frac{3\sqrt{p_{\srch}}}{32\tau^2\eta^2}
 +\frac{\eta}{16}\sqrt{W_{\HL}W_{\LH}}.
 \label{eq:td-cx-raw}
\end{equation}
Using~\eqref{eq:td-hard-case-weights},~\eqref{eq:td-parameters}, and~\eqref{eq:td-temporary-assumption},
\begin{align}
 \frac{3\sqrt{p_{\srch}}}
 {32\tau^2\eta^2\sqrt{W_{\HL}W_{\LH}}}
 &\le
 \frac{3\sqrt{p_{\srch}}}{32\tau^2\eta^2\delta}
 \le\frac{3\eta}{32}.
 \label{eq:td-first-constant}
\end{align}
Therefore
\begin{equation}
 \frac{|\expect_xc_x|}{\sqrt{W_{\HL}W_{\LH}}}
 \le\frac{3\eta}{32}+\frac{\eta}{16}
 =\frac{5\eta}{32}.
 \label{eq:td-cx-ratio}
\end{equation}
Since
\begin{equation}
 \eta-\frac{\eta^3}{2}\ge\frac{31\eta}{32},
 \label{eq:td-margin}
\end{equation}
the desired bound~\eqref{eq:td-cx-target} holds with room to spare.
The reduction~\eqref{eq:td-main-reduction} now gives
\begin{equation}
 \expect_x\bra{\phi^x}\bfG_x\ket{\phi^x}
 \le
 \left(\frac12+\eta\right)
 \expect_x\norm{\ket{\phi^x}}^2
 \le\frac12+\eta.
 \label{eq:td-retained-finished}
\end{equation}
This conclusion also held in the easy case~\eqref{eq:td-easy-case}.  Therefore it is valid in all cases.

Returning to~\eqref{eq:td-deletion-bound},
\begin{equation}
 p_{\pred}\le\frac12+\eta+2\sqrt{q_{\mathrm{omit}}}+q_{\mathrm{omit}}.
 \label{eq:td-p-with-q}
\end{equation}
By~\eqref{eq:td-q-bound},~\eqref{eq:td-parameters}, and~\eqref{eq:td-temporary-assumption},
\begin{align}
 q_{\mathrm{omit}}
 &\le\frac{\eta^{22}}{3072^2\eta^{12}}
 +\frac{3\eta^2}{3072}
 \notag\\
 &=\eta^2\left(\frac{\eta^8}{3072^2}+\frac1{1024}\right)
 <\frac{\eta^2}{1000}.
 \label{eq:td-q-numerical}
\end{align}
Hence
\begin{equation}
 2\sqrt{q_{\mathrm{omit}}}+q_{\mathrm{omit}}
 <\frac{2\eta}{\sqrt{1000}}+\frac{\eta^2}{1000}
 <0.064\eta.
 \label{eq:td-q-cost}
\end{equation}
Substituting into~\eqref{eq:td-p-with-q} proves~\eqref{eq:td-conditional-goal}.

Finally, since $p_{\srch}\ge2^{-n}>0$, suppose first that
$p_{\srch}^{1/22}<1/10$.  Choose
\begin{equation}
 \eta:=\frac52p_{\srch}^{1/22}.
 \label{eq:td-final-eta}
\end{equation}
Then $\eta<1/4$, and~\eqref{eq:td-temporary-assumption} holds because
\[
 \frac{\eta^{22}}{3072^2}
 =p_{\srch}\frac{(5/2)^{22}}{3072^2}
 \ge p_{\srch}.
\]
Thus
\[
 p_{\pred}
 <\frac12+1.064\cdot\frac52p_{\srch}^{1/22}
 <\frac12+5p_{\srch}^{1/22}.
\]
If $p_{\srch}^{1/22}\ge1/10$, then
$1/2+5p_{\srch}^{1/22}\ge1$, and the theorem follows from $p_{\pred}\le1$.
This proves~\eqref{eqn:thm:ftwo:main-bound}.  The asymptotic consequence
follows because taking a fixed positive root preserves negligibility.
\end{proof}

\section{Extensions over arbitrary finite fields}
\label{sec:qary-extensions}

This section turns to a prime-power alphabet.  After fixing the
common-challenge model and its normalization, we prove a uniform cubic theorem
for independent challenges over every finite field and derive from it a
common-challenge consequence, including an efficient corollary for growing
fields.

\subsection{Model and normalization}
\label{sec:qary-model}

\begin{definition}[Common-challenge prediction and search]
\label{def:qary-common-model}
Let $q$ be a prime power, let $\ell\geq1$, and put
\[
 \cX_q:=\mathbb F_q^\ell,
 \qquad N_q:=|\cX_q|=q^\ell.
\]
The hidden label $x$ is uniform in $\cX_q$, and Bob and Charlie hold an
arbitrary state $\rho_x\in\density(\regB\otimes\regC)$.  On a common uniform
challenge $r\in\cX_q$, they use local $q$-outcome POVMs
\[
 \{P_a^r:a\in\mathbb F_q\},
 \qquad
 \{Q_a^r:a\in\mathbb F_q\}.
\]
Here $\langle r,x\rangle:=\sum_{i=1}^{\ell}r_ix_i$ is evaluated in
$\mathbb F_q$.
Their exact common-challenge prediction probability is
\begin{equation}
 p_{\mathrm{com}}
 :=\mathbb E_{x,r}\operatorname{Tr}\!\left[
  \left(
   P^{r}_{\langle r,x\rangle}\otimes
   Q^{r}_{\langle r,x\rangle}
  \right)\rho_x
 \right].
 \label{eq:qary-common-exact}
\end{equation}
Their optimal simultaneous full-string recovery probability is
\begin{equation}
 p_{\mathrm{search}}^\star
 :=\sup_{\{E_x\},\{F_x\}}
 \mathbb E_x\operatorname{Tr}\!\left[
  (E_x\otimes F_x)\rho_x
 \right],
 \label{eq:qary-common-search}
\end{equation}
where the supremum is over local POVMs indexed by $x\in\cX_q$.
\end{definition}

The definition samples $r$ from all of $\cX_q$.  If
$p_{\mathrm{com}}^{\ne0}$ denotes the score conditioned on $r\ne0$, then
\begin{equation}
 p_{\mathrm{com}}
 \leq \frac1{N_q}+\left(1-\frac1{N_q}\right)p_{\mathrm{com}}^{\ne0}.
 \label{eq:qary-zero-challenge}
\end{equation}
Thus the natural all-challenge baseline obtained from the nonzero-challenge
guessing baseline $1/q$ is
$N_q^{-1}+(1-N_q^{-1})q^{-1}$.  Equation~\eqref{eq:qary-zero-challenge}
relates the two conventions.  The independent-challenge theorem
below uses a more general two-secret model; Section~\ref{sec:qary-common-finite}
then returns to Definition~\ref{def:qary-common-model}.

\subsection{Independent challenges over arbitrary finite fields}
\label{sec:independent-qary}

This section presents a version of the independent-challenge non-local
search-to-decision theorem over arbitrary finite fields. This can be viewed as a two-party extension of the one-party quantum Goldreich-Levin for large finite fields by  Sudo, Tezuka, Hara, and
Yoshida~\cite[Theorem~1]{STHY23}.
\par Let $q=p^t$ be a prime power, and let $\ell\geq 1$.  We consider a
classical--quantum state
\begin{equation}
  \rho_{\regX\regY\regB\regC}
  =
  \sum_{x,y\in\mathbb F_q^\ell}
  \pi(x,y)\,
  \lvert x,y\rangle\!\langle x,y\rvert_{\regX\regY}
  \otimes \rho^{x,y}_{\regB\regC},
  \label{eq:qary-cq-state}
\end{equation}
where the classical registers $\regX,\regY$ store the hidden random
vectors $X^{\mathsf B},X^{\mathsf C}$, respectively, and $\regB,\regC$
are Bob's and Charlie's original quantum shares.
The distribution $\pi$ is arbitrary.  Thus, in particular, the two
hidden strings need not be independent or identically distributed.  Bob and
Charlie receive independent uniform challenges
$V^{\mathsf B},V^{\mathsf C}\in\mathbb F_q^\ell$.  On challenges $v$ and
$w$, respectively, they apply local POVMs
\[
  \mathsf P^v=\{P^v_g:g\in\mathbb F_q\},
  \qquad
  \mathsf Q^w=\{Q^w_h:h\in\mathbb F_q\},
\]
and return $G^{\mathsf B}=g$ and $G^{\mathsf C}=h$.  Define the two
prediction errors by
\begin{equation}
  E^{\mathsf B}=G^{\mathsf B}
       -\langle X^{\mathsf B},V^{\mathsf B}\rangle,
  \qquad
  E^{\mathsf C}=G^{\mathsf C}
       -\langle X^{\mathsf C},V^{\mathsf C}\rangle,
  \label{eq:qary-errors}
\end{equation}
where $\langle x,v\rangle=\sum_{i=1}^{\ell}x_iv_i$ is evaluated in
$\mathbb F_q$.

For later comparison, let $p_{\mathrm{search}}^\star$ denote the optimal
simultaneous full-search probability:
\begin{equation}
  p_{\mathrm{search}}^\star
  :=
  \sup_{\widehat{\mathsf P},\widehat{\mathsf Q}}
  \sum_{x,y}\pi(x,y)
  \operatorname{Tr}\!\left[
    \bigl(\widehat P_x\otimes\widehat Q_y\bigr)
    \rho^{x,y}_{\regB\regC}
  \right],
  \label{eq:qary-optimal-search}
\end{equation}
where the supremum is over local POVMs with outcomes in
$\mathbb F_q^\ell$.

\begin{theorem}[Independent-challenge extraction over $\mathbb F_q$]
\label{thm:independent-qary}
Suppose that, for some $0<\epsilon\leq 1-1/q$,
\begin{equation}
  \Pr\!\left[E^{\mathsf B}+E^{\mathsf C}=0\right]
  \geq \frac{1}{q}+\epsilon.
  \label{eq:qary-relational-advantage}
\end{equation}
Then there are local full-string measurements for Bob and Charlie such that
\begin{equation}
  \Pr\!\left[
    \widehat X^{\mathsf B}=X^{\mathsf B}
    \ \wedge\
    \widehat X^{\mathsf C}=X^{\mathsf C}
  \right]
  \geq
  \frac{q^2}{8(q-1)^2}\,\epsilon^3
  \geq \frac{\epsilon^3}{8}.
  \label{eq:qary-search-lower-bound}
\end{equation}
Equivalently,
\begin{equation}
  \Pr\!\left[E^{\mathsf B}+E^{\mathsf C}=0\right]
  \leq
  \frac{1}{q}
  +2\left(\frac{q-1}{q}\right)^{2/3}
       \bigl(p_{\mathrm{search}}^\star\bigr)^{1/3}
  \leq
  \frac{1}{q}+2\bigl(p_{\mathrm{search}}^\star\bigr)^{1/3}.
  \label{eq:qary-contrapositive}
\end{equation}
The same upper bound applies to the smaller exact-prediction event
\[
  \left\{
    G^{\mathsf B}=\langle X^{\mathsf B},V^{\mathsf B}\rangle
    \ \wedge\
    G^{\mathsf C}=\langle X^{\mathsf C},V^{\mathsf C}\rangle
  \right\}.
\]
\end{theorem}

In particular, for a security-parameter-indexed family, negligible optimal
simultaneous search probability implies that the advantage over $1/q$ in
the relational experiment, and hence in the exact-prediction experiment, is
negligible.  This remains true when $q$ depends on the security parameter,
because the field-size factor in~\eqref{eq:qary-contrapositive} is at most
one.

\begin{proof}
An \emph{additive character} of the finite field $\mathbb F_q$ is a map
$\chi:\mathbb F_q\to\mathbb C$ satisfying
\[
  |\chi(s)|=1,
  \qquad
  \chi(s+u)=\chi(s)\chi(u)
  \qquad (s,u\in\mathbb F_q).
\]
Thus, a character converts addition in the field into multiplication of
complex phases.  Its basic properties include
\[
  \chi(0)=1,
  \qquad
  \chi(-s)=\overline{\chi(s)},
  \qquad
  \chi(s)^p=1.
\]
The last identity follows because $\mathbb F_q$ has characteristic $p$:
adding any field element to itself $p$ times gives zero.  A character is
\emph{nontrivial} if it is not identically one; this does not require
$\chi(s)\ne1$ for every nonzero $s$.

We use a specific nontrivial additive character, namely the canonical
trace-based choice~\cite[Section~2]{ConradGaussJacobi}
\begin{equation}
  \chi(s)
  :=
  \exp\!\left(
    \frac{2\pi i}{p}\,
    \operatorname{Tr}_{\mathbb F_q/\mathbb F_p}(s)
  \right).
  \label{eq:qary-additive-character}
\end{equation}
Here $q=p^t$, and the \emph{field trace} is the additive map
\[
  \operatorname{Tr}_{\mathbb F_q/\mathbb F_p}(s)
  =s+s^p+\cdots+s^{p^{t-1}}
  \in\mathbb F_p.
\]
In the exponential, its value is represented by an integer modulo $p$.
This is a field trace, rather than the trace of a density operator.
When $q=p$, the field trace is the identity, so
$\chi(s)=e^{2\pi i s/p}$; in particular, over $\mathbb F_2$ this becomes
$\chi(s)=(-1)^s$.

For a nontrivial character, the sum of its values over the field is zero.
For $s\ne0$, multiplication by $s$ permutes $\mathbb F_q$, whereas for
$s=0$ every term below equals one.  Consequently, character orthogonality
gives
\begin{equation}
  \sum_{a\in\mathbb F_q}\chi(as)
  =q\,\mathbf 1[s=0]
  \qquad (s\in\mathbb F_q).
  \label{eq:qary-character-orthogonality}
\end{equation}

\paragraph{Step 1: Identifying the good set.}
For a hidden pair $z=(x,y)$ in the support of $\pi$, set
\begin{equation}
  p_z
  :=
  \Pr\!\left[
    E^{\mathsf B}+E^{\mathsf C}=0
    \mathrel{\big|}
    X^{\mathsf B}=x,\ X^{\mathsf C}=y
  \right].
  \label{eq:qary-conditional-relational-success}
\end{equation}
Call $z$ good if $p_z\geq 1/q+\epsilon/2$, and let $\mu$ be the total
$\pi$-mass of the good pairs.  Since $p_z\leq 1$ for every $z$,
\[
  \frac{1}{q}+\epsilon
  \leq
  (1-\mu)\left(\frac{1}{q}+\frac{\epsilon}{2}\right)
  +\mu.
\]
It follows that
\begin{equation}
  \mu
  \geq
  \frac{\epsilon/2}{1-1/q-\epsilon/2}
  \geq \frac{\epsilon}{2}.
  \label{eq:qary-good-mass}
\end{equation}
The good set is used only in the analysis; the extraction measurements do
not need to recognize it.

\paragraph{Step 2: Extraction procedure, parameterized by $a$.}
We use the standard quantum Fourier transform over the additive group
$(\mathbb F_q^\ell,+)$, namely the coordinatewise finite-field transform
of de Beaudrap, Cleve, and Watrous~\cite[Definition~2.1]{BCW02}:
for $u\in\mathbb F_q^\ell$,
\begin{equation}
  F_{q,\ell}\lvert u\rangle
  :=
  q^{-\ell/2}
  \sum_{v\in\mathbb F_q^\ell}
  \chi(\langle u,v\rangle)\lvert v\rangle.
  \label{eq:qary-qft}
\end{equation}
Thus, if a uniform superposition over $v$ is given the phase
$\chi(\langle u,v\rangle)$, applying $F_{q,\ell}^{\dagger}$ returns
$\lvert u\rangle$.

\Needspace{4\baselineskip}
\emph{Registers.}
Each register has a single-letter name, typeset in gray boldface, which
we use throughout the procedure and its analysis.  Bob and Charlie hold
their original quantum shares in $\regB$ and $\regC$, respectively, with
joint state $\rho^{x,y}_{\regB\regC}$.  Their challenge registers are
$\regR$ and $\regS$, respectively, with computational bases indexed by
$\mathbb F_q^\ell$; these same registers will hold the final extraction
outputs.  Their predictor-answer registers are $\regA$ and $\regD$,
respectively, with computational bases indexed by $\mathbb F_q$, and
their finite-dimensional workspace registers are $\regW$ and $\regZ$,
respectively.
The reference register $\regE$ is used only to purify the original state;
neither party acts on it.  The classical registers $\regX,\regY$
in~\eqref{eq:qary-cq-state} record the hidden labels and are not available
to the extractors.  We condition on their values $z=(x,y)$ in the analysis,
so these two classical registers do not appear in the state evolution below.

All six added registers $\regR,\regS,\regA,\regD,\regW,\regZ$ start in
their all-zero state.  Juxtaposition denotes a joint system of the named
registers, not an additional register.  In particular, we write
\[
  \lvert0\rangle_{\regA\regW}
  :=\lvert0\rangle_{\regA}\otimes\lvert0\rangle_{\regW},
  \qquad
  \lvert0\rangle_{\regD\regZ}
  :=\lvert0\rangle_{\regD}\otimes\lvert0\rangle_{\regZ}.
\]
Thus Bob's answer and workspace registers together are $\regA\regW$,
and Charlie's are $\regD\regZ$.

Fix $z=(x,y)$ and a purification $\lvert\psi_z\rangle_{\regB\regC\regE}$
of $\rho^{x,y}_{\regB\regC}$.  For the analysis, denote the initial state
of all registers other than the two challenge registers by
\[
  \lvert\Psi_z^0\rangle
  :=\lvert\psi_z\rangle_{\regB\regC\regE}
    \otimes\lvert0\rangle_{\regA\regW}
    \otimes\lvert0\rangle_{\regD\regZ}.
\]
Tensor products below act on the named registers; local operators are
ordered as $\regB\regA\regW,\regC\regD\regZ,\regE$.  Bob and Charlie
receive their original shares and add their own zero-initialized registers;
they do not need to know or prepare the purification.

\emph{Coherent predictors.}
For each challenge $v$, choose a unitary $U_{\mathsf B}^v$ on
$\regB\regA\regW$ whose answer-register measurement implements
$\mathsf P^v$.  Explicitly, its POVM effects obey
\[
  P_g^v
  =(I_{\regB}\otimes\langle0\rvert_{\regA\regW})
    (U_{\mathsf B}^v)^\dagger
    (I_{\regB}\otimes\lvert g\rangle\langle g\rvert_{\regA}
          \otimes I_{\regW})
    U_{\mathsf B}^v
    (I_{\regB}\otimes\lvert0\rangle_{\regA\regW}).
\]
Define $U_{\mathsf C}^w$ on $\regC\regD\regZ$ in the same way for $Q_h^w$.
Running the predictors on a superposition of challenges means applying
the controlled unitaries
\[
  \mathbf U_{\mathsf B}
  :=\sum_v\lvert v\rangle\langle v\rvert_{\regR}
              \otimes U_{\mathsf B}^v,
  \qquad
  \mathbf U_{\mathsf C}
  :=\sum_w\lvert w\rangle\langle w\rvert_{\regS}
              \otimes U_{\mathsf C}^w.
\]
In particular, these unitaries leave the challenge labels unchanged.

Fix a nonzero multiplier $a\in\mathbb F_q^\ast$, known to both parties.
The following is a joint description of two entirely local procedures;
every sum over a challenge or Fourier-output label ranges over
$\mathbb F_q^\ell$.
\begin{enumerate}
  \item \emph{Prepare the challenges.}
  Starting with $\lvert0\rangle_{\regR}
  \lvert0\rangle_{\regS}\lvert\Psi_z^0\rangle$, each party
  applies $F_{q,\ell}$ to its own challenge register.  The state becomes
  \[
    q^{-\ell}\sum_{v,w}
    \lvert v\rangle_{\regR}\lvert w\rangle_{\regS}
    \otimes\lvert\Psi_z^0\rangle.
  \]
  Each local superposition has normalization $q^{-\ell/2}$, so the
  joint normalization is $q^{-\ell}$.

  \item \emph{Run the predictors coherently.}
  Apply $\mathbf U_{\mathsf B}$ and $\mathbf U_{\mathsf C}$, without
  measuring either answer register.  The resulting state is
  \[
    q^{-\ell}\sum_{v,w}
    \lvert v\rangle_{\regR}\lvert w\rangle_{\regS}
    \otimes(U_{\mathsf B}^v\otimes U_{\mathsf C}^w\otimes I_{\regE})
    \lvert\Psi_z^0\rangle.
  \]
  To display the answer registers explicitly, write
  \[
    (U_{\mathsf B}^v\otimes U_{\mathsf C}^w\otimes I_{\regE})
      \lvert\Psi_z^0\rangle
    =\sum_{g,h\in\mathbb F_q}
      \lvert g\rangle_{\regA}\lvert h\rangle_{\regD}
      \otimes\lvert\xi_{z;v,w}^{g,h}\rangle_{\regB\regW\regC\regZ\regE}.
  \]
  These residual vectors need not be normalized or product states;
  their squared norms are
  $\|\xi_{z;v,w}^{g,h}\|^2
  =\operatorname{Tr}[(P_g^v\otimes Q_h^w)\rho_{\regB\regC}^{x,y}]$.

  \item \emph{Turn the answers into phases.}
  Bob applies the diagonal unitary
  \[
    D_a^{\mathsf B}
    :=I_{\regB}\otimes\left(\sum_{g\in\mathbb F_q}
       \chi(ag)\lvert g\rangle\langle g\rvert_{\regA}\right)
       \otimes I_{\regW},
  \]
  and Charlie applies its counterpart $D_a^{\mathsf C}$ on
  $\regC\regD\regZ$.  These gates change the phases of the answer
  registers $\regA,\regD$, respectively, and act as the identity on
  the original shares and workspaces.  Thus a branch with answers $(g,h)$
  acquires the phase $\chi(ag)\chi(ah)=\chi(a(g+h))$.  The state is now
  \begin{align*}
    q^{-\ell}\sum_{v,w}\sum_{g,h\in\mathbb F_q}
    &\chi(a(g+h))
      \lvert v\rangle_{\regR}\lvert w\rangle_{\regS}
      \\
    &\quad\otimes\lvert g\rangle_{\regA}
      \lvert h\rangle_{\regD}
      \otimes\lvert\xi_{z;v,w}^{g,h}\rangle_{\regB\regW\regC\regZ\regE}.
  \end{align*}

  \item \emph{Run the predictors backward.}
  Apply $\mathbf U_{\mathsf B}^\dagger$ and
  $\mathbf U_{\mathsf C}^\dagger$.  With
  \[
    K_{v,a}^{\mathsf B}:=(U_{\mathsf B}^v)^\dagger
          D_a^{\mathsf B}U_{\mathsf B}^v,
    \qquad
    K_{w,a}^{\mathsf C}:=(U_{\mathsf C}^w)^\dagger
          D_a^{\mathsf C}U_{\mathsf C}^w,
  \]
  the state becomes
  \[
    q^{-\ell}\sum_{v,w}
    \lvert v\rangle_{\regR}\lvert w\rangle_{\regS}
    \otimes(K_{v,a}^{\mathsf B}\otimes K_{w,a}^{\mathsf C}\otimes I_{\regE})
    \lvert\Psi_z^0\rangle.
  \]
  Because a phase operation was inserted, this backward computation
  need not return the answer and workspace registers to zero, nor restore
  the original state on $\regB\regC$.

  \item \emph{Decode the Fourier labels.}
  Apply $F_{q,\ell}^\dagger$ to $\regR$ and $\regS$.
  Expanding the two transforms gives the full state
  \begin{align*}
    q^{-2\ell}\sum_{b,c,v,w}
    &\chi\!\left(-\langle b,v\rangle-\langle c,w\rangle\right)
      \lvert b\rangle_{\regR}\lvert c\rangle_{\regS}
      \\
    &\quad\otimes
      (K_{v,a}^{\mathsf B}\otimes K_{w,a}^{\mathsf C}\otimes I_{\regE})
      \lvert\Psi_z^0\rangle.
  \end{align*}
  The additional factor $q^{-\ell}$ comes from the two inverse
  transforms; their phases have a minus sign because they are inverses.

  \item \emph{Rescale and output.}
  On each challenge register apply the permutation
  $\lvert b\rangle\mapsto\lvert a^{-1}b\rangle$, with multiplication
  performed coordinatewise.  Measure only the challenge registers in
  the computational basis and output the resulting vectors
  $\widehat X^{\mathsf B}$ and $\widehat X^{\mathsf C}$.
  The other local registers may be discarded.  Success means that the
  Fourier labels before rescaling were $b=ax$ and $c=ay$.
\end{enumerate}

\emph{Analyzing the extraction probability.}
To determine when the procedure succeeds, return to the full state in
step~5.  The rescaling in step~6 produces the correct outputs $(x,y)$
exactly when the Fourier labels are $b=ax$ and $c=ay$.  Collecting the
terms with these labels, the vector accompanying
$\lvert ax\rangle_{\regR}\lvert ay\rangle_{\regS}$ is
\begin{align*}
  \lvert\Upsilon_{z,a}\rangle
  :=\mathbb E_{v,w}
  &\chi\!\left(-a\langle x,v\rangle-a\langle y,w\rangle\right)
  \\
  &\cdot(K^{\mathsf B}_{v,a}\otimes K^{\mathsf C}_{w,a}\otimes I_{\regE})
  \lvert\Psi_z^0\rangle.
\end{align*}
Here $q^{-2\ell}\sum_{v,w}=\mathbb E_{v,w}$.  This vector retains all
the registers $\regB\regA\regW\regC\regD\regZ\regE$ and is not normalized.
Rescaling changes only the output labels, so Born's rule gives
\[
  \Pr[\widehat X^{\mathsf B}=x,\widehat X^{\mathsf C}=y\mid z,a]
  =\|\Upsilon_{z,a}\|^2.
\]
Our task is therefore to lower-bound this squared norm.

The success probability counts every possible final state of the answer
and workspace registers.  For a lower bound, it suffices to count just
the successful outcomes in which all four registers
$\regA,\regW,\regD,\regZ$ are zero.  We choose this part because, as the next
calculation shows, it can be expressed directly using the prediction
POVMs that appear in the hypothesis.

First consider one term with fixed challenges $v,w$.  On Bob's side,
the registers $\regA,\regW$ start at zero, and the operator applied
in step~4 is
$K^{\mathsf B}_{v,a}=(U_{\mathsf B}^v)^\dagger
D_a^{\mathsf B}U_{\mathsf B}^v$.  For an input
$\lvert\varphi\rangle_{\regB}\lvert0\rangle_{\regA\regW}$, the $\regB$-vector
accompanying output ancillas $\lvert0\rangle_{\regA\regW}$ is
$T^{\mathsf B}_{v,a}\lvert\varphi\rangle_{\regB}$, where
\begin{align}
  T^{\mathsf B}_{v,a}
  &=(I_{\regB}\otimes\langle0\rvert_{\regA\regW})
    K_{v,a}^{\mathsf B}(I_{\regB}\otimes\lvert0\rangle_{\regA\regW})
    =\sum_{g\in\mathbb F_q}\chi(ag)P^v_g,
    \nonumber\\
  T^{\mathsf C}_{w,a}
  &=(I_{\regC}\otimes\langle0\rvert_{\regD\regZ})
    K_{w,a}^{\mathsf C}(I_{\regC}\otimes\lvert0\rangle_{\regD\regZ})
    =\sum_{h\in\mathbb F_q}\chi(ah)Q^w_h.
  \label{eq:qary-phase-operators}
\end{align}
The second line is Charlie's identical calculation.  To obtain the sums
on the right, expand the answer phase $D_a^{\mathsf B}$ into its
answer-register projectors, with coefficients $\chi(ag)$.  The unitary
and zero-ancilla factors around each projector give precisely $P_g^v$
by the earlier displayed POVM identity; similarly, Charlie obtains
$Q_h^w$.  Thus these $T$ operators arise from selecting the desired
ancilla output in the actual computation, and act on $\regB$ and $\regC$,
respectively.

Applying this calculation on both sides of each term in
$\lvert\Upsilon_{z,a}\rangle$ gives its part with all four added answer and workspace registers zero.
After leaving out the explicit zero-ancilla kets, the resulting vector
on $\regB\regC\regE$ is
\begin{align}
  \lvert\Theta_{z,a}\rangle
  :=
  \mathbb E_{v,w\leftarrow\mathbb F_q^\ell}
  &\chi\!\left(-a\langle x,v\rangle-a\langle y,w\rangle\right)
  \nonumber\\[-1mm]
  &\cdot
  \bigl(T^{\mathsf B}_{v,a}\otimes
        T^{\mathsf C}_{w,a}\otimes I_{\regE}\bigr)
  \lvert\psi_z\rangle.
  \label{eq:qary-target-vector}
\end{align}
In other words, the complete successful-output vector decomposes as
\[
  \lvert\Upsilon_{z,a}\rangle
  =\lvert\Theta_{z,a}\rangle_{\regB\regC\regE}
    \otimes\lvert0\rangle_{\regA\regW}
    \otimes\lvert0\rangle_{\regD\regZ}
    +\lvert\Upsilon_{z,a}^{\perp}\rangle,
\]
where the joint ancillas of $\lvert\Upsilon_{z,a}^{\perp}\rangle$ are
orthogonal to their all-zero state.  The two parts are orthogonal, so
\[
  \Pr[\widehat X^{\mathsf B}=x,\widehat X^{\mathsf C}=y\mid z,a]
  =\|\Upsilon_{z,a}\|^2
  =\|\Theta_{z,a}\|^2+\|\Upsilon_{z,a}^{\perp}\|^2
  \geq\|\Theta_{z,a}\|^2.
\]
This is the desired lower bound: successful outcomes with nonzero
ancillas can only increase the probability.  Splitting the vector in
this way is an analysis of the existing procedure, not an additional
measurement or postselection step performed by the parties.

For fixed $v,w,z$, write
\begin{equation}
  p_z(g,h\mid v,w)
  :=
  \operatorname{Tr}\!\left[
    (P^v_g\otimes Q^w_h)\rho^{x,y}_{\regB\regC}
  \right].
  \label{eq:qary-conditional-output-distribution}
\end{equation}
Taking the overlap of~\eqref{eq:qary-target-vector} with
$\lvert\psi_z\rangle$ gives
\begin{align}
  \Phi_z(a)
  &:={}
  \langle\psi_z\mid\Theta_{z,a}\rangle
  \nonumber\\
  &={}
  \mathbb E_{v,w}
  \sum_{g,h\in\mathbb F_q}
  p_z(g,h\mid v,w)
  \chi\!\left(
    a\bigl(
      g-\langle x,v\rangle
      +h-\langle y,w\rangle
    \bigr)
  \right).
  \label{eq:qary-return-amplitude}
\end{align}
By Cauchy--Schwarz, $\|\Theta_{z,a}\|^2\geq |\Phi_z(a)|^2$.  Hence
\begin{equation}
  \Pr[\widehat X^{\mathsf B}=x,\widehat X^{\mathsf C}=y\mid z,a]
  \geq |\Phi_z(a)|^2.
  \label{eq:qary-conditional-search-from-amplitude}
\end{equation}
Notice that this is only a lower bound: other orthogonal workspace branches
may also lead to the correct two outputs.

\emph{Connection to Adcock--Cleve and intuition.}
The local procedure is the finite-field analogue of the single-trial
subroutine in Adcock--Cleve's quantum Goldreich--Levin theorem
\cite[proof of Theorem~6, Figure~1]{AC02}; its formulation with an
auxiliary quantum input appears in the full version of
\cite[Lemma~B.12]{CLLZ21}.  In the binary case, a perfect parity predictor
inserts the phases $(-1)^{\langle x,v\rangle}$.  These encode the whole
string $x$ in the Fourier basis:
\[
  2^{-\ell/2}\sum_v(-1)^{\langle x,v\rangle}\lvert v\rangle
  =H^{\otimes\ell}\lvert x\rangle.
\]
Applying the Hadamards again therefore returns $x$.  For an imperfect
predictor whose success on a fixed $x$ is $1/2+\beta$, the decoded $x$
component has overlap with the initial purified input and zero ancillas
equal to
\[
  \mathbb E_{v,g}\!\left[(-1)^{g-\langle x,v\rangle}\right]
  =\Pr[g=\langle x,v\rangle]-\Pr[g\ne\langle x,v\rangle]
  =2\beta.
\]
Here $g$ follows the predictor's output distribution on $v$.  Correct
answers contribute $+1$ and incorrect answers $-1$: prediction bias
becomes a coherent amplitude, and squaring yields recovery probability
at least $4\beta^2$.  Returning the unknown input state is only an
overlap used for this bound, not a test performed by the extractor.
We use this single trial, not Adcock--Cleve's further amplitude
amplification, which uses an equality oracle.

In our setting, both parties perform this operation locally with the
same multiplier $a$.  Their shares may be entangled, so we do not
multiply two one-party success bounds.  Instead, the joint overlap
in~\eqref{eq:qary-return-amplitude} is
$\mathbb E[\chi(a(E^{\mathsf B}+E^{\mathsf C}))\mid z]$.
The following step averages over the common multiplier to find useful
Fourier energy over an arbitrary finite field.

\paragraph{Step 3: Averaging over $a$.}
Let $A_z$ be the distribution of the total error, averaged over the two
independent challenges:
\begin{equation}
  A_z(s)
  :=
  \mathbb E_{v,w}
  \sum_{\substack{g,h\in\mathbb F_q:\\
        g-\langle x,v\rangle+h-\langle y,w\rangle=s}}
  p_z(g,h\mid v,w).
  \label{eq:qary-total-error-distribution}
\end{equation}
Then $A_z(0)=p_z$, $\sum_sA_z(s)=1$, and
\begin{equation}
  \Phi_z(a)=\sum_{s\in\mathbb F_q}A_z(s)\chi(as).
  \label{eq:qary-error-fourier-transform}
\end{equation}
Using~\eqref{eq:qary-character-orthogonality} and the fact that
$\Phi_z(0)=1$, we obtain the exact identity
\begin{equation}
  \mathbb E_{a\leftarrow\mathbb F_q^\ast}|\Phi_z(a)|^2
  =
  \frac{q\sum_{s\in\mathbb F_q}A_z(s)^2-1}{q-1}.
  \label{eq:qary-nonzero-parseval}
\end{equation}

Suppose now that $z$ is good and put $\eta=\epsilon/2$.  For a prescribed
value of $A_z(0)$, the sum of the squares of the other $q-1$ probabilities is
minimized when those probabilities are equal.  Moreover, the resulting lower
bound is increasing for $A_z(0)\geq 1/q$.  Therefore
\begin{align}
  \sum_{s\in\mathbb F_q}A_z(s)^2
  &\geq
  \left(\frac1q+\eta\right)^2
  +(q-1)\left(\frac1q-\frac{\eta}{q-1}\right)^2
  \nonumber\\
  &=
  \frac1q+\frac{q}{q-1}\eta^2.
  \label{eq:qary-collision-lower-bound}
\end{align}
Substitution into~\eqref{eq:qary-nonzero-parseval} yields
\begin{equation}
  \mathbb E_{a\leftarrow\mathbb F_q^\ast}|\Phi_z(a)|^2
  \geq
  \frac{q^2}{(q-1)^2}\eta^2
  =
  \frac{q^2}{4(q-1)^2}\epsilon^2.
  \label{eq:qary-good-pair-fourier-energy}
\end{equation}

\paragraph{Step 4: One multiplier works globally.}
Let $p_{\mathrm{search}}(a)$ be the overall simultaneous search success of
the preceding two local procedures with multiplier $a$ hardcoded.  Averaging
first over the original hidden-pair distribution and then over $a$, and
discarding the nonnegative contribution of the bad pairs, gives
\begin{align}
  \mathbb E_{a\leftarrow\mathbb F_q^\ast}
       p_{\mathrm{search}}(a)
  &\geq
  \mu\cdot\frac{q^2}{4(q-1)^2}\epsilon^2
  \nonumber\\
  &\geq
  \frac{q^2}{8(q-1)^2}\epsilon^3.
  \label{eq:qary-global-search-average}
\end{align}
Thus there is one $a_\star\in\mathbb F_q^\ast$, independent of the hidden
pair, for which the same lower bound holds.  Hardcoding $a_\star$ gives valid
local full-string measurements and proves~\eqref{eq:qary-search-lower-bound}.
Rearranging this inequality gives~\eqref{eq:qary-contrapositive}.  Finally,
the exact-prediction event is contained in
$\{E^{\mathsf B}+E^{\mathsf C}=0\}$, which proves the last assertion.
\end{proof}

\begin{remark}[Relation to Kundu--Tan]
\label{rem:qary-kundu-tan}
\itshape
Kundu and Tan prove the corresponding independent-challenge statement for
the prime fields $\mathbb F_2$ and $\mathbb F_3$~\cite[Lemma~27]{KT25}. It is not clear if their proof extends to larger fields.
\end{remark}

\subsection{Identical Challenge Setting over Arbitrary Finite Fields}
\label{sec:qary-common-finite}

The quantitative result below applies to every finite field.  It combines the
independent-challenge theorem proved in Section~\ref{sec:independent-qary}
with the correlated-to-independent inequality of Ananth, Kaleoglu, and
Liu~\cite[Lemma~18]{AKL23}.  Its growing-field consequence, where the
reciprocal field size is negligible, is stated separately afterward.

\begin{theorem}[Identical challenge setting over arbitrary finite fields]
\label{thm:qary-common-finite}
For every instance of the common-challenge experiment in
Definition~\ref{def:qary-common-model},
\begin{equation}
 p_{\mathrm{search}}^\star
 \geq
 \frac{q^2}{8(q-1)^2}
 \left(
   \frac{p_{\mathrm{com}}^3}{216}-\frac1q
 \right)_{+}^{\!3}.
 \label{eq:qary-growing-search-lower}
\end{equation}
Equivalently,
\begin{equation}
 p_{\mathrm{com}}
 \leq
 6\left(
   \frac1q
   +
   2\left(\frac{q-1}{q}\right)^{2/3}
     \bigl(p_{\mathrm{search}}^\star\bigr)^{1/3}
 \right)^{1/3}.
 \label{eq:qary-growing-common-upper}
\end{equation}
\end{theorem}

\begin{proof}
Let $p_{\mathrm{ind}}$ be the exact success probability of the same local
response algorithms when Bob and Charlie instead receive independent
uniform challenges.  For a fixed hidden string $x$, let $c_x$ and $i_x$
denote the corresponding common- and independent-challenge success
probabilities.  The correlated-to-independent inequality of Ananth,
Kaleoglu, and Liu, applied to $\rho_x$ with the correct-answer effects
$B_r:=P^r_{\langle r,x\rangle}$ and
$C_r:=Q^r_{\langle r,x\rangle}$, gives $c_x\leq 6i_x^{1/3}$ for every
$x$.  This is a pointwise analytical application; the operational predictor
remains the original $x$-independent family of $q$-outcome measurements.
Averaging over $x$ and using concavity of $u\mapsto u^{1/3}$ gives
\begin{equation}
 p_{\mathrm{com}}
 \leq 6p_{\mathrm{ind}}^{1/3},
 \qquad\text{and hence}\qquad
 p_{\mathrm{ind}}\geq\frac{p_{\mathrm{com}}^3}{216}.
 \label{eq:qary-correlated-independent}
\end{equation}
In the independent-challenge experiment, exact correctness on both sides
implies that the two errors sum to zero.  Therefore
Theorem~\ref{thm:independent-qary}, specialized to
$X^{\mathsf B}=X^{\mathsf C}=x$, gives
\begin{equation}
 p_{\mathrm{ind}}
 \leq \frac1q
 +
 2\left(\frac{q-1}{q}\right)^{2/3}
   \bigl(p_{\mathrm{search}}^\star\bigr)^{1/3}.
 \label{eq:qary-independent-exact-upper}
\end{equation}
Combining~\eqref{eq:qary-correlated-independent} and
\eqref{eq:qary-independent-exact-upper} proves
\eqref{eq:qary-growing-common-upper}; rearranging proves
\eqref{eq:qary-growing-search-lower}.
\end{proof}

\begin{corollary}[Efficient consequence for growing fields]
\label{cor:qary-growing-common}
Suppose that $q=q(\lambda)$ has an efficient field representation, that
$\log q(\lambda)$ and $\ell(\lambda)$ are polynomially bounded, that
$1/q(\lambda)$ is negligible. Consequently,
computational joint search security implies
\begin{equation}
 p_{\mathrm{com}}
 \leq \frac1{q(\lambda)}+\operatorname{negl}(\lambda).
 \label{eq:qary-growing-asymptotic}
\end{equation}
This applies, for example, to the infinite family
$q(\lambda)=2^\lambda$.
\end{corollary}

\begin{proof}
Use the original predictor circuits with
independently prepared local challenge registers and then run the coherent
extractor from Theorem~\ref{thm:independent-qary}, sampling its common
multiplier as shared classical randomness before the parties separate.  If
$p_{\mathrm{com}}$ is nonnegligible while $1/q$ is negligible, then
\eqref{eq:qary-correlated-independent} shows that the independent exact
success exceeds $1/q$ by a nonnegligible amount. Taking the contrapositive proves
\eqref{eq:qary-growing-asymptotic}.
\end{proof}

\paragraph{Acknowledgements}
Prabhanjan Ananth is supported by the National Science Foundation under the grants FET-2329938, CAREER-2341004, and FET-2530160.

\section*{Glossary of Notation}

The table below records the notation used throughout the statement and proof.
Operators on $\regB\regC$ are understood to act as the identity on the
purification register $\regE$ whenever that register is present.

\begingroup
\footnotesize
\renewcommand{\arraystretch}{0.93}
\setlength{\LTpre}{0pt}
\begin{longtable}{>{\raggedright\arraybackslash}p{0.23\textwidth}
                  >{\raggedright\arraybackslash}p{0.70\textwidth}}
\textbf{Notation} & \textbf{Meaning} \\
\hline
\endfirsthead
\textbf{Notation} & \textbf{Meaning} \\
\hline
\endhead

\multicolumn{2}{l}{\textbf{Basic notation and probability}}\\[0.2em]
$[n]$ & The set $\{1,\ldots,n\}$.\\
$\ftwo^n$ & The $n$-dimensional vector space of binary strings.\\
$\langle r,x\rangle$ & The mod-$2$ inner product of $r,x\in\ftwo^n$.\\
$\chi_x(r)$ & The character $(-1)^{\langle r,x\rangle}$.\\
$s\xleftarrow{\$}S$ & A uniformly random sample from the finite set $S$.\\
$\expect_x$ & Expectation over a uniformly random $x$; multiple subscripts denote independent uniform choices unless stated otherwise.\\
$\linear(\regR),\density(\regR)$ & Linear operators and density operators on register $\regR$.\\
$I_{\regR},\Tr$ & Identity on $\regR$ and the trace.\\
$M\preceq N$ & Loewner order: $N-M$ is positive semidefinite.\\
$\norm{M}_\infty$ & Operator norm of $M$.\\
$\operatorname{ran}(P)$ & Range of the operator or projection $P$.\\[0.4em]

\multicolumn{2}{l}{\textbf{Registers, states, and experiments}}\\[0.2em]
$\regX$ & Classical input register containing the basis state indexed by $x$.\\
$\regB,\regC$ & Quantum share registers held by Bob and Charlie, respectively.\\
$\regE$ & Auxiliary register used to purify $\rho_x$.\\
$\Phi$ & Channel mapping the classical input encoding to the bipartite state on $\regB\regC$.\\
$\rho_x$ & The bipartite state indexed by $x$; in the introductory channel formulation, $\rho_x=\Phi(\ketbra{x}{x})$.\\
$\boldsymbol{\rho}$ & The uniformly indexed ensemble $(\rho_x)_{x\in\ftwo^n}$.\\
$\ket{\psi^x}$ & A fixed purification of $\rho_x$ on $\regB\regC\regE$.\\
$\Lambda_{\bob}^{\srch},\Lambda_{\charlie}^{\srch}$ & Local full-string POVMs with effects $\{\extractor_{\bob}^z\}_z$ and $\{\extractor_{\charlie}^z\}_z$.\\
$p_{\srch}(\boldsymbol{\rho})$ & Optimal probability that both local full-string POVMs output the correct uniformly sampled ensemble index $x$; it lies in $[2^{-n},1]$.  Inside the proof this quantity is abbreviated as $p_{\srch}$.\\
$\Lambda_{\bob}^{\pred}[r],\Lambda_{\charlie}^{\pred}[r]$ & Local binary POVMs used when both parties receive challenge $r$.\\
$p_{\pred}$ & Probability that both parties output $\langle r,x\rangle$ in the same prediction experiment.\\[0.4em]

\multicolumn{2}{l}{\textbf{Prediction and Fourier operators}}\\[0.2em]
$\Delta_{\bob}^r,\Delta_{\charlie}^r$ & Difference observables $\predpovm^{r,0}-\predpovm^{r,1}$ for Bob and Charlie.\\
$\bfX^x,\bfY^x$ & Fourier coefficients $\expect_r[\chi_x(r)\Delta_{\bob}^r]$ and $\expect_r[\chi_x(r)\Delta_{\charlie}^r]$.\\
$\bfZ$ & Common-challenge correlation operator $\expect_r[\Delta_{\bob}^r\otimes\Delta_{\charlie}^r]$.\\
$\bfG_x$ & Effect for the event that both parties correctly answer the parity challenge for hidden string $x$.\\
$\eta,\tau,\delta$ & Proof parameters, with $\tau=\eta^3$ and $\delta=\eta^2/3072$.\\[0.4em]

\multicolumn{2}{l}{\textbf{High/low decomposition}}\\[0.2em]
$\Pi_x,N_x$ & Bob's high and low spectral projections for $\bfX^x$: $\Pi_x$ projects onto eigenvectors whose eigenvalues $\lambda$ satisfy $|\lambda|\ge\tau$, and $N_x=I-\Pi_x$.\\
$\Theta_x,O_x$ & Charlie's high and low spectral projections for $\bfY^x$: $\Theta_x$ projects onto eigenvectors whose eigenvalues $\mu$ satisfy $|\mu|\ge\tau$, and $O_x=I-\Theta_x$.\\
$P_U^x$ & Joint projection for $U\in\{\HH,\HL,\LH,\LL\}$; the first letter refers to Bob and the second to Charlie.\\
$\ket{\psi_U^x}$ & Component $P_U^x\ket{\psi^x}$ of the purified state.\\
$\operatorname{wt}(U)$ & Average component weight $\expect_x\norm{\ket{\psi_U^x}}^2$.\\
$\kappa_U$ & Indicator that the average weight of component family $U$ is at least $\delta$.\\
$\ket{\phi_U^x},\ket{\phi^x}$ & Retained component $\kappa_U\ket{\psi_U^x}$ and their sum over $U\in\{\HL,\LH,\LL\}$.\\
$\ket{e^x},q_{\mathrm{omit}}$ & Omitted vector $\ket{\psi^x}-\ket{\phi^x}$ and its average squared norm.\\
$W_U$ & Retained average weight $\expect_x\norm{\ket{\phi_U^x}}^2$.\\[0.4em]

\multicolumn{2}{l}{\textbf{Effective cross term and extractors}}\\[0.2em]
$M_x$ & Compression $P_{\LL}^x((1/2+\eta)I-\bfG_x)P_{\LL}^x$, viewed as an operator on $\operatorname{ran}(P_{\LL}^x)$.\\
$\ket{u_x}$ & Proof-local vector $P_{\LL}^x\bfG_x(\ket{\phi_{\HL}^x}+\ket{\phi_{\LH}^x})$ used in completing the square.\\
$c_x$ & Effective $\HL$--$\LH$ cross scalar after minimizing over the $\LL$ component.\\
$Q_x,R_x,A_x$ & Proof-local compressed operators defined in~\eqref{eq:td-claim25-QRA}; they satisfy $4M_x=A_x-R_x$.\\
$\widetilde c_x$ & Approximation to $c_x$ obtained by replacing $M_x^{-1}$ with the geometric inverse.\\
$\Gamma_m$ & Averaged contribution containing $m$ copies of the common-challenge operator $\bfZ$.\\
$w=(r_1,\ldots,r_m)$ & Auxiliary word of independent uniform challenge vectors; after fixing $w$, it is hardwired into the extraction measurements.\\
$V_{x,w}^{B},V_{x,w}^{C}$ & Fixed-word local products whose normalized adjoint-products define extraction effects.\\[0.4em]

\multicolumn{2}{l}{\textbf{Arbitrary-field extensions}}\\[0.2em]
$\lambda$ & Security parameter for asymptotic families; the field size $q=q(\lambda)$ may depend on it.\\
$q=p^t,\mathbb F_q$ & A prime power and the finite field of order $q$; all arbitrary-field inner products are evaluated in $\mathbb F_q$.\\
$\chi$ & A fixed nontrivial additive character of $\mathbb F_q$, defined using the field trace to $\mathbb F_p$.\\
$X^{\mathsf B},X^{\mathsf C}$ & Possibly correlated hidden strings in the independent-challenge experiment of Section~\ref{sec:independent-qary}.\\
$\regX,\regY$ & Classical registers recording these two hidden strings in Section~\ref{sec:independent-qary}; not given to the extractors.\\
$\regR,\regS$ & Bob's and Charlie's challenge registers, respectively; also hold the final extraction outputs.\\
$\regA,\regD$ & Bob's and Charlie's predictor-answer registers, respectively.\\
$\regW,\regZ$ & Bob's and Charlie's predictor-workspace registers, respectively.\\
$V^{\mathsf B},V^{\mathsf C}$ & Independent uniform challenge vectors given locally to Bob and Charlie in that experiment.\\
$E^{\mathsf B},E^{\mathsf C}$ & The corresponding additive prediction errors.\\
$p_{\mathrm{search}}^\star$ & Optimal joint full-string recovery probability in the arbitrary-field experiments.\\
$p_{\mathrm{com}},p_{\mathrm{ind}}$ & Exact joint prediction success with a common challenge and with independent local challenges, respectively, in Theorem~\ref{thm:qary-common-finite}.\\

\hline
\end{longtable}
\endgroup

\begin{mdframed}[
 nobreak=true,
 linewidth=0.7pt,
 linecolor=black!55,
 backgroundcolor=black!2,
 roundcorner=3pt,
 skipabove=1em,
 skipbelow=1em,
 innerleftmargin=10pt,
 innerrightmargin=10pt,
 innertopmargin=8pt,
 innerbottommargin=8pt
]
\noindent
{\bf AI Disclosure}: We used ChatGPT using GPT-5.6 Sol Pro
and Codex at ultra reasoning effort to assist with development of all the proofs as well as for the preparation of the paper. 
\par The first 10 pages were mainly prepared by the authors with minor assistance from the above tools, especially in the generation of the figures. All the technical ideas in the paper were suggested by the above tools and were then iteratively polished.  
\par The authors verified the correctness and originality of all content including references.
\end{mdframed}


\begin{thebibliography}{CLX26b}
\raggedright

\bibitem[GL89]{GL89}
O.~Goldreich and L.~A. Levin.
\newblock A hard-core predicate for all one-way functions.
\newblock In \emph{Proceedings of the 21st Annual ACM Symposium on Theory of
Computing}, pages 25--32, 1989.
\newblock \href{https://doi.org/10.1145/73007.73010}{doi:10.1145/73007.73010}.

\bibitem[Reg05]{Reg05}
O.~Regev.
\newblock On lattices, learning with errors, random linear codes, and cryptography.
\newblock In \emph{Proceedings of the 37th Annual ACM Symposium on Theory of
Computing}, pages 84--93, 2005.
\newblock \href{https://doi.org/10.1145/1060590.1060603}{doi:10.1145/1060590.1060603}.

\bibitem[AC02]{AC02}
M.~Adcock and R.~Cleve.
\newblock A quantum Goldreich--Levin theorem with cryptographic applications.
\newblock In \emph{STACS 2002}, volume 2285 of \emph{LNCS}, pages 323--334,
2002.
\newblock \href{https://doi.org/10.1007/3-540-45841-7_26}{doi:10.1007/3-540-45841-7\_26}.

\bibitem[HJ12]{HJ12}
R.~A. Horn and C.~R. Johnson.
\newblock \emph{Matrix Analysis}.
\newblock Second edition, Cambridge University Press, 2012.
\newblock \href{https://doi.org/10.1017/CBO9781139020411}{doi:10.1017/CBO9781139020411}.

\bibitem[Tes14]{Teschl2014}
G.~Teschl.
\newblock \emph{Mathematical Methods in Quantum Mechanics: With
Applications to Schrödinger Operators}.
\newblock Second edition, volume 157 of \emph{Graduate Studies in
Mathematics}, American Mathematical Society, 2014.
\newblock \href{https://doi.org/10.1090/gsm/157}{doi:10.1090/gsm/157}.

\bibitem[Wat18]{Wat18}
J.~Watrous.
\newblock \emph{The Theory of Quantum Information}.
\newblock Cambridge University Press, 2018.
\newblock \href{https://doi.org/10.1017/9781316848142}{doi:10.1017/9781316848142}.


\bibitem[CLL+21]{CLLZ21}
A.~Coladangelo, J.~Liu, Q.~Liu, and M.~Zhandry.
\newblock Hidden cosets and applications to unclonable cryptography.
\newblock In \emph{Advances in Cryptology -- CRYPTO 2021, Part I}, volume
12825 of \emph{LNCS}, pages 556--584, 2021.
\newblock \href{https://doi.org/10.1007/978-3-030-84242-0_20}{doi:10.1007/978-3-030-84242-0\_20}.
\newblock Full version: \url{https://arxiv.org/abs/2107.05692}.

\bibitem[KT25]{KT25}
S.~Kundu and E.~Y.-Z. Tan.
\newblock Device-independent uncloneable encryption.
\newblock \emph{Quantum}, 9:1582, 2025.
\newblock \href{https://doi.org/10.22331/q-2025-01-08-1582}{doi:10.22331/q-2025-01-08-1582}.

\bibitem[AKL23]{AKL23}
P.~Ananth, F.~Kaleoglu, and Q.~Liu.
\newblock Cloning games: A general framework for unclonable primitives.
\newblock In \emph{Advances in Cryptology -- CRYPTO 2023, Part V}, pages
66--98, 2023.
\newblock \href{https://doi.org/10.1007/978-3-031-38554-4_3}{doi:10.1007/978-3-031-38554-4\_3}.

\bibitem[BKL24]{BKL24}
A.~Broadbent, M.~Karvonen, and S.~Lord.
\newblock Uncloneable quantum advice.
\newblock \emph{IACR Communications in Cryptology}, 1(3), 2024.
\newblock \href{https://doi.org/10.62056/abe0fhbmo}{doi:10.62056/abe0fhbmo}.

\bibitem[CG24a]{CG24}
A.~Coladangelo and S.~Gunn.
\newblock How to use quantum indistinguishability obfuscation.
\newblock In \emph{Proceedings of the 56th Annual ACM Symposium on Theory of
Computing}, pages 1003--1008, 2024.
\newblock \href{https://doi.org/10.1145/3618260.3649779}{doi:10.1145/3618260.3649779}.

\bibitem[AB24]{AB24}
P.~Ananth and A.~Behera.
\newblock A modular approach to unclonable cryptography.
\newblock In \emph{Advances in Cryptology -- CRYPTO 2024, Part VII}, volume 14926 of
\emph{LNCS}, pages 3--37, 2024.
\newblock \href{https://doi.org/10.1007/978-3-031-68394-7_1}{doi:10.1007/978-3-031-68394-7\_1}.

\bibitem[AKY25]{AKY25}
P.~Ananth, F.~Kaleoglu, and H.~Yuen.
\newblock Simultaneous Haar indistinguishability with applications to
unclonable cryptography.
\newblock In \emph{ITCS 2025}, volume 325 of \emph{LIPIcs}, article 7, 2025.
\newblock \href{https://doi.org/10.4230/LIPIcs.ITCS.2025.7}{doi:10.4230/LIPIcs.ITCS.2025.7}.


\bibitem[CHV23]{CHV23}
C.~Chevalier, P.~Hermouet, and Q.-H.~Vu.
\newblock Towards unclonable cryptography in the plain model.
\newblock Cryptology ePrint Archive, Paper 2023/1825, 2023.
\newblock \href{https://eprint.iacr.org/2023/1825}{ePrint 2023/1825}.

\bibitem[BCG26]{BCG26}
A.~Behera, A.~\c{C}akan, and V.~Goyal.
\newblock Copy-protection with correlated challenges: Point functions and more via decisional coset monogamy.
\newblock arXiv:2608.18841, 2026.
\newblock \url{https://arxiv.org/abs/2608.18841}.

\bibitem[CG24b]{CakanGoyal24}
A.~\c{C}akan and V.~Goyal.
\newblock Unbounded leakage-resilient encryption and signatures.
\newblock Cryptology ePrint Archive, Paper 2024/1876, 2024.
\newblock \href{https://eprint.iacr.org/2024/1876}{ePrint 2024/1876}.

\bibitem[CKN+25]{CKNY25}
J.~Champion, F.~Kitagawa, R.~Nishimaki, and T.~Yamakawa.
\newblock Untelegraphable encryption and its applications.
\newblock In \emph{Theory of Cryptography -- TCC 2025, Part III}, volume
16270 of \emph{LNCS}, pages 3--35, 2025.
\newblock \href{https://doi.org/10.1007/978-3-032-12296-4_1}{doi:10.1007/978-3-032-12296-4\_1}.

\bibitem[CLX26a]{CLX26}
A.~Coladangelo, Q.~Liu, and Z.~Xie.
\newblock The curious case of ``XOR repetition'' of monogamy-of-entanglement
games.
\newblock In \emph{ITCS 2026}, volume 362 of \emph{LIPIcs}, article 41, 2026.
\newblock \href{https://doi.org/10.4230/LIPIcs.ITCS.2026.41}{doi:10.4230/LIPIcs.ITCS.2026.41}.

\bibitem[CLX26b]{CLX26Common}
A.~Coladangelo, Q.~Liu, and Z.~Xie.
\newblock Unclonable encryption from {BB84} states: A simultaneous
Goldreich--Levin reduction.
\newblock arXiv:2608.17629, 2026.
\newblock \url{https://arxiv.org/abs/2608.17629}.

\bibitem[HMS04]{HMS04}
T.~Holenstein, U.~Maurer, and J.~Sj\"odin.
\newblock Complete classification of bilinear hard-core functions.
\newblock In \emph{Advances in Cryptology -- CRYPTO 2004}, volume 3152 of
\emph{LNCS}, pages 73--91, 2004.
\newblock \href{https://doi.org/10.1007/978-3-540-28628-8_5}{doi:10.1007/978-3-540-28628-8\_5}.

\bibitem[APV23]{APV23}
P.~Ananth, A.~Poremba, and V.~Vaikuntanathan.
\newblock Revocable cryptography from learning with errors.
\newblock In \emph{Theory of Cryptography -- TCC 2023}, volume 14372 of
\emph{LNCS}, pages 93--122, 2023.
\newblock \href{https://doi.org/10.1007/978-3-031-48624-1_4}{doi:10.1007/978-3-031-48624-1\_4}.

\bibitem[STH+23]{STHY23}
K.~Sudo, M.~Tezuka, K.~Hara, and Y.~Yoshida.
\newblock Quantum search-to-decision reduction for the {LWE} problem.
\newblock In \emph{Progress in Cryptology -- AFRICACRYPT 2023}, volume
14064 of \emph{LNCS}, pages 395--413, 2023.
\newblock \href{https://doi.org/10.1007/978-3-031-37679-5_17}{doi:10.1007/978-3-031-37679-5\_17}.


\bibitem[BL20]{BL20}
A.~Broadbent and S.~Lord.
\newblock Uncloneable quantum encryption via oracles.
\newblock In \emph{15th Conference on the Theory of Quantum Computation,
Communication and Cryptography (TQC 2020)}, volume 158 of \emph{LIPIcs},
pages 4:1--4:22, 2020.
\newblock \href{https://doi.org/10.4230/LIPIcs.TQC.2020.4}{doi:10.4230/LIPIcs.TQC.2020.4}.

\bibitem[AK21]{AK21}
P.~Ananth and F.~Kaleoglu.
\newblock Unclonable encryption, revisited.
\newblock In \emph{Theory of Cryptography -- TCC 2021, Part I}, volume 13042
of \emph{LNCS}, pages 299--329, 2021.
\newblock \href{https://doi.org/10.1007/978-3-030-90459-3_11}{doi:10.1007/978-3-030-90459-3\_11}.

\bibitem[AKL+22]{AKLLZ22}
P.~Ananth, F.~Kaleoglu, X.~Li, Q.~Liu, and M.~Zhandry.
\newblock On the feasibility of unclonable encryption, and more.
\newblock In \emph{Advances in Cryptology -- CRYPTO 2022, Part II}, volume
13508 of \emph{LNCS}, pages 212--241, 2022.
\newblock \href{https://doi.org/10.1007/978-3-031-15979-4_8}{doi:10.1007/978-3-031-15979-4\_8}.

\bibitem[GZ20]{GZ20}
M.~Georgiou and M.~Zhandry.
\newblock Unclonable decryption keys.
\newblock Cryptology ePrint Archive, Paper 2020/877, 2020.
\newblock \href{https://eprint.iacr.org/2020/877}{ePrint 2020/877}.

\bibitem[AS26]{AS26}
P.~Ananth and A.~Sahai.
\newblock Unconditional unclonable encryption.
\newblock arXiv:2607.21551, 2026.
\newblock \url{https://arxiv.org/abs/2607.21551}.

\bibitem[Rag26]{Rag26}
S.~Ragavan.
\newblock Efficient unclonable encryption from Pauli eigenstates.
\newblock arXiv:2607.21811, 2026.
\newblock \url{https://arxiv.org/abs/2607.21811}.

\bibitem[BBC26]{BBC26}
A.~Bhattacharyya, A.~Broadbent, and E.~Culf.
\newblock The uncloneable bit exists.
\newblock arXiv:2603.08916v2, 2026.
\newblock \url{https://arxiv.org/abs/2603.08916}.

\bibitem[Con]{ConradGaussJacobi}
K.~Conrad.
\newblock Gauss and Jacobi sums on finite fields and $\mathbb Z/m\mathbb Z$.
\newblock Expository notes, Sections~1--2.
\newblock \url{https://kconrad.math.uconn.edu/blurbs/gradnumthy/Gauss-Jacobi-sums.pdf}.

\bibitem[BCW02]{BCW02}
J.~N.~de Beaudrap, R.~Cleve, and J.~Watrous.
\newblock Sharp quantum versus classical query complexity separations.
\newblock \emph{Algorithmica}, 34(4):449--461, 2002.
\newblock \url{https://doi.org/10.1007/s00453-002-0978-1}.

\bibitem[CW00]{CW00}
R.~Cleve and J.~Watrous.
\newblock Fast parallel circuits for the quantum Fourier transform.
\newblock In \emph{FOCS 2000}, pages 526--536, 2000.
\newblock \url{https://cs.uwaterloo.ca/~watrous/Papers/ParallelCircuitsQFT.pdf}.

\end{thebibliography}
\end{document}